\documentclass{acmsmall}

\usepackage{times}
\usepackage{amsmath,amssymb}
\usepackage{graphicx}
\usepackage{epsfig}
\usepackage{subfig}
\usepackage{fancyvrb}
\usepackage{color}
\usepackage{float}
\usepackage{dsfont}
\usepackage{amsfonts}
\usepackage{bm}
\usepackage[T1]{fontenc}
\usepackage{latexsym}
\usepackage{array}
\usepackage{clrscode}

\newcommand{\out}[1]{\mbox{{\small \sc OUT}\ensuremath{({#1})}}}
\newcommand{\droot}[1]{\mbox{{\small \sc ROOT}\ensuremath{({#1})}}}
\newcommand{\mb}[1]{\mbox{{\small \sc MB}\ensuremath{({#1})}}}
\newcommand{\nodes}[1]{\mbox{{\small \sc NODES}\ensuremath{({#1})}}}
\newcommand{\edges}[1]{\mbox{{\small \sc EDGES}\ensuremath{({#1})}}}
\newcommand{\mbn}[1]{\mbox{{\small \sc MBN}\ensuremath{({#1})}}}
\newcommand{\etiq}[1]{\mbox{{\small \sc LABEL}\ensuremath{({#1})}}}
\newcommand{\nf}[1]{\mbox{{\it nf}\ensuremath{({#1})}}}
\newcommand{\dagp}[1]{\mbox{{\it dag}\ensuremath{({#1})}}}
\newcommand{\pattern}[1]{\mbox{{\it pattern}\ensuremath{({#1})}}}
\newcommand{\collapse}[2]{\mbox{{\it collapse}\ensuremath{_{#1}({#2})}}}
\newcommand{\doc}[1]{\mbox{{\it doc}\ensuremath{({#1})}}}
\newcommand{\xpath}[1]{\mbox{{\it xpath}\ensuremath{({#1})}}}
\newcommand{\interleave}[1]{\mbox{{\it interleave}\ensuremath{({#1})}}}

\newcommand{\unfold}[1]{\mbox{{\it unfold}\ensuremath{({#1})}}}

\newcommand{\tp}[2]{\mbox{{\rm TP}\ensuremath{ _{#1}({#2})}}}

\newcommand{\sub}[2]{\mbox{{\rm SP}\ensuremath{ _{#1}({#2})}}}

\newcommand{\xp}{\textit{XP }}

\newcommand{\xpp}{\textit{XP}}

\newcommand{\xppes}{\ensuremath{\xpp_{es}}}
\newcommand{\xppdesc}{\ensuremath{\xpp_{\snippet{//}}}}
\newcommand{\xppescap}{\ensuremath{\xpp^{\cap}_{es}}}
\newcommand{\xppdesccap}{\ensuremath{\xpp^{\cap}_{//}}}

\newcommand{\xpi}{\mbox{{\it XP }}\ensuremath{^{\cap, \cup}~}}
\newcommand{\xpcap}{\mbox{{\it XP}}\ensuremath{^{\cap}~}}

\newcommand{\xpint}{\textit{XPint}~}

\newcommand{\bbbracketleft}{[\hspace{-1.5pt}[}
\newcommand{\bbbracketright}{]\hspace{-1.5pt}]}

\newcommand{\rpath}{\textit{rpath} }
\newcommand{\apath}{\textit{apath} }

\newcommand{\ipath}{\textit{ipath} }
\newcommand{\cpath}{\textit{cpath} }
\newcommand{\jpath}{\textit{jpath} }

\newtheorem{claim}{Claim}[section]

\newcommand{\f}[2]{\ensuremath{f_{\rm #1}({#2})}}

\newcommand{\snippet}[1]{\textsf{\small{#1}}}

\newcommand{\txt}[1]{\mbox{{\it text}}\ensuremath{\left({#1}\right)}}
\newcommand{\tst}[1]{\mbox{{\it test}}\ensuremath{({#1})}}
\newcommand{\algoname}[0]{\proc{Rewrite} }
\newcommand{\allrw}[0]{\proc{All-Rewrites} }
\newcommand{\efficient}[0]{\proc{Efficient-Rw} }
\newcommand{\apprules}[0]{\proc{Apply-Rules} }

\newcommand{\eat}[1]{}

\newcommand{\candRw}[2]{{\ensuremath{\mathit{cand}_{\mathit{RW}}({#1},{#2})}}}

\makeatletter
\def\setliststart#1{\setcounter{\@listctr}{#1}%
\addtocounter{\@listctr}{-1}}
\makeatother

\begin{document}
\date{}
\title{Rewriting   XPath Queries using View Intersections: Tractability versus Completeness}
\markboth{}{Rewriting   XPath Queries using View Intersections: Tractability versus Completeness}
\author{
BOGDAN CAUTIS\affil{Universit\'{e} Paris-Sud -- INRIA\footnote{Work partially done while this author was affiliated with Telecom ParisTech.}}
ALIN DEUTSCH\affil{UC San Diego}
IOANA ILEANA\affil{Telecom ParisTech}
NICOLA ONOSE\affil{Google Inc.\footnote{Work partially done while this author was affiliated with UC San Diego.}}
}

\begin{abstract}
The standard approach for optimization of XPath queries by rewriting using views techniques consists in navigating inside a view's output, thus allowing the usage of only one view in the rewritten query. Algorithms for richer classes of XPath rewritings, using intersection or joins on node identifiers, have been proposed, but they either lack completeness guarantees, or require additional information about the data. We identify the tightest restrictions under which an XPath can be rewritten in polynomial time using an intersection of views and
propose an algorithm that works for any documents or type of identifiers. As a side-effect, we analyze the complexity of the related problem of deciding if an XPath with intersection can
be equivalently rewritten as one without intersection or union.   We extend our formal study of the view-based rewriting problem for XPath by describing also (i) algorithms for more complex rewrite plans, with no limitations on the number of intersection and navigation steps inside view outputs they employ, and (ii) adaptations of our techniques to deal with XML documents without persistent node Ids, in the presence of XML keys.  Complementing our computational complexity study, we describe  a proof-of-concept implementation of our  techniques and possible choices that may speed up execution in practice, regarding how rewrite plans are built, tested and executed. We also give a thorough experimental evaluation of these techniques, focusing on scalability and the  running time improvements achieved  by the execution of view-based plans. 
\vspace{-3mm}
\end{abstract}

\maketitle

\section{Introduction}\label{section:intro}
The problem of equivalently rewriting queries using views is fundamental to several classical data management tasks. While the rewriting problem has been well studied for the relational data model, its XML counterpart is not yet equally well understood, even for basic XML query languages such as XPath, due to the novel challenges raised by the features of the XML data model.

XPath~\cite{xpath1} is the standard for navigational queries over XML data and it is widely used, either directly, or as part of more complex languages (such as XQuery~\cite{xquery}). Early research~\cite{DBLP:conf/vldb/XuO05,mandhani-suciu,DBLP:conf/xsym/TangZ05,CachingVLDB03} studied the problem of equivalently rewriting an XPath by navigating inside a {\em single} materialized XPath view. This is the only kind of rewritings supported when the query cache can only store or can only obtain {\em copies} of the XML elements in the query answer, and so the original node identities are lost.

We have recently witnessed an industrial trend towards enhancing XPath queries with the ability to expose node identifiers and exploit them using intersection of node sets (via identity-based equality). This trend is supported by such systems as~\cite{DBLP:conf/vldb/BalminOBCP04} and has culminated in the adoption of intersection as a first-class primitive of the XPath standard, starting from XPath 2.0~\cite{xpath2} and through the new XPath 3.0 standard~\cite{xpath3}. In a more general setting, intersection between collections of nodes can be based not only on physical node identifiers, but also on logical ids or keys. Research on keys for XML, such as the ones proposed 
in~\cite{BunemanDFHT03}, led to the introduction of a special \emph{key} construct in the XML Schema \cite{xsd} standard, which allows to uniquely identify a node based on the result of an XPath expression.

This development enables for the first time multiple-view rewritings obtained by intersecting several materialized view results. The single-view rewritings considered in early XPath research have only limited benefit, as many queries with no single-view rewriting can be rewritten using multiple views. Our work is the first to characterize the complexity of the in\-tersection\--aware rewriting problem.  We are interested in sound and complete algorithms, which are guaranteed to find  a rewriting whenever one exists.

%We consider XPath queries with child and descendant navigation and predicates, without wildcards.  

Our main objective  is to identify a   fragment of XPath that is as large as possible, while admitting polynomial-time rewriting that remains complete.  We exhibit a fragment of XPath with these properties showing that it is maximal in the sense that extending it renders the rewriting problem intractable (coNP hard). The fragment is practically interesting as it permits expressive queries and views with child and descendant navigation and path filter predicates, but no wildcard labels.

As a side-effect of our study on rewriting, we analyze the complexity of the problem of deciding if an XPath with intersection can be equivalently rewritten as one without intersection or union, case in which we say it is \emph{union-free}.  We also study the effect of intersection on the complexity of containment. Our hardness results thus immediately apply to XPath 2.0 and XPath 3.0 queries.

Prior work on XPath containment derived coNP lower bounds in the presence of wildcard  navigation, yet showed PTIME for tree patterns without wildcard \cite{DBLP:journals/jacm/MiklauS04}. In contrast, we show that extending wildcard-free tree patterns with intersection already leads to intractability.

\vspace{1mm}

\noindent
{\bf Running example.}  Throughout the paper we will consider an
example based on XPath queries over a digital library, which consists
in a large number of publications, including scientific papers. A
paper is organized into a hierarchy of sections, which may include,
among other things, figures and images, usually related to the
theorems and other results stated in the papers.

Let us assume that there has already been a query $v_1$, that
retrieved all images appearing in sections with theorem statements:
%%\vspace{-0.5cm}
$$v_1: \snippet{doc(``L'')//paper//section[theorem]//image}$$
The result of $v_1$ is stored in the cache as a materialized view, rooted at 
an element named $v_1$.
%% and containing nodes that preserved the identifiers
%% of the original document \snippet{L}.
Later, the query processor had to answer another XPath $v_2$ looking
for images inside (floating) figures that can be referenced:
%%appearing in sections with theorem statements:
%%\vspace{-0.15cm}
$$v_2: \snippet{doc(``L'')/lib/paper//section//figure[caption//label]/image}$$%%\vspace{-0.15cm}$$
The result of $v_2$ is not contained in that of $v_1$, so it was also
executed and its answer cached.

Let us first look at an incoming query $q_1$, asking for all
postscript images that appear in sections with theorems:
%%\vspace{-0.15cm}
$$q_1: \snippet{doc(``L'')//paper//section[theorem]//image[ps]}$$%%\vspace{-0.15cm}$$
$q_1$ can be easily answered by navigating inside the view $v_1$, using
the following XPath query:
%%\vspace{-0.15cm}
$$r_1:~~~ \snippet{doc(``$v_1$'')/$v_1$/image[ps]}$$%%\vspace{-0.15cm}$$
Now, consider a query $q_2$ looking for the files corresponding to
images inside labeled figures from sections stating theorems:
%%\vspace{-0.15cm}
$$q_2: \snippet{doc(``L'')/lib/paper//section[theorem]//figure[caption//label]/image/file}$$%%\vspace{-0.15cm}$$
It is easy to see that $q_2$ cannot be answered in isolation using only
$v_1$ or only $v_2$, because, for instance, there is no way to enforce
that an image is both in a section having theorems and inside a
labeled figure.
However, by intersecting the results of the two views 
(assuming they both preserve the identities of the original image elements), one can build a
rewriting equivalent to $q_2$:
%%\vspace{-0.15cm}
%% $r_2$: \snippet{doc(``$v_1$'')/$v_1$/image/file $\cap$ doc(``$v_2$'')/$v_2$/image/file}
$$r_2: \snippet{(doc(``$v_1$'')/$v_1$/image $\cap$ doc(``$v_2$'')/$v_2$/image)/file}$$

\vspace{-3mm} \noindent \textbf{Outline. }
This paper is organized as follows. We discuss related work in Section~\ref{sec:related}. Section~\ref{sec:preliminaries} introduces general notions for tree and DAG patterns, and the rewriting problem. In Section~\ref{sec:algo} we give a high-level view on our rewriting algorithm \algoname. We then zoom in on the rewrite rules on which it is based in Section~\ref{sec:rules}. We discuss in Section~\ref{sec:keys} how these techniques can apply even in the absence of persistent node Ids, under XML key constraints. We present the formal guarantees of algorithm \algoname in Section~\ref{sec:guarantees}, in terms of soundness,  completeness and complexity bounds;  we also analyze the related problem of union-freedom for DAG patterns. As the general rewriting problem is coNP-complete, we then study the most permissive restrictions on the language of queries or rewrite plans that enable a sound and complete approach (Sections~\ref{sec:frontier1} and~\ref{sec:frontier2}). We consider a richer language for rewrite plans in Section~\ref{sec:nestedinter}, which can have arbitrary many steps of intersection and compensation of views. We discuss implementation issues and optimization opportunities in Section~\ref{sec:implementation} and we present our experiments is Section~\ref{sec:experiments}. We conclude in Section~\ref{sec:conclusion}. We detail two of the more involved proofs in an appendix, in Sections~\ref{sec:proof1} and~\ref{sec:completenessUF-desc-akin}.

%% (:: Notations: satisfiability - SAT, containment - CNT ::)
%\vspace{-0.1cm}
\section{Related Work}\label{sec:related}
The area of  rewriting XPath queries using views  lacks in general theoretical foundations, as most related works propose incomplete algorithms or impose strong limitations.

XPath rewriting using only one view (no intersection) was the target
of several studies \cite{DBLP:conf/vldb/XuO05,mandhani-suciu,DBLP:conf/xsym/TangZ05,CachingVLDB03,Wu2009}, possibly in the presence of DTD constraints~\cite{AravogliadisV11}.    
%% Balmin et al.
%% \cite{DBLP:conf/vldb/BalminOBCP04} use a heuristic algorithm that
%% enables the use of multiple views, but did not investigate the
%% complexity of the problem.  Wang et al. \cite{DBLP:conf/wise/WangYL07}
%% consider the same XPath fragment as we do, reporting results on two
%% problems: the conditions for a view to be redundant for contained
%% rewritings and the question whether any equivalent rewriting written
%% as a union of intersections of XPath expressions can be expressed as
%% only one XPath (intersection). The case without intersection is
%% implied by our Lemma~\ref{lem:equiv_tree_union}. A more recent
%% paper~\cite{waterlooicde08} proposes heuristic methods to select
%% relevant views and to compute a rewriting by joining a subset of them,
%% without providing any completeness guarantees.  Their method assumes
%% an encoding scheme for the node ids of the input XML data that
%% captures relationships between nodes, while we can also use
%% application level ids.  Joins based on id-equality are also used in
%% \cite{DBLP:conf/vldb/ArionBMP07} for the purpose of rewriting tree
%% patterns using views, but assuming knowledge of the structure of the
%% document as a \emph{summary}.  Our algorithm is valid for queries
%% asked over arbitrary documents, and we determine the most expressive
%% language (XPath with a few restrictions) for which finding a rewriting
%% is tractable.
Previously proposed join-based rewriting methods either give no
completeness guarantees~\cite{DBLP:conf/vldb/BalminOBCP04,waterlooicde08} 
or can do so only if the query engine has extra knowledge
about the structure and nesting depth of the XML
document~\cite{DBLP:conf/vldb/ArionBMP07}.
Others~\cite{waterlooicde08} can only be used if the
node ids are in a special encoding, containing structural information.
Our algorithm works for any documents and type of identifiers,
including application level ids, such as the id attributes defined in the
XML standard~\cite{xml} or XML Schema keys~\cite{xsd}.
%% Moreover, we identify the most expressive
%% language (XPath with a few restrictions) for which finding a rewriting
%% is tractable.
 In \cite{DBLP:conf/vldb/LakshmananWZ06,DBLP:conf/dexa/GaoWY07,Wang2011}, the authors 
look at a different problem, that of finding maximally contained
rewritings of XPath queries using views. Rewriting more expressive
XML queries using views was studied in
\cite{DBLP:conf/webdb/ChenR02,DBLP:conf/vldb/DeutschT03,DBLP:conf/sigmod/OnoseDPC06}, 
but without considering intersection.
Fan et al \cite{DBLP:conf/icde/FanGJK07} define views using DTDs instead
of queries and study the problem of rewriting an XPath using one view 
DTD.
 In~\cite{Afrati2011}, for a different XPath fragment (including wildcard labels), the authors describe a sound but incomplete algorithm for finding equivalent rewritings  as unions
of single-view rewritings.
%% In \cite{DBLP:conf/wise/WangYL07}, besides equivalent rewritings, they
%% also consider maximally contained rewritings MCRs which are unions of
%% intersections. However, the only result they claim is that when an MCR
%% is an equivalent rewriting, then one of the disjuncts must be
%% equivalent to the original query.
%% \cite{waterlooicde08} rewrites by compensating the views and then
%% joining the roots(?) of the views.
 Several works considered the problem of choosing the optimal set of views to materialize in order to support a given query workload (see~\cite{Katsifodimos2012} and the main references therein).

\eat{
\bogdan{
\textbf{Comparison with prior publication.} This article is based on an earlier extended abstract~\cite{CautisWebDB08} (workshop paper),  with respect to which it brings several new important contributions.  We provide a more complete study on the computational complexity of the view-based rewriting problem, describing decision procedures for settings in which the techniques of  \cite{CautisWebDB08} would not be applicable. More precisely, (i) we consider more complex rewrite plans, with no limitations on the number of intersections steps they employ, and (ii) we  describe how our techniques  can be modified to deal with XML documents without persistent node Ids, in the presence of XML keys.  Moreover, optimization issues - on how  rewrite plans are chosen, built and evaluated - were not considered in \cite{CautisWebDB08}, whose focus was on computational complexity.  We discuss such issues in this paper, as well as certain implementation choices that may speed up execution in practice.  We  also  give a thorough experimental evaluation of all our techniques. We provide  the complete proofs for all theoretical results (no proofs were given in~\cite{CautisWebDB08}). We believe  these are of interest on their own,  as they are based on  various novel techniques for analyzing and  reasoning about XPath. We also present  additional examples, whose role is to illustrate more complex rewrite plans or to complement our proofs. }
}

\textbf{Most related prior work.} This article extends the results we present in the extended abstract~\cite{CautisWebDB08},  with respect to which it brings several new important contributions.  We provide a more complete study on the computational complexity of the view-based rewriting problem, describing decision procedures for settings in which the techniques of  \cite{CautisWebDB08} would not be applicable. More precisely, (i) we consider more complex rewrite plans, with no limitations on the number of intersection steps they employ, and (ii) we  describe how our techniques  can be modified to deal with XML documents without persistent node Ids, in the presence of XML keys.  Moreover, we report on a systems contribution, pertaining to the implementation
of an XPath rewriting engine.  Optimization issues -- on how  rewrite plans are chosen, built and evaluated -- were not considered in \cite{CautisWebDB08}, whose focus was on computational complexity only.  We discuss such issues in this paper, as well as certain implementation choices that may speed up execution in practice.   In particular, we introduce the theoretical foundations and we report on the implementation of a PTIME technique for partially minimizing redundancy in rewritings without paying the price of full minimization (which is NP-complete).  We  also  give a thorough experimental evaluation of the presented techniques. We provide  the complete proofs for all theoretical results (no proofs were given in~\cite{CautisWebDB08}). We believe  these are of interest on their own,  as they are based on  various novel techniques for analyzing and  reasoning about XPath. We also present  additional examples, whose role is to illustrate more complex rewrite plans or to complement our proofs.

Sound and complete algorithms for rewriting XML queries using multiple views  were also proposed later (after the  publication of~\cite{CautisWebDB08}), 
in~\cite{Manolescu11}.  There,  the focus is not on tractable rewriting. Indeed, the authors target a more expressive language,  tree pattern queries with value joins and multiple arity,  for which equivalence is intractable and no complete rewriting algorithm implementation can go below the exponential bound; this is for two reasons: (i) the coNP-hardness result (ref. theorem),  even in the absence of value joins and unary tree pattern queries, and (ii) the NP-hardness of relational conjunctive query rewriting, which can be encoded by tree patterns with joins (\cite{ChandraM77,DBLP:conf/sigmod/OnoseDPC06}).  Moreover,~\cite{Manolescu11} focuses on the minimality of rewrite plans, which brings another exponential in the total running time. In essence, both the algorithm in~\cite{Manolescu11} and the one in~\cite{CautisWebDB08} (for the intractable case) amount to reformulating an intersection of tree patterns into a union of intersection-free tree patterns (a.k.a. interleavings; see Section~\ref{sec:preliminaries}). There can be exponentially many interleavings, which is unavoidable given the coNP lower bound.   

 In contrast to~\cite{Manolescu11},  this submission studies the most expressive language for input queries (XPath with a few restrictions) for which  finding a rewriting is tractable (we present a sound and complete procedure for which we can guarantee polynomial time under these restrictions). In addition,  we study the more general problem of rewriting XPath queries using multiple views joined by Ids, show intractability beyond our restrictions and for that case we present  an exponential  rewriting algorithm that is sound and complete.
 %     and we give matching upper and lower bounds for the complexity of this problem.  
 
 Ways to explore the space of possible rewrite plans using views, for minimization purposes, have been considered in previous literature (see for instance~\cite{DBLP:conf/sigmod/PopaDST00}). Intuitively, they start from a rewriting and randomly prune certain components of the plan while maintaining equivalence; this reveals  a threshold on the size of minimal plans (in terms of number of views) to be considered. Even though minimization lies beyond the scope of our paper, we note that  the techniques presented here do create the search space in which  all minimal rewritings are to be found, creating the opportunity to plug in techniques for exploring the search space, such as in~ \cite{DBLP:conf/sigmod/PopaDST00}.

%\bogdan{Maybe mention experiments}

\eat{
\bogdan{If Ioana did a lot on minimization we should not enter into competition with her. Claim that completeness for minimal rewritings at exponential cost, many ways to explore the space, all this has been explored in previous literature (chase too far paper). Start from a rewriting, randomly prune while maintaining equivalence, gives  a threshold then just put views of number at most that threshold. The important thing is that we can create the search space in which the minimal rewritings can be found. But minimality beyond the scope of the paper, many works on this...}
}

Containment and satisfiability for several extensions of XPath with
intersection have been previously investigated, but all considered
problems were at least NP-hard or coNP-hard.  For our language,
containment is also intractable, but the equivalence test used in the
rewriting algorithm is in PTIME for practically relevant
restrictions. Satisfiability of XPath in the presence of the intersect
operator and of wildcards was analyzed in
\cite{DBLP:conf/dbpl/Hidders03}, which proved its NP-completeness.
As noticed in
\cite{DBLP:conf/pods/BenediktFG05}, there is a tight relationship
between satisfiability and containment for languages that can express
unsatisfiable queries. If containment is in the class K,
satisfiability is in coK and if satisfiability is K-hard, containment
is coK-hard.\footnote{This is based on the observation that a query is
satisfiable iff it is not contained in a query that always returns the empty set.} We give even stronger coNP completeness results for
the containment of an XPath $p_1$ into an XPath $p_2$,
by allowing intersection only in $p_1$ and disallowing
wildcards.  Satisfiability is analyzed in
\cite{DBLP:conf/pods/BenediktFG05} for various fragments of XPath,
including negation and disjunction, which could together simulate
intersection, but lead to coPSPACE-hardness for checking
containment. Richer sublanguages of XPath 2.0, including path
intersection and equality, are considered in
\cite{1265541}, where complexity of checking containment goes up to EXPTIME or
higher. None of these studies tries to identify an efficient test for
using intersection in query rewriting.  A different approach, taken by
\cite{DBLP:conf/icde/GroppeBG06} is to replace intersection by using a
rich set of language features, and then try to simplify the expression
using heuristics.

Finally, closure under intersection was analyzed in
\cite{BenediktTCS05} for various XPath fragments, all of which use
wildcard.  We study the case without wildcard and prove that
\emph{union-freedom} (equivalence between an intersection of XPaths
and an XPath without intersection or union) is coNP-hard.
However, under restrictions similar to those for the rewriting
problem, union-freedom can be solved in polynomial time.  Thus, we
also answer a question was previously raised
in~\cite{DBLP:conf/pods/CautisAM07} regarding whether an intersection
of XPath queries without wildcard can be reduced in PTIME to
only one XPath.

{\cite{DBLP:journals/tods/CautisDOV11}  describes an algorithm for rewriting using multiple views, designed especially for views specified by means of a program (a Query Set Specification). 
Completeness is achieved there for  input queries having at least one descendant edge in the root to output-node path (so called multi-token queries), and for a restricted language for rewrite plans (intersections of views). For this reason, the result  does not apply to our setting, and a different technique is needed. Indeed,  the technique of~\cite{DBLP:journals/tods/CautisDOV11} applies individual tests on the view definitions instead of rewrite rules on the corresponding DAG pattern.

\section{Preliminaries}\label{sec:preliminaries}
We consider an XML document as an unranked, unordered rooted tree $t$ modeled by a set of edges \edges{t}, a set of nodes \nodes{t}, a distinguished root node \droot{t} and a labeling function $\lambda_t$, assigning to each node a label from an infinite alphabet $\Sigma$. Every node $n$ of a tree has a text value \txt{n}, possibly empty.

We consider XPath queries with child / and descendant // navigation, without wildcards. We call the resulting language \xp  and define its grammar as:
\begin{eqnarray*}
\apath &::=& doc(``name")/\rpath ~|~ doc(``name")//\rpath\\
\rpath &::=& step ~ |~ \rpath/\rpath  ~|~ \rpath//\rpath\\
step &::=&label~pred \\
pred &::=& \epsilon ~ |~ [\rpath] ~ | ~ [\rpath=C] ~ | ~ [.//\rpath] ~ | ~  [.//\rpath=C] ~| ~ pred ~ pred
\end{eqnarray*}
Expressions in \xp are produced from the symbol \apath, and they correspond to \emph{absolute paths}, that is, queries expressed starting from the document root. The \rpath symbol generates \emph{relative path} expressions, i.e. encoding navigation relative to a given document context. The sub-expressions inside brackets are called \emph{predicates}. $C$ terminals stand for text constants.

The semantics of \xp can be defined as follows: 
\begin{definition}[\xp Semantics]
The result of evaluating an \xp expression $q$ over an XML tree $t$ is defined as a binary relation over \nodes{t}:
\begin{enumerate}
\item ${\bbbracketleft label \bbbracketright}_t = \{(n,n') | (n,n') \in \textsc{EDGES}(t),
\lambda_t(n')=label\}$
\vspace{1mm}
\item ${\bbbracketleft pred \bbbracketright}_t = \{n | n\in
\nodes{t}, \textit{pred}(n)=\textit{true}\}$.
\begin{enumerate}
\vspace{1mm}
\item Let \textit{pred} be defined as $[rp]$ or $[.//rp]$ and let $t_n$ denote the subtree rooted at $n$ in $t$.  We say that
$\textit{pred}(n)=\textit{true}$ iff ${\bbbracketleft \lambda_t(n)/rp
\bbbracketright}_{t_n} \neq \oslash $ ( ${\bbbracketleft
\lambda_t(n)//rp \bbbracketright}_{t_n} \neq \oslash$, resp.).
\vspace{1mm}
\item If
\textit{pred} is of the form $[rp=C]$ (or
$[.//rp=C]$) then $\textit{pred}(n)=\textit{true}$ iff
$\txt{ {\bbbracketleft \lambda_t(n)/rp \bbbracketright}_{t_n} } =C$
(or $\txt{{\bbbracketleft \lambda_t(n)//rp
\bbbracketright}_{t_n}} = C$, resp.).  
\end{enumerate}
\vspace{1mm}
\item ${\bbbracketleft pred_1 ~pred_2  \bbbracketright}_t = { \bbbracketleft pred_1 \bbbracketright}_t   \cap {\bbbracketleft pred_2  \bbbracketright}_t$
\vspace{1mm}
\item ${\bbbracketleft \epsilon  \bbbracketright}_t = \nodes{t}$
\vspace{1mm}
\item ${\bbbracketleft label ~ pred \bbbracketright }_t=\{(n, n') |
(n,n')\in {\bbbracketleft label\bbbracketright }_t, n'\in
{\bbbracketleft pred\bbbracketright }_t\}$
\vspace{1mm}
\item ${\bbbracketleft \rpath_1/\rpath_2\bbbracketright }_t = \{(n,
n') | (n,n')\in {\bbbracketleft \rpath_1\bbbracketright }_t \circ
{\bbbracketleft \rpath_2\bbbracketright }_t\}$
\vspace{1mm}
\item ${\bbbracketleft \rpath_1//\rpath_2\bbbracketright }_t = \{(n,
n') | (n,n')\in {\bbbracketleft \rpath_1\bbbracketright }_t \circ \textsc{EDGES}^*(t)
\circ {\bbbracketleft \rpath_2\bbbracketright }_t\}$
\vspace{1mm}
\item ${\bbbracketleft \doc{``name"}/\rpath \bbbracketright }_t = \{ (\droot{t},n)
  \,|\,  (\droot{t},n) \in {\bbbracketleft \rpath\bbbracketright }_t \}$
  \vspace{1mm}
\item ${\bbbracketleft \doc{``name"}//\rpath \bbbracketright }_t = \{ (\droot{t},n')
  \,|\,  (\droot{t},n') \in \textsc{EDGES}^*(t) \circ
     {\bbbracketleft \rpath\bbbracketright }_t\}$.
\end{enumerate}
 \doc{``name"} returns the root of the document storing $t$.  We denote by $\circ$ the standard binary relation composition, that is $R \circ S = \{ (r,s) | (r,x) \in R, (x,s) \in S \}$.
%The symbol $\epsilon$ in the predicate definition has the semantics of selecting the entire set of nodes $\nodes{t}$ (it denotes the empty predicate) :
%$${\bbbracketleft \epsilon \bbbracketright}_t = \{n | n\in \nodes{t}\}$$
%\color{black}
\end{definition}

%{\color{red} Remark about binary semantics? }

In the following, we will prefer for XML queries an alternative representation widely used in literature, the unary \emph{tree patterns}~\cite{DBLP:journals/jacm/MiklauS04}\footnote
{Miklau and Suciu (and most follow-up works) provide a node set semantics for tree patterns. Our semantics is equivalent to node set semantics, despite the binary representation. 
We just  repeat the context node with each of the selected nodes, instead of writing it once for the entire set as in~\cite{DBLP:journals/jacm/MiklauS04}.  This will prove more convenient for our formal development. There exists a line of
 work on distinguishing the expressive power between binary (path set semantics) and node set semantics~\cite{DBLP:conf/bncod/WuGGP09}, but it does not apply here. In~\cite{DBLP:conf/bncod/WuGGP09}, the distinction boils down to 
 allowing one versus two distinguished nodes in the pattern. In our work we only have one distinguished node.}:
\begin{definition}\label{def:treepattern}
A \emph{tree pattern} $p$ is a non empty rooted tree, with a set of nodes \nodes{p} labeled with symbols from $\Sigma$, a distinguished node called the \emph{output node} $\out{p}$, and two types of edges: \emph{child edges}, labeled by $/$ and \emph{descendant edges}, labeled by $//$. The root of $p$ is denoted \droot{p}. Every node $n$ in $p$ has a test of equality \tst{n} that is either the empty word $\epsilon$, or a constant $C$. If $n$ is on a path between \droot{p} and \out{p}, then \tst{n} is $\epsilon$.
\end{definition}

Any \xp expression can be translated into a tree pattern query and vice versa (see, for instance~\cite{DBLP:journals/jacm/MiklauS04}). For a given \xp expression $q$, by \pattern{q} we denote the associated tree pattern $p$ and by $\xpath{p} \equiv q$ the reverse transformation.

The semantics of a tree pattern can be given using embeddings:
\begin{definition}\label{def:embedding}
An \emph{embedding} of a tree pattern $p$ into a tree $t$ over $\Sigma$ is a
function $e$ from \nodes{p} to \nodes{t} that has
the following properties:
\begin{enumerate}[(1)]
\vspace{1mm}
  \item\label{it:emb-root} $e(\droot{p}) = \droot{t}$;
  \vspace{1mm}
  \item\label{it:emb-label} \label{it:emb-label} for any $n \in \nodes{p}$,
    $\etiq{e(n)} = \etiq{n}$;
    \vspace{1mm}
  \item \label{it:emb-test} for any $n \in \nodes{p}$,
    if $\tst{n}=C$ then $\txt{e(n)}=C$;
    \vspace{1mm}
  \item \label{it:emb-child} for any /-edge $(n_1,n_2)$ in $p$,
    $(e(n_1),e(n_2))$ is an edge in $t$;
    \vspace{1mm}
  \item\label{it:emb-desc} \label{it:emb-desc} for any //-edge $(n_1,n_2)$ in $p$, there
    is a path from $e(n_1)$ to $e(n_2)$ in $t$.
\end{enumerate}
\end{definition}

\noindent The \emph{result} of applying a tree pattern $p$ to an XML tree $t$ is the set: 
$$\left\{\left(\droot{t}, e(\out{p})\right) | \textit{ e is an embedding of p into t } \right\}$$

We will consider in this paper the extension \xpcap of \xp with respect to intersection. Expressions in \xpcap are generated from the symbol \ipath, by adding the following rules to the grammar of \xp:
\begin{eqnarray*}
\ipath &::=& \cpath~|~(\cpath) |~(\cpath)/\rpath~|~(\cpath)//\rpath\\
\cpath &::=& \apath ~|~\cpath \cap \apath 
\end{eqnarray*}
The symbol \cpath defines a single level of intersection of \xp expressions, e.g. $$\snippet{doc(``$v_1$'')/$v_1$/image $\cap$ doc(``$v_2$'')/$v_2$/image}.$$
\ipath adds to this intersection an \rpath expression, thus allowing additional (relative) navigation from the nodes in the intersection result, e.g. $$\snippet{(doc(``$v_1$'')/$v_1$/image $\cap$ doc(``$v_2$'')/$v_2$/image)/file}$$
Note that by definition \xpcap does not include arbitrary nested intersections of \xp queries. We defer the analysis of the language expressing such nested intersections (a superset of \xpcap) to Section \ref{sec:nestedinter}.

Formally,  \xpcap has the following semantics:
\begin{itemize}
\vspace{1mm}
\item $\cpath \cap \apath = {\bbbracketleft \cpath \bbbracketright}_t \cap {\bbbracketleft \apath \bbbracketright}_t$\footnote{We overloaded the intersection operator: while on the left side it refers to the \xpcap syntax, on the right side it denotes the classic set intersection operation.}
\vspace{1mm}
\item ${\bbbracketleft \cpath/\rpath\bbbracketright }_t = \{(n,
n') | (n,n')\in {\bbbracketleft \cpath \bbbracketright }_t \circ
{\bbbracketleft \rpath\bbbracketright }_t\}$
\vspace{1mm}
\item ${\bbbracketleft \cpath//\rpath\bbbracketright }_t = \{(n,
n') | (n,n')\in {\bbbracketleft \cpath\bbbracketright }_t \circ  \textsc{EDGES}^*(t)
\circ {\bbbracketleft \rpath \bbbracketright }_t\}$

\end{itemize}

By \xpcap expressions over a set of documents $D$ we denote those that use only  \apath expressions that navigate inside the documents $D$. For a fragment $\cal L \subseteq \xp$,  by $\cal L^\cap \subseteq \xpcap$ we denote the \xpcap expressions that use only \apath expressions from $\cal L$.

Similar to the \xp - tree pattern duality, we can represent \xpcap ex\-pressions using the more general \emph{DAG patterns}:

\begin{definition}\label{def:dag}
A \emph{DAG pattern} $d$ is a directed acyclic graph, with a set of
nodes \nodes{d} labeled with symbols from $\Sigma $, a distinguished
node called the \emph{output node} $\out{d}$, and two types of edges:
\emph{child edges}, labeled by $/$ and \emph{descendant edges},
labeled by $//$.
$d$ has to satisfy the
property that any $n \in \nodes{d}$ is accessible via a path starting
from a special node \droot{d}. In addition, all the nodes that are not on a path from $\droot{d}$ to $\out{d}$ (denoted \emph{predicate nodes}) have only one incoming edge.
Every node $n$ in $d$ has a test of equality \tst{n}
that is either the empty word $\epsilon$, or a constant $C$.
If $n$ is on a path between \droot{d} and \out{d}, then 
\tst{n} is always $\epsilon$.
\end{definition}

Figure~\ref{fig:runexample}(a) gives an example of a DAG pattern.  \droot{d} is the \doc{L} node and \out{d} is the \textit{image} node indicated by a square.

\begin{figure}[t]
\vspace{-3mm}
\begin{center}
\input{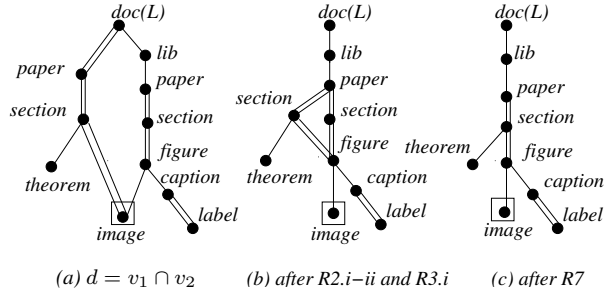}
\end{center}
\vspace{-1mm}
\caption{\small Running the rules on the example of Section~\ref{section:intro}.\label{fig:runexample} }
\vspace{-2mm}
\end{figure}

In our algorithm (Section~\ref{sec:algo}) we obtain the \xpcap expressions that are rewriting candidates directly. We only use the DAG pattern representation for the equivalence check involved in  validating  these candidates. We therefore only need to translate from \xpcap into 
DAG patterns, but not conversely.  We specify the one-way translation below. 

\textbf{Representing \xpcap by DAG patterns.} For a query $q$
in $\xpcap$, we construct the associated pattern, denoted \texttt{dag}($q$), as follows:
\begin{enumerate}
\vspace{1mm}
  \item for every \textit{apath} (\xp path with no $\cap$),
    \texttt{dag}(\apath\!\!) is the tree pattern corresponding to the \apath.
    
    \vspace{1mm}
  \item  \texttt{dag}($p_1 \cap p_2$) is obtained from \texttt{dag}($p_1$)
    and \texttt{dag}($p_2$) as follows: (i) provided there are no labeling conflicts and both $p_1$ and $p_2$ are not empty, by coalescing \droot{\texttt{dag}(p_1)} with
    \droot{\texttt{dag}(p_2)} and  \out{\texttt{dag}(p_1)} with
    \out{\texttt{dag}(p_2)} respectively, (ii) otherwise, as the empty pattern.
    \vspace{1mm}
  \item \texttt{dag}($x$/\rpath) and
    \texttt{dag}($x$//\rpath) are obtained as follows: (i) for non-empty $x$, by appending the
    pattern corresponding to \rpath to
    \out{\texttt{dag}(x)} with a /- and a //-edge respectively, (ii) as $x$, if $x$ is the empty pattern.
\end{enumerate}

By a pattern from the language $\cal L$ we denote any pattern built as $dag(q)$, for any $q \in {\cal L}$. Note that a tree pattern is a DAG pattern as well. The notion of \emph{embedding} and the semantics of a pattern can be extended in straightforward manner from trees to DAGs. In the following, unless stated otherwise, all patterns are DAG patterns. We can prove the following:
\begin{theorem}\label{thm:dag-equiv}
For any $q \in \xpcap$ and any tree $t$, $q(t) = \texttt{dag}(q)(t)$.
\end{theorem}

By the main branch nodes of a pattern $d$, \mbn{d}, we denote the set of nodes found on paths starting with \droot{d} and ending with \out{d}. We refer to main branch paths between \droot{d} and \out{d} as \emph{main branches} of $d$. The (unique) main branch of a tree pattern $p$ is denoted \mb{p}. 

\begin{definition}
  A pattern $d_1$ is \emph{contained} in another pattern $d_2$
  iff for any input tree $t$, $d_1(t) \subseteq d_2(t)$.
  We write this shortly as $d_1 \sqsubseteq d_2$.
  We say that $d_1$ is \emph{equivalent} to $d_2$,
  and write $d_1 \equiv d_2$, iff $d_1(t) = d_2(t)$ for any
  input tree $t$.
\end{definition}

We say that a pattern $p$ is \emph{minimal}~\cite{DBLP:journals/vldb/Amer-YahiaCLS02} if it is equivalent to none of its strict sub-patterns. 

\begin{definition}\label{def:dag-mapping}
  A \emph{mapping} between two patterns $d_1$
  and $d_2$ is a function $h: \nodes{d_1} \rightarrow
  \nodes{d_2}$ that satisfies the properties
  \ref{it:emb-label},\ref{it:emb-desc} of an
  embedding (allowing the target to be a pattern) plus three others:
\begin{enumerate}[(1)]\setliststart{6}
\vspace{1mm}
 \item for any $n \in \mbn{d_1}$, $h(n) \in \mbn{d_2}$;
 
 \vspace{1mm}
 \item \label{it:mapping-child} for any /-edge $(n_1,n_2)$ in $d_1$,
    $(h(n_1),h(n_2))$ is a /-edge in $d_2$.
    
    \vspace{1mm}
 \item for any $n \in \nodes{d_1}$, if $\tst{n}=C$ then
     $\tst{h(n)} = C$;
\end{enumerate}

A \emph{root-mapping} is a mapping that satisfies (1).
A \emph{containment mapping} is a root-mapping $h$ such that $h(\out{d_1}) = \out{d_2}$.
\end{definition}

\begin{lemma}\label{l:mapping-suff}
  If there is a containment mapping from $d_1$ into $d_2$ then $d_2
  \sqsubseteq d_1$.
\end{lemma}

\begin{lemma}
\label{lem:equiv-iso}
Two tree patterns  are equivalent iff they are isomorphic after minimization.
\end{lemma}

\begin{proof}
It is a direct consequence of Theorem~1 from~\cite{mandhani-suciu}, because equivalence in \xp (we remind that our language \xp has no wildcard) is always witnessed by containment mappings in both directions.
\end{proof}

\begin{lemma}\label{l:containment-tree-dag}
A tree pattern $p$ is contained into a  DAG pattern $d$ iff there is a containment mapping from $d$ into $p$.
\end{lemma}

\begin{proof}[Sketch]  Consider the model $\textit{mod}'_p$ of $p$ in which //-edges are replaced by a sequence $/z/$ (two child edges),where $z$ is a fresh new label.If $p \subseteq d$, then in particular $d(\textit{mod}'_p) \neq \emptyset$. Since $z$ is a new label, $d$ can only embed a //-edge in a path fragment containing $z$.
\end{proof}

Note that in \xpcap unsatisfiable DAG patterns are possible (when there exists no model with non-empty results). For the purposes of this paper, we assume in the following only satisfiable patterns.

We say that two \xp queries $q_1$ and $q_2$ are \emph{incomparable} if there is no containment mapping between them. 

We now prove that we can always reformulate a DAG pattern as a (possibly empty) union of tree patterns.

As in \cite{BenediktTCS05}, a \emph{code} is a string of  $\Sigma$ symbols alternating with either / or //.

\begin{definition}[Interleaving]
\label{def:interleaving} By the
\emph{interleavings} of a pattern $d$ we denote any tree pattern $p_i$ produced as follows:
\begin{enumerate}
\vspace{1mm}
\item choose a code $i$ and a total onto function $f_i$ that maps \mbn{d}
 into $\Sigma$-positions of $i$ such that:

\begin{enumerate}
 \vspace{1mm}
 \item for any $n \in \mbn{d}, \etiq{f_i(n)} =
    \etiq{n}$
    
    \vspace{1mm}
 \item for any /-edge $(n_1,n_2)$ in $d$, the code $i$ is of the form
    %\vspace{-2mm}
    $\dots f_i(n_1)/f_i(n_2) \dots$,%\vspace{-0.1cm}$$
    
    \vspace{1mm}
 \item for any //-edge $(n_1,n_2)$ in $d$, the code $i$ is of the form
   %\vspace{-2mm}
   $\dots f_i(n_1)\dots f_i(n_2)\dots$.%\vspace{-0.1cm}$$
\end{enumerate}

\vspace{1mm}
\item build the smallest pattern $p_i$ such that:

\begin{enumerate}
\vspace{1mm}
\item $i$ is a code for the main branch \mb{p_i},

\vspace{1mm}
\item for any $n \in \mbn{d}$ and its image $n'$ in $p_i$ (via
$f_i$), if a predicate subtree $st$ appears below $n$ then a copy of
$st$ appears below $n'$, connected by same kind of edge.
\end{enumerate}
\end{enumerate}
Two nodes $n_1$, $n_2$ from \mbn{d} are said to be
\emph{collapsed} (or \emph{coalesced}) if $f_i(n_1)=f_i(n_2)$, with $f_i$ as above. The tree
patterns $p_i$ thus obtained are called \emph{interleavings} of $d$
and we denote their set by \emph{interleave($d$)}.
\end{definition}
For instance, one of the seven interleavings of $d$ in
Figure~\ref{fig:runexample}(a) is the pattern in
Figure~\ref{fig:runexample}(c) and another one corresponds to the
XPath
$$\snippet{\doc{L}/lib/paper//paper//section[theorem]//figure[caption[.//label]]/image}$$%%%\vspace{-0.1cm}$$
We say that a pattern $d$ is \emph{satisfiable} if it is non-empty and the set \emph{interleave($d$)} is non-empty. By definition, there is always a containment mapping from a satisfiable pattern into each of its interleavings. Then, by Lemma~\ref{l:mapping-suff}, a pattern will always contain its interleavings. Similar to a result from~\cite{BenediktTCS05}, it also holds that:

\begin{lemma}\label{lem:cap_sub_cup}
Any DAG pattern is equivalent to the union of its interleavings.
\end{lemma}

\begin{proof}
We only need to consider the other inclusion, from $d$ into $\bigcup p_i$. We show that for any XML tree $t$ and any node $n \in t$ such that $(\droot{t},n) \in e(t)$, for some embedding $e$ of $d$ into $t$ (so $e(\out{d})=n$), we can always find an interleaving $p_i$ and embedding $e_i$ of $p_i$ in $t$ such that $(\droot{t},n) \in
e_i(t)$. This would be enough to conclude the proof of inclusion (and equivalence).

Let $p$ denote the linear path from $\droot{t}$ to $n$ (endpoints included) and let $c$ denote the code of $p$. Let $id$ denote the one-to-one mapping from $p$ to $c$. Note that $e$ gives us a mapping $\textit{id} \circ e$ from \mbn{d} to $c$, such that all the child/descendant relationships between main branch nodes are accordingly translated in
the ordering of $c$. Let $c'$ denote the code obtained from $c$ by: Step 1) replacing by the empty string all the positions that are not the image of some node $n' \in \mbn{d}$ under $e\circ id$, Step 2) replacing any sequence of consecutive /-characters of length more than 2 (i.e.,``///\dots'') by the slash-slash sequence (i.e, ``//'').

We can now construct the interleaving $p_i$ and its embedding $e_i$, such that $(\droot{t},n) \in e_i(t)$.

Let us book keep by a partial function $f_c$ the correspondence
between used $c$ positions and $c'$ positions. Let $p_i$ be defined by
the code $i=c'$, and let $f_i$ be defined by $f_c \circ id \circ e$ on
all the nodes in \mbn{d}. It is easy to see that $i$ and $f_i$ give
indeed an interleaving $p_i$, as it obeys all the conditions and $p_i$
is minimal. Let $id''$ denote the one-to-one mapping from \mb{p_i}
into $c'$. Now, we can define its embedding $e_i$ into $t$ as follows:
for all main branch nodes $n' \in \mb{p_i}$ we have
$e_i(n')=id^{-1}\circ f_c^{-1} \circ id''$. It is easy to see that for
any node $n'' \in MBN(d)$ such that $n_i'=id''^{-1}(f_i(n''))$, we
have $e(n'')= e _i(n')$ so all the predicate subtrees in $p_i$ can be
mapped at $e_i(n')$ for all $n'$.

Since $t$ and $n$ were chosen at random, this concludes the proof of containment for $d \sqsubseteq \bigcup_i p_i$.
\end{proof}

The following also hold:\footnote{This is reminiscent of similar results from relational database theory, on comparing  conjunctive queries with unions of conjunctive queries.}
\begin{lemma}\label{lem:equiv_tree_union}
If a tree pattern is equivalent to a union of tree patterns, then it is equivalent to a member of the union.
\end{lemma}
\vspace{-3mm}
\begin{lemma}\label{l:cnt-union-tp}
Let $p = \cup_i p_i$ and $q = \cup_j q_j$ be two finite unions of tree patterns. Then $p \sqsubseteq q$ iff $\forall i, \exists j$ s.t. $p_i
\sqsubseteq q_j$.
\end{lemma}

Given an DAG pattern  $d$, by the \emph{normal form of $d$} (in short, \nf{d}) we denote the equivalent formulation of $d$ as the union of incomparable interleavings with respect to containment.

Note that the set of interleavings $p_i$ of a DAG pattern $p$ can be exponentially larger than $p$. Indeed, it was shown that the \xpcap fragment is not included in \xp (i.e, the union of its interleavings cannot always be reduced to one \xp query by eliminating interleavings contained in others) and that a DAG pattern may only be
translatable into a union of exponentially many tree patterns (\cite{BenediktTCS05}). Nevertheless, testing if a DAG is satisfiable can be done in polynomial time. 

\begin{definition}
A DAG pattern is \emph{union-free} iff it is equivalent to a single tree pattern. \end{definition}

By Lemmas~\ref{lem:cap_sub_cup} and~\ref{lem:equiv_tree_union}, a satisfiable pattern is union-free iff it has an interleaving that contains all other possible interleavings. A naive, exponential-time procedure to test union-freedom would thus be to generate all possible interleavings and to check whether one of them contains all others. 

\subsection{Additional notation}
A \emph{/-pattern} is a tree pattern that has only /-edges in the main branch. We call
\emph{predicate subtree} of a pattern $p$ any subtree of $p$ rooted at a non-main branch node. By a \emph{/-subpredicate} $st$ we denote a predicate subtree whose root is connected by a /-path to the main branch node to which $st$ is associated.
A \emph{//-predicate} is a predicate subtree connected by a //-edge to the main branch.  A \emph{tree skeleton} is a tree pattern without //-edges in predicate subtrees.

A \emph{prefix} $p$ of a tree pattern $q$ is any tree pattern with $\droot{p} = \droot{q}$, $m=\mb{p}$ a subpath of \mb{q} and having all the predicates attached to the nodes of $m$ in $q$. For instance, the pattern shown in Figure~\ref{fig:runexample}(c) is a prefix of the pattern of $q_2$, since it has all the nodes of $q_2$, except for the output one.

A \emph{lossless prefix} $p$ of a tree pattern $q$ is any tree pattern obtained from $q$ by setting the output node to some other main branch node (i.e., an ancestor of \out{q}). Note that this means that the rest of the main branch becomes a side branch, hence a predicate.

For a pattern $d$ and node $n \in \mbn{d}$, by \sub{d}{n} we denote the subpattern rooted at $n$ in $d$.

The $\func{compensate}$ function generalizes the concatenation operation from~\cite{DBLP:conf/vldb/XuO05}, by copying extra navigation from the query into the rewrite plan. For $r \in \xpcap$ and a tree pattern $p$, $\func{compensate}(r,p,n)$ returns the query obtained by deleting the first symbol from $x\!=\!\xpath{\sub{p}{n}}$ and concatenating the rest to $r$. For instance, the result of compensating $r$ = \snippet{a/b} with $x$ = \snippet{b[c][d]/e} at the $b$-node is the concatenation of \snippet{a/b} and \snippet{[c][d]/e}, i.e. \snippet{a/b[c][d]/e}.

We also refer to the \emph{tokens} of tree pattern $p$: more specifically, the main branch of a tree pattern $p$ can be partitioned by its sub-sequences separated by //-edges, and each /-pattern from this partitioning is called a \emph{token}. We can thus see a pattern $p$ as a sequence of tokens (/-patterns) $p = t_1//t_2//\dots//t_k$. We
call $t_1$, the token starting with \droot{p}, the \emph{root token} of $p$. The token $t_k$, which ends by \out{p}, is called the \emph{result token} of $p$. The other tokens are denoted \emph{intermediary tokens}, and by the \emph{intermediary part} of a
tree pattern we denote the sequence of intermediary tokens. Note that a tree pattern may have only one token, if it does not have //-edges in the main branch. By a \emph{token-suffix} of $p$ we denote any tree pattern defined by a suffix of the sequence of tokens $(t_1, \dots, t_k)$. Symmetrically, we introduce the notion of \emph{token-prefix} of $p$.

\subsection{The rewriting problem} 
Given a set of views ${\cal V}$, defined by \xp queries over a document $D$, by $D_{\cal V}$ we denote the set of view documents $\{doc(``v")| v \in {\cal V}\}$, in which the topmost element is labelled with the view name. Given a query $r \in \xpcap$ over the view documents $D_{\cal V}$, we define \unfold{r} as the \xpcap query obtained by replacing in $r$ each $\doc{``v"}/v$ with the definition of $v$.

We are now ready to describe the view-based rewriting problem. Given a query $q$ and a finite set of views ${\cal V}$ over $D$ in a language $\cal L \subseteq \xp$,
we look for an alternative plan $r$, called a \emph{rewriting}, that can be used to answer $q$. We define rewritings as follows:

\begin{definition}\label{def:equiv-rw}
For a given document $D$, an \xp query $q$ and \xp views ${\cal V}$ over $D$, a \emph{rewrite plan} of $q$ using ${\cal V}$ is a query $r \in \xpcap$ over $D_{\cal V}$. If $\unfold{r} \equiv q$, then we also say r is a \emph{rewriting}.
\end{definition}

According to the definition above and the definition of \xpcap, a rewriting $r$ is of the form ${\cal I} = (\bigcap_{i,j} u_{ij})$, ${\cal I}/\rpath$ or ${\cal I}//\rpath$, with $u_{ij}$ of the form $\doc{``v_j"}/v_j/p_i$ or $\doc{``v_j"}/v_j//p_i$.

\begin{lemma}\label{l:rwplan-eval}
A rewrite plan from \xpcap can be evaluated over a set of view documents $D_{\cal V}$ in polynomial time in the size of $D_{\cal V}$.
\end{lemma}
\vspace{-3mm}

\begin{proof}[Sketch]
Consider a plan $r$ over a set of view documents $D_{\cal V}$.  $r$ 
gives a tractable evaluation strategy that: start from the
document nodes and navigate from each of them down to the
intersection node. All navigations can be done in PTIME, as they can
be seen equivalently as tree patterns. We can prove  by induction
on the structure of $r$ that the input of each intersection node
is polynomial, hence its input is also polynomial, because the
result is always a set (arity is 1), and it has at most as many
elements as the largest of its inputs. Hence the size of each
intermediate result is bounded by the size of the largest view.
As  the number of steps, navigation and intersections is constant
w.r.t. $t$ (it is proportional to the size of $r$), the overall
computation is in PTIME in $|D_{\cal V}|$.
\end{proof}
\vspace{-1mm}
\textbf{Completeness.} Hereafter, an algorithm is said to be \emph{complete for rewriting} $\cal L \subseteq \xp$ if it solves the rewriting problem for queries and views in $\cal L$, i.e., it finds a rewriting whenever one exists.

%%%% RESTRUCTURE BY THE FOLLOWING PLAN : 

% STRAWMAN NAIVE APPROACH (NOT DETAILED) WHICH SEARCHES THROUGH AT LEAST EXP MANY VIEW COMBINATIONS; NO FURTHER DETAILING OF THE SEARCH SPACE
% BIG IDEA NUMBER 1: LINEAR SEARCH SPACE SUFFICES ; THAT IS REWRITE W/O CALL T APPLY RULES
% BIG IDEA NUMBER 2: RESTRICT INTERLEAVINGS IN THE UNION; THIS IS WHERE WE INTRODUCE APPLY-RULES. 

\section{The Rewriting Algorithm}\label{sec:algo}
Our approach for testing the existence of a rewriting (algorithm
$\algoname$) is the following: for each rewrite plan $r$ using views that satisfies certain conditions w.r.t the query $q$, we test whether its unfolding is equivalent to $q$. 
A remarkable feature of the algorithm is that it considers only \emph{a linear number of  candidate plans}. Indeed, we show in Section \ref{sec:guarantees} that if a rewriting exists then one of these candidate plans is a rewriting, which implies the completeness of our algorithm. This result is a pleasant surprise, given that intuitively one would expect the number of distinct plans to be inspected to reflect the number of distinct subsets of views.

\emph{Testing equivalence between the tree pattern $q$ and a DAG pattern $d$ corresponding to the unfolding of $r$ will be the central task in our algorithm}. As by construction the plans / DAGs to be considered will always contain $q$, testing equivalence will amount to testing the opposite containment, of $d$ into $q$.

However, Lemmas~\ref{lem:cap_sub_cup} and~\ref{lem:equiv_tree_union} imply that equivalence holds iff $d=\unfold{r}$ has an interleaving $p_i$ such that $d \equiv p_i \equiv q$. From this observation, a na\"ive approach for the rewrite test would be to simply compute the interleavings of $d$ (a union of interleavings), check that this union reduces  by containments to one interleaving $p_i$ (union-freedom), and that $p_i$ is equivalent to $q$. The reason we call this approach na\"ive is that the number of interleavings in the union can be exponential (recall Section~\ref{sec:preliminaries}), even when the unfolding of $r$ is equivalent to a single tree pattern. In this case, one ``dominant'' interleaving in the union will contain all others. %(This procedure will indeed be used as  the baseline algorithm in our experiments.) 

To avoid the cost of the na\"ive approach in these cases, we set out to directly detect the dominant interleaving prior to checking equivalence.  We  devise an algorithm, \proc{Apply-Rules}, that operates a series of transformations on the candidate rewrite plans, expressed in the form of nine \textit{rewrite rules}. Starting from $d$, each rule application will produce \textit{an equivalent DAG pattern that is one step closer to the dominant interleaving that contains all others, if one such interleaving exists,  i.e., if $d$ is union-free.}

Our rule-based rewriting algorithm applies to any DAG patterns and is a decision procedure for union-freedom under practically relevant restrictions. More precisely, we show in Section~\ref{sec:guarantees} that under the restrictions,   $\proc{Apply-Rules}(d)$ is a tree whenever $d$ is union-free. In general, this is not guaranteed and additional containment tests between the remaining possible interleavings may be necessary to find one $p_i$ such that $p_i \equiv d$. Nevertheless, using $\proc{Apply-Rules}$ can be beneficial even in the general case, by reducing the number of interleavings we have to check.

We give below the global form of our rule-based algorithm. Section \ref{sec:rules} will be dedicated to the detailed description of each of the nine rules, showing that they preserve equivalence. We will discuss several possible optimizations and strategies for triggering rewrite rules in Section~\ref{sec:implementation}. Section \ref{sec:guarantees} shows PTIME complexity for \proc{Apply-Rules}.

\begin{codebox}
\Procname{$\proc{Apply-Rules}(d)$}
\li \Repeat
\li \Repeat apply R1 to $d$
\li \Until no change
\li \Repeat apply R2-R9 to $d$, in arbitrary order
\li \Until no change
\li \Until no change
\end{codebox}

%%%%% CONSIDER CHANGING THE FLOW, changing APPLY-Rules into Reduce-Interleavings or similar

We use \proc{Apply-Rules} in the \proc{Rewrite} algorithm, that rewrites $q$ using views $\cal V$:

\begin{codebox}
\Procname{$\algoname(q,\cal V)$}
\li \label{li:init-prefs}$\id{Prefs} \gets \{ (p,\{(v_i,b_i)\}) \;|\; v_i \in {\cal V},
  p \textrm{ a lossless prefix of } q, b_i \in \mb{p}, $
\zi $\exists \textrm{ a root-mapping } h \textrm{ from } u_i\!=\!\pattern{v_i}
    \textrm{ into } p, h(\out{u_i}) = b_i \}$
\li\label{li:mainloop} \For $(p,W) \in \id{Prefs}$
\li\label{li:rw-build-compensations} \Do let ${\cal V'} \gets \{\func{compensate}(\doc{``v"}/v,p,b) \;|\; (v,b) \in W \}$
\li\label{li:build-plan}let $r$ be the \xpcap query $\left(\bigcap_{v_j \in {\cal V'}} v_j\right)$
\li let $d$ be the DAG corresponding to \unfold{r}
\li\label{li:apply-rules}\proc{Apply-Rules}(d)
\li\label{li:rw-cnt-check} \If $d \sqsubseteq p$
\li\label{li:ret-rw} \Then \Return $\func{compensate}(r,q,\out{p})$ 
\End
\End
\li\label{li:ret-fail} \Return {\bf fail}
\end{codebox}

\algoname starts the construction of rewrite-plan candidates by collecting the sets of relevant compensated views w.r.t. the input query or lossless prefixes thereof. For each prefix $p$ in separation, all possible compensated views are combined in the intersection step. The resulting DAG pattern is then tested for equivalence w.r.t. $p$, and if this holds (i.e., we have a rewriting for $p$) this prefix is compensated once more to obtain  a rewriting for $q$. At line~\ref{li:ret-rw}, if $p$ is $q$ itself, \func{compensate} returns just $r$,  as all needed navigation had already been added at line~\ref{li:rw-build-compensations}. Note  that $d$ is, in all cases, satisfiable, because we intersect views that contain a satisfiable query. Note also that, while the output of \proc{Apply-Rules} may be an arbitrary DAG, the algorithm always returns the initial DAG (plus some compensation), thus ensuring straight-forward conversion towards an \xpcap expression.

As an extension to \algoname , \allrw searches for all the rewritings of $q$ using views $\cal V$:

\allrw -- same code as \algoname with the modifications:
\begin{itemize}
\item replace line~\ref{li:mainloop} with: ($\ref{li:mainloop}'$) \For $(p,U) \in \id{Prefs}$ \For $W \subseteq U$
\item remove line~\ref{li:ret-fail},
\item  continue to run even when the return at line~\ref{li:ret-rw} is reached.
\end{itemize}

While we will show in Theorem \ref{th:completeness} that \algoname is sound and complete for all queries and views in \xp, its complexity depends on that of the containment test on line \ref{li:rw-cnt-check}. While in general this containment test is hard, it becomes efficient if $d$ is a tree. We identify fairly permissible restrictions under which the resulting $d$ is always a tree (thus allowing the containment test in PTIME), and consider a specialized version of \algoname, as below:

\efficient -- same code as \algoname, with the following modification 
\begin{itemize}
\item  line~\ref{li:rw-cnt-check}  becomes: ($\ref{li:rw-cnt-check}'$)
\If $d$ is a tree $\Then \If\, d\, \sqsubseteq p$.
\end{itemize}

As mentioned above, the number of plans to be considered is linear, and both  \proc{Apply-Rules} and the containment test when $d$ is a tree have PTIME complexity, thus announcing overall polynomial complexity for \efficient. We indeed show in Section \ref{sec:guarantees} that \efficient always runs in PTIME. Moreover, we show that under fairly permissible and practically relevant restrictions, the resulting $d$ is always a tree, thus \efficient becomes sound and complete. 

\color{black}

\section{The Rewrite Rules (of subroutine \proc{Apply-Rules})}
\label{sec:rules}
We present in this section a set of rewrite rules, such that  each application of one of the rules  brings the DAG pattern one step closer to a tree pattern. We will prove that the result of \proc{Apply-Rules} is always
equivalent to the original DAG. This implies that \algoname gives a
sound algorithm for the rewriting problem, and we will show it is also
a decision procedure.

We present the rules
R1-R9 as pairs formed by a test condition, which checks if the rule is
applicable, and a graphical description, which shows how the rule
transforms the DAG. The left-hand side of the rule description will
match main branch nodes and paths in the DAG. If the matching nodes
and paths verify the test conditions, then the consequent
transformation is applied on them. Each transformation either 
\begin{itemize}
\item
collapses two main branch nodes $n_1$, $n_2$ into a new node $n_{1,2}$
(which inherits the predicate subtrees, incoming and outgoing main
branch edges), 
\item  removes some redundant main branch nodes and edges, or 
\item  appends a new predicate subtree below an existing main branch
node.
\end{itemize}
%
%Explanation for the more complicated conditions of rules:
%\begin{enumerate}
%\item [R1] Self-explanatory.
%\item [R2] Collapsing the nodes $\lambda_2$ and $\lambda_3$ leads to an unsatisfiable DAG pattern. \reminder{testing the condition...}
%\item  [R3] Collapsing the $\lambda$ node of $p_2$ with an ancestor of the $\lambda$ node of $p_{1,2}$ leads to an unsatisfiable DAG pattern.
%   \item [R4] Collapsing the $\lambda$ node of $p_4$ with a node of $p_1$ leads to an unsatisfiable DAG pattern.
%\item [R5] By adding $st$ on the $\lambda$ node of $p_1$ we only add inheritable //-subpredicates to the /-pattern $p_1$
%\item [R6] $p_2$ and $p_3$ have equivalent extended prefixes.
%\item [R7] Self-explanatory.
%\item [R8] Collapsing the $\lambda$ node of $p_3$ with any other node of $p_1$ or $p_2$ except the $\lambda$ node leads to an unsatisfiable DAG pattern.
%\end{enumerate}

\noindent \textbf{Graphical notation.} We use the following notation in the illustration of our rewrite rules: linear paths corresponding to part
of a main branch are designated in italic by the letter $p$, nodes are
designated by the letter $n$, the result of collapsing two nodes
$n_i$, $n_j$ will be denoted $n_{i,j}$, simple lines represent
/-edges, double lines represent //-edges, simple dotted lines
represent /-paths, and double dotted lines represent arbitrary paths
(may have both / and //). We only represent main branch nodes or paths
in the depiction of rules (predicates are omitted). Exception are rules R5 and R9, where we need to refer to a subtree predicate, respectively a /-subpredicate, by its \xp
expression $[Q]$. We refer to the tree pattern containing just a main
branch path $p$ simply by $p$, and to the tree pattern having $p$
as main branch by \tp{d}{p}:  for a main branch path $p$ in $d$, given by a sequence of nodes $(n_1,
\dots, n_k)$, we define \tp{d}{p} as the tree pattern having $p$ as main
branch, $n_1$ as root and $n_k$ as output, plus all the
predicate subtrees (from $d$) of the nodes of $p$.
%All the paths (both simple and dashed lines) of the DAG fragments on which the rules apply
%are trees. Note that this does not preclude the existence of $\cap$ inside the depicted paths in the overall DAG pattern.
We represent by a rhombus main branch paths that are not followed by
any / (main branch) edge.  Paths include their end points.

%% In particular, a rhombus will indicate the end of the
%% maximal /-path incoming or outgoing from a given node.

\textbf{Test Conditions.}
In the test conditions, we say that a pattern $d$ is
\emph{immediately unsatisfiable} if by applying to saturation rule R1
on it we reach a pattern in which either there are two /-paths of
different lengths but with the same start and end node, or there is a
node with two incoming /-edges $\lambda_1/\lambda$ and
$\lambda_2/\lambda$, such that $\lambda_1 \neq \lambda_2$. Note
that the test of immediate unsatisfiability is just a sufficient
condition for the unsatisfiability of the entire DAG. %not a necessary one.
 For instance, a DAG pattern that has in parallel the branches \snippet{doc(``L'')//paper//section} and \snippet{doc(``L'')/book/section} is not satisfiable yet R1 does not apply on it.
%For a DAG pattern $d$ and node $n \in \mbn{d}$, by the
%\emph{subpattern rooted at} $n$ we denote the pattern \sub{d}{n}
%corresponding to the subgraph of $d$ formed by the main branch nodes
%that are descendant-or-self of $n$ and their predicate subtrees.
\eat{
For a main branch path $p$ in $d$, given by a sequence of nodes $(n_1,
\dots, n_k)$, we define \tp{d}{p} as the tree pattern having $p$ as main
branch, $n_1$ as root and $n_k$ as output, plus all the
predicate subtrees (from $d$) of the nodes of $p$.
%\vspace{-1ex}
}
\begin{definition}
\label{def:similar}
We say that two /-patterns $p_1$, $p_2$ are \emph{similar} if (a)
their main branches have the same code, and (b) both have root
mappings into any pattern $p_{12}$ built from $p_1$, $p_2$ as
follows:
\begin{enumerate}
\item
choose a code  $i_{12}$ and a total onto function $f_{12}$ that maps the nodes of
$m_{12} = \mbn{p_1} \cup \mbn{p_2}$ into $i_{12}$ such that:
\begin{enumerate}
\item
for any node $n$  in $m_{12}$, $\etiq{f_{12}(n)} =
    \etiq{n}$
\item
for any /-edge $(n_1,n_2)$ in the main branch of $p_1$ or
    $p_2$, the code $i_{12}$ contains $f_{12}(n_1)/f_{12}(n_2)$
\end{enumerate}
\item
build the minimal pattern $p_{12}$ such that:
\begin{enumerate}
\item
$i_{12}$ is a code for the main branch \mb{p_{12}},
\item
for each node $n$ in $\mbn{p_1}\,\cup\, \mbn{p_2}$ and its image $n'$ in \mb{p_{12}} (via
$f_{12}$), if a predicate subtree $st$ appears below $n$ then a copy of
$st$ appears below $n'$, connected by the same kind of edge.
\end{enumerate}
\end{enumerate}
\end{definition}
\begin{example} For instance,  the patterns $p_1=\snippet{a/b[.//c]/d[.//e]}$ and $p_2=\snippet{a[b//e]/b/d[.//c]}$ are similar, given that the  patterns $p_{12}$ that can be built from them, according to Definition~\ref{def:similar}, are of the form (before minimization) $p_{12}=\snippet{a[b//e]/b[.//c]/d[.//e][.//c]}$, $p_{12}= \snippet{a/b[.//c]/d[.//e] \dots a[b//e]/b/d[.//c]}$ or $p_{12}=\snippet{  a[b//e]/b/d[.//c] \dots a/b[.//c]/d[.//e]}$.
\end{example}

For two nodes $n_1, n_2 \in \mbn{d}$, such that $\lambda_d(n_1) =
\lambda_d(n_2)=\lambda$, by \collapse{d}{n_1, n_2} we denote the DAG
obtained from $d$ by replacing $n_1$ and $n_2$ with a
$\lambda$-labeled node $n_{1,2}$ that inherits the incoming and
outgoing edges of both $n_1$ and $n_2$. We say that two nodes $n_1$,
$n_2$ are \emph{collapsible} iff they have the same label and the DAG
pattern \collapse{d}{n_1, n_2} is not immediately unsatisfiable.

%\begin{proof} Since we are dealing with /-patterns, the number of constructions $p_{12}$ is linear in the size of \mbn{p_1} and \mbn{p_2}. For each of them we need to test the existence of two containment mappings.
%\end{proof}

%\begin{lemma}
%Two /-patterns that have the same main branch are similar only if they have equivalent extended prefixes.
%\end{lemma}

We have now all the ingredients to present the rewrite rules. With each rule presentation we will also prove soundness, i.e., that equivalence is preserved.  We thus have the following result:
\begin{proposition}%%[Soundness]
\label{l:step-soundness}
  The application of any of the rules from the set R1-R9
  on a DAG $d$ produces another DAG $d'$ such that $d' \equiv d$.
\end{proposition}
%After introducing each rule, we will prove that it is sound. 

We use the following schema for the soundness proofs. Each rule $r$ in our set has an associated function $f_r$ that takes a
DAG $d$ as input and outputs another DAG $f_r(d)$ that is the result
of applying $r$ to $d$. By the way rules transform $d$, the containment $f_r(d) \sqsubseteq d$ is immediate.  We will thus discuss why $d \sqsubseteq f_r(d)$ holds after any  rule $r$ triggers, proving in fact the following lemma. 
\begin{lemma}\label{lemma:dintofr}
  For a rule $r$, a DAG $d$ and a document $t$, if $d$ has an
  embedding $e$ in $t$ and $r$ is applicable to $d$, then $f_r(d)$ has
  also an embedding $e'$ into $t$ such that
  $e(\out{d})=e'(\out{f_r(d)})$.
\end{lemma}

\eat{
We divided the proof in two parts:
Lemma~\ref{lemma:frintod} gives the containment $f_r(d) \sqsubseteq d$
and Lemma~\ref{lemma:dintofr} $d \sqsubseteq f_r(d)$.

\begin{lemma}\label{lemma:frintod}
  For every rule $r$, $f_r$ is a containment mapping.
\end{lemma}
%% \begin{proof}
%%   seems to be true ....
%% \end{proof}
}

\textbf{Remark 1.} Lemma~\ref{lemma:dintofr} implies that
$d \sqsubseteq \f{r}{d}$ and since the opposite containment mapping trivially holds, Proposition~\ref{l:step-soundness} follows immediately (i.e., $d \equiv \f{r}{d}$). 
%\setlength{\pltopsep}{0mm}
%\setlength{\plitemsep}{0mm}

%\newpage
%\setlength{\pltopsep}{0.2\baselineskip}
%\setlength{\plitemsep}{0.2\baselineskip}

\subsection{Rule R1}
 This rule triggers when $\lambda_d(n_1) = \lambda_d(n_2)$
%% and either $\tst{n_1} = \tst{n_2}$ or $\tst{n_1}=\epsilon$ or $\tst{n_2}=\epsilon$.
%merges two main branch nodes that have the same
%label and are parents (resp. children) of the same node.
%The resulting node inherits the label, predicate subtrees, incoming and outgoing main branch edges.

\begin{center}
\includegraphics[trim=0mm 163mm 60mm 2mm, clip=true, scale=0.45]{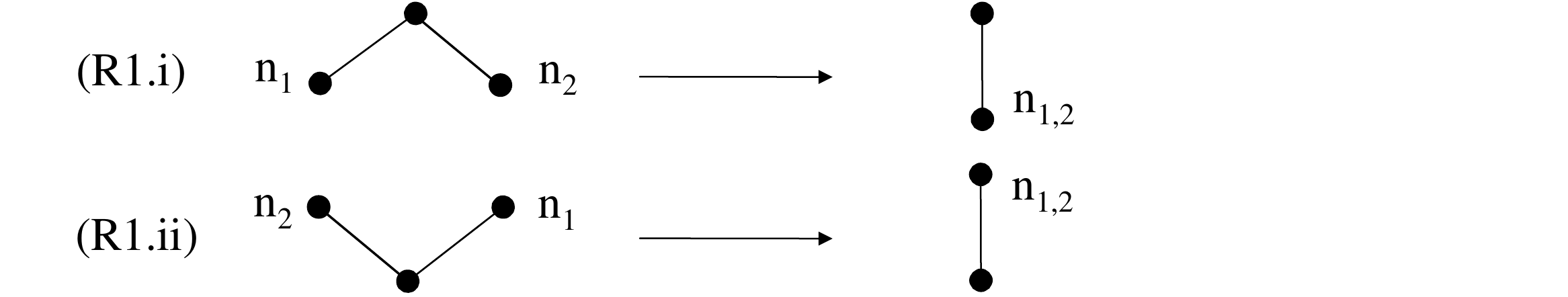}
\end{center}
\begin{example}
%\ioana:{For instance, the DAG obtained by intersecting the views doc/a//b/d and doc/a//c[.//e[f]]/d will be subject to R1 application with n1 and n2 being a-labeled nodes.}
The DAG pattern that would be obtained by intersecting some two views $\snippet{doc(``L'')/paper//\dots}$ and  $\snippet{doc(``L'')/paper/ \dots}$ would be subject to R1's application,
with $n_1$ and $n_2$ being its two nodes labeled $\snippet{paper}$. 
\end{example}
\begin{proof}[for Lemma\ref{lemma:dintofr} -  soundness of R1]
$n_1$ and $n_2$ belong to two different main branches,
but they have a common parent $n$. (Remember that all paths depicted in
the rules are part of main branches.)
Remember also that, by the definition of a main branch, the branches
of $n/n_1$ and $n/n_2$ have at least one common node below $n$: \out{d}.
Then, in any embedding $e$ of $d$ into a tree $t$,
$n/n_1$ and $n/n_2$ need to map in the same path of $t$
and it must be true that $e(n_1) = e(n_2) = x$, where $x \in \nodes{t}$.
Thus there is also an embedding $e'$ from \f{R1}{d} into $t$
that maps $n_{1,2}$ into $x$ and is equal to $e$ on all the other nodes.
\end{proof}

\subsection{Rule R2}
This rule triggers if $n_1$ and $n_2$ are not collapsible
and $n_2$ is not reachable from $n_1$ (resp. $n_1$ is not reachable
from $n_2$, in the case of R2.ii). %Furthermore, if
%in the case of R2.i, $n_3$ is reachable from $n_2$ in $d$,
%or, in the case of R2.ii, $n_2$ is reachable from $n_3$ in $d$,
%we do not introduce the //-edge between $n_2$ and $n_3$.

%%\begin{figure}[h!]
%%\vspace{-0.2cm}
\begin{center}
\includegraphics[trim=0mm 144mm 60mm 3mm, clip=true, scale=0.45]{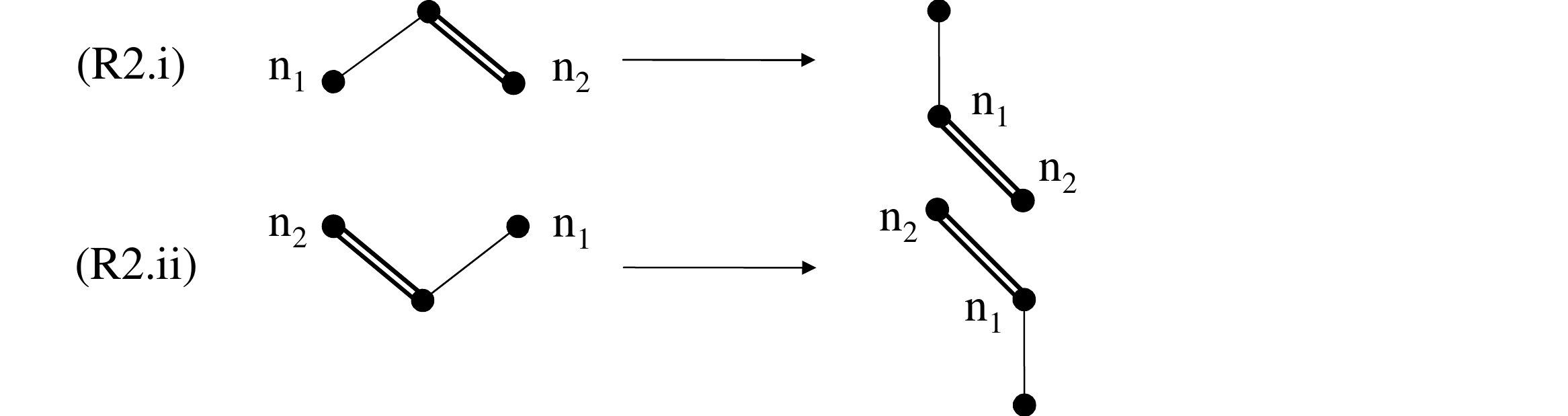}
\end{center}
%\caption{Rule R2.}
%%\vspace{-0.3cm}
%%\end{figure}
\begin{example}
%\ioana:{For instance, the DAG obtained by intersecting the views doc/b/d and doc//c[.//e]/d will be subject to R2.i application with n1 being the b-labeled node and n2 being the c-labeled node.}
Notice the application of rule R2.i in our running example (Figure~\ref{fig:runexample}), with $n_1$ being the node  labeled $\snippet{lib}$ and $n_2$ being the node labeled $\snippet{paper}$ in the left branch of the DAG pattern. Symmetrically, rule R2.ii applies with     $n_1$ being the node  labeled $\snippet{figure}$ and $n_2$ being the node labeled $\snippet{section}$ in the left branch of the DAG pattern. 
\end{example}
\begin{proof}[for Lemma\ref{lemma:dintofr} -  soundness of R2]
We first discuss R2.i. Let $n_0$ be the parent of $n_1$ and $n_2$. From the condition that $n_1$ and $n_2$ are not collapsible, we
infer that either they have different labels, or the pattern obtained
by trying to collapse $n_1$ and $n_2$ is immediately
unsatisfiable. Both cases imply that, for an embedding $e$ into $t$,
we cannot have $e(n_1)= e(n_2)$. The former case is obvious. For the
latter, supposing that $e(n_1)= e(n_2) = x$, we observe that the
beginning of main branches under $n_1$ and $n_2$, formed only by
/-edges, call them $p_{n1}$ and $p_{n2}$ respectively, need to map
into the same nodes under $x$ (as all main branches have at least one
common ending point, \out{d}, and we are mapping them into a tree).
Then the pattern obtained by collapsing $n_1$ and $n_2$ would also
have an embedding into $t$, that can be computed from $e$ by equating
nodes $n_1$ and $n_2$. But this contradicts the assumption that the
pattern obtained by trying to collapse $n_1$ and $n_2$ is immediately
unsatisfiable.  Hence, for an embedding $e$ into $t$, $e(n_1) \neq
e(n_2)$, and, since $e(n_1)$ has to be a child of $e(n_0)$, $e(n_2)$
has to be a strict descendant of $e(n_1)$. This guarantees that
$n_0/n_1//n_2$ will also map into $t$ if $d$ does. %Moreover, if there
%is already a path from $n_2$ to $n_3$ in $d$, there is also an
%isomorphic one in \f{R2}{d} (because, by our convention, all nodes and
%edges that are not drawn are copied, by \f{R2}{d},as they are), and
%$n_2//n_3$ is redundant.
%
%Let us also notice that if there is another sequence $s$ of edges from
%$n_2$ to $n_3$, there is also an isomorphic one in \f{R2}{d} (because,
%by our convention, all nodes and edges that are not drawn are copied,
%by \f{R2}{d},as they are), and $n_2//n_3$ does not need to be
%introduced. The reason is that, in any embedding $e$, the condition
%that $e(n_2)$ is an ancestor of $e(n_3$ will be guaranteed by the fact
%that $s$ can be mapped through $e$.

The proof  for R2.ii is very similar to the one for R2.i.
\end{proof}

\subsection{Rule R3.i} This rule triggers if the following conditions hold:
%\vspace{-0.2cm}
\begin{itemize}
%\item $p_2$ denotes the /-path starting in $n_3$,
\item $p_1\equiv p_2$,
 \item each of $p_2$'s nodes has only one incoming main branch edge,
  \item  \tp{d}{p_2} contains \tp{d}{p_1}.
\end{itemize}
%\vspace{-0.2cm}
%%\begin{figure}[h!]
%%\vspace{-0.2cm}
\begin{center}
\includegraphics[trim=0mm 177mm 95mm 0mm, clip=true, scale=0.45]{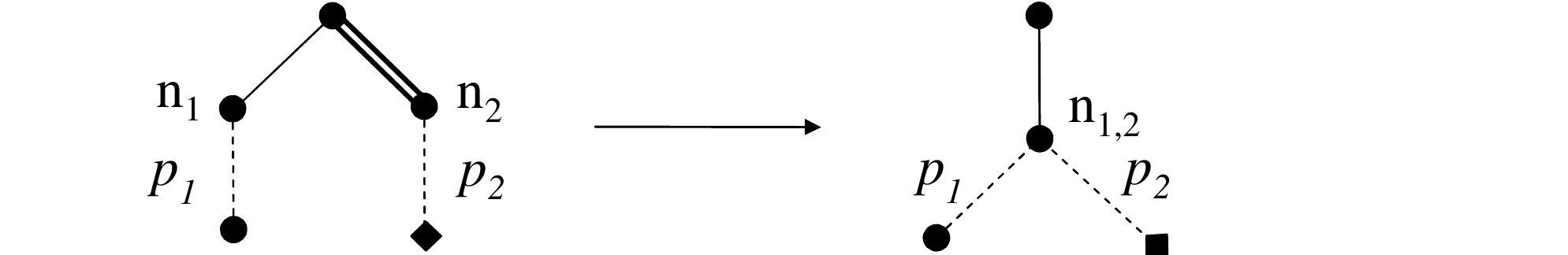}
%\caption{Rule R3.}
%%\vspace{-0.3cm}
%%\end{figure}
\end{center}
\begin{example}
%\ioana:{For instance, the DAG obtained by intersecting the views doc/a/b/c/e and doc//a/b/c//e will be subject to R3.i application with p1 being the path a/b/c coming from the first view and p2 being the path a/b/c coming from the second view.}
Notice the application of this rule in our running example (Figure~\ref{fig:runexample}), with $n_1$ and $n_2$ being the two nodes  labeled $\snippet{paper}$ and the paths $p_1$ and $p_2$ consisting of only these nodes.
\end{example}
\begin{proof}[for Lemma\ref{lemma:dintofr} -  soundness of R3.i]
Let $n_0$ be the parent of $n_1$ and $n_2$. For convenience, let us first rename $n_1$ by
$n_1'$ and $n_2$ by $n_1''$, let $p_1$ be defined by the sequence of
nodes $(n_1', \dots, n_k')$ and let $p_2$ be defined by the sequence of
nodes $(n_1'',\dots, n_k'')$. We know that $\lambda_d(n_i')=\lambda_d(n_i'')$, for each $i=1,k$.

Note that by applying R1 to saturation (after R3.i) all the pairs $(n_i', n_i'')$  will be collapsed. By $\f{R3i}{d}$ we denote
directly the result of R3i followed by these R1 steps. Let $n_{ii} \in \mbn{\f{R3i}{d}}$ denote the node that results from the collapsing of the $(n_i',n_i'')$ pair.

%% For a tree $t$ into which $d$ has an embedding $e$, either
%% $e(n_2) = e(n_3)$ or $e(n_3)$ is a descendant of $e(n_2)$.  Suppose
%% $e(n_2) = e(n_3)$. Then obviously \f{R3i}{d} has also an embedding
%% $e'$ into $d$, by taking $e'(n_{2,3}) = e(n_2)$. Suppose now that
%% $e(n_3)$ is a descendant of $e(n_2)$.
First, if the two main branches $n_0/p_1$ and $n_0//p_2$ contain the output
node, for any embedding $e$, $e(n_1')$ needs to be equal to $e(n_1'')$
and likewise, for each $i$ the image of $n_i'$ needs to be equal to the one of $n_i''$.
The reason is that the two branches have the
same endpoints ($n_0$ and \out{d}), the same length and, since $p_1$ has
no //-edge, $n_0/p_1$ needs to be isomorphic to the path
from $e(n_0)$ to $e(\out{d})$. Thus, it is obvious that by merging each
$n_i'$ and $n_i''$ we obtain a pattern that also has an embedding, if
$d$ does, and it maps \out{d} into the same node.

Let us now assume that $p_2$ ends above \out{d}. % and let $n_5$ be the last node of $p_2$. %$p_2$ can only be followed  by main branch //-edges starting at a node $n_5$ (the node where
%$p_2$ ends).
%
%Let $p'$ be the path containing the remainder of the
%main branch under $n_5$.
 Let $m$ be the function from $d$ into $d$ that maps
\tp{d}{p_2} into \tp{d}{p_1} and is the identity everywhere else (in particular, $m(p_2)=p_1$).  We can show that $e \circ m$ is another
embedding of $d$ into $t$, one that takes each pair $n_i', n_i''$ into the same image and preserves the image of the output. The main
reason is that $p_1$ and $p_2$ contain only /-edges, hence they can map only into
a sequence of /-edges in $t$. And since the branches containing
$p_1$ and $p_2$ respectively both start at $n_0$ and meet at \out{d} or at some other node above it, they have to map into the same path $p$ of $t$, from $e(n_0)$ to $e(\out{d})$.  So $e(n_1')$ is either above or equal to
$e(n_1'')$, because $n_1'$ is connected by a /-edge to its parent
$n_0$). But then all the nodes on $p_1$ map above or in the same place
in $p$ as the nodes of $p_2$; in particular $e(n_k')$ is above the
image of any main branch node $n$ that is //-child of $n_k''$. Therefore $e \circ m$ satisfies the
condition imposed by the //-edge between $n_k''$ and $n$. All the other
conditions for showing $e\circ m$ is an embedding follow directly from
the fact $e$ is an embedding. The image of \out{d} is the same in $e
\circ m$ and $e$, since \out{d} is not part of $p_2$.

We argue now that
 $e \circ m$ is also an embedding for the DAG $d'$ obtained from $d$ as follows: (a) for each $i$, append the predicate subtrees of $n_i''$ below $n_i'$, (b)  remove the edge $n_0/n_1''$ and the tree pattern $\tp{d}{p_2}$, and (c) connect the dangling incoming //-edges of children of $n_k''$ to $n_k$.  (By the test conditions these must be the only dangling edges.)

But \f{R3i}{d} has a straightforward mapping $h$ into $d'$, as $h = \{n_{ii} \mapsto n_i'; x\mapsto x \textrm{ elsewhere } \}$, hence we obtain the desired embedding $e'$ as $h \circ e \circ m$, with $e'(\out{\f{R3i}{d}}) = e(\out{d})$.
\end{proof}

\subsection{Rule R3.ii} This rule triggers if the following conditions hold:
%\vspace{-0.2cm}
\begin{itemize}
%\item $p_2$ denotes the /-path starting in $n_3$,
\item $p_1\equiv p_2$,
 \item  each of $p_2$'s nodes has only one outgoing main branch edge,
  \item  \tp{d}{p_2} contains \tp{d}{p_1}.% such that the image of $n_2$ is $n_1$.
\end{itemize}
%\vspace{-0.2cm}
%%\begin{figure}[h!]
%%\vspace{-0.2cm}
\begin{center}
\includegraphics[trim=0mm 163mm 100mm 0mm, clip=true, scale=0.45]{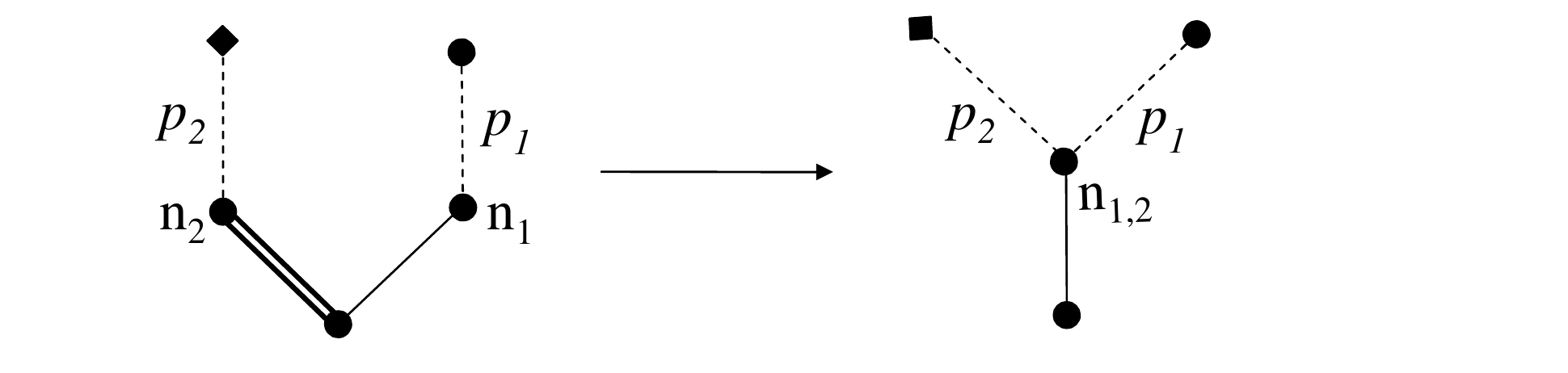}
%\caption{Rule R3.}
%%\end{figure}
\end{center}
%%\vspace{-0.4cm}

%\newpage

\begin{proof}[for Lemma\ref{lemma:dintofr} -  soundness of R3.ii]
We can use same argument as for R3.i.%, noticing that, in any
%embedding $e$, $n_3$ maps at or above $e(n_2)$, and, for a main branch
%$p_3$ under $n_4$, every node $x_i$ at pos $i$ in $p_2$ maps at or
%above $e(x_i)$.
\end{proof}

\subsection{Rule R4.i} The rule triggers if the following conditions hold \emph{for all nodes $n_4$}:

 \begin{itemize}
 \item $n_3$ has one incoming main branch edge, all other nodes of $p_2$ have one incoming and one outgoing main branch edge,

   \item  there exists a mapping from \tp{d}{p_2} into \sub{d}{n_1},  mapping all the nodes of $p_2$ into
   nodes of $p_1$.

  \item the path $p_2//n_4$ does not map into $p_1$.

\end{itemize}

%%\begin{figure}[h!]
%%\vspace{-0.2cm}
\begin{center}
\includegraphics[trim=0mm 152mm 100mm 0mm, clip=true, scale=0.45]{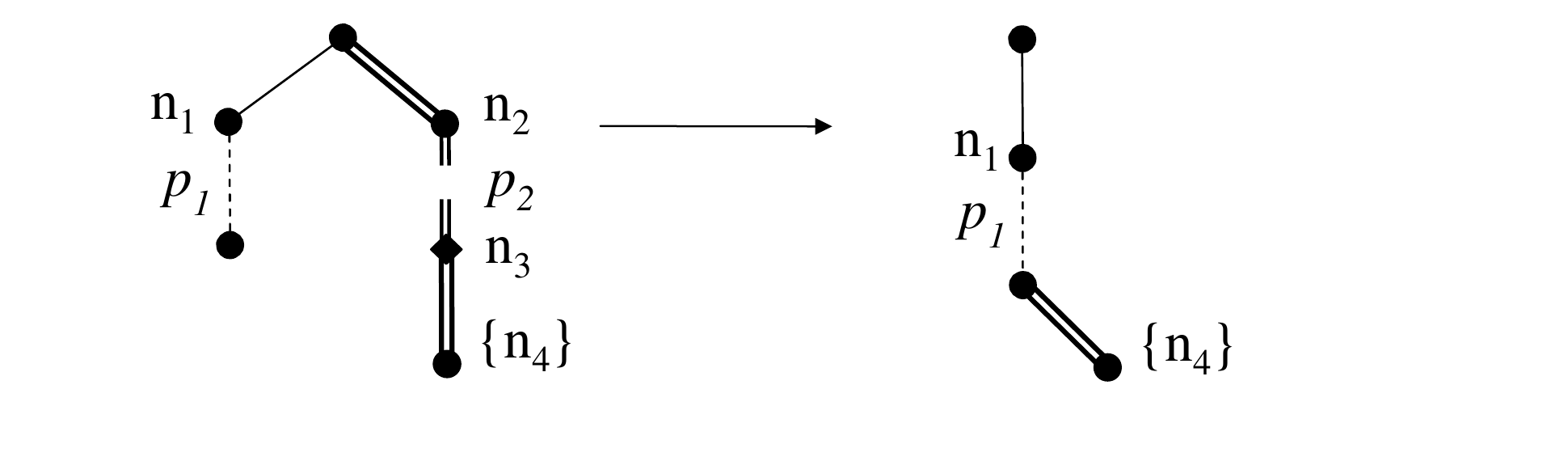}
%\caption{Rule R4.}
%%\vspace{-0.3cm}
%%\end{figure}
\end{center}

%\ioana:{For instance, the DAG obtained by intersecting the views doc/a/c/d and doc//a[.//d]//b//d  will be subject to R4.i application with p1 being the path a coming from the first view and p2 being the path a coming from the second view.}
\begin{example} The DAG pattern that would be obtained by intersecting some two views $\snippet{doc(``L'')/lib/paper/section/\dots/figure[caption]}$ and  $\snippet{doc(``L'')//lib[.//caption]//section//theorem//\dots}$ would be subject to R4.i's application,
with $p_1$ being the path corresponding to  $\snippet{lib/paper/section}$, $p_2$ being the path corresponding to $\snippet{lib//section}$, and $n_4$ being the node labeled $\snippet{theorem}$. 
\end{example}
\begin{proof}[for Lemma\ref{lemma:dintofr} -  soundness of R4.i]
 Let $n_0$ be the parent of $n_1$ and $n_2$. Suppose that $d$ has an embedding $e$ into a tree $t$.
$p_1$ and $p_2$ are parts of main
branches starting from $n_0$ and ending in a common node, at or above
\out{d}. Hence, if $d$ has an embedding $e$ into a tree $t$, the nodes
of $p_1$, of $p_2$ and each $n_4$ must all map into the same path $p$ of
$t$.  Moreover, since $n_0/p_1$ has only /-edges, it is necessarily
isomorphic to the fragment of $p$ starting at $n_0$ and of length
$|p_1| + 1$.

Let $n_5$ be the end node of $p_1$ and $n_3$ the end node of $p_2$.
With necessity, either there is a node $n'\in p_1$ such that $e(n_3) =
e(n')$ or $e(n_3)$ is below $e(n_5)$. In the former case, we can show
that for any $n_4$ there is no node $n''$ of $p_1$ such that $e(n'') =
e(n_4)$. For a given $n_4$, let us assume that such a node $n''$ exists. Since $p_1$ is isomorphic to $e(p_1)$, and the mapping of
$p_2//n_4$ through $e$ would imply also a mapping of $p_2//n_4$ into a
suffix of $p_1$, this leads to a contradiction. Since $n_0/p_1$ is isomorphic to the
beginning of $p$, it means that $e(n_4)$ is below $e(n_5)$. %In the later case it follows immediately that for each $n_4$ its image $e(n_4)$ is below $e(n_5)$.
 But then we can also map all nodes of \f{R4i}{d} exactly following $e$ because
$e$ verifies the condition imposed by the //-edge between $p_1$ and
$n_4$ and the part that was removed ($p_2$) was not connected to other
main branches. Moreover, the image of the output is the same, because
\out{d} is below $p_1$ and \f{R4i}{d} keeps all nodes under $p_1$
unchanged.

If $e(n_3)$ is below $e(n_5)$ in $t$, then, following the same
reasoning, we can argue that $e$ can be reused to map the nodes of
\f{R4i}{d} into $t$.
\end{proof}

\subsection{Rule R4.ii} This rule triggers if the following conditions hold\emph{ for all nodes $n_4$}:

 \begin{itemize}
 \item $n_3$ has only one outgoing main branch edge, all the other nodes of $p_2$  have one incoming and one outgoing main branch edge,

   \item  there exists a mapping from \tp{d}{p_2} into \tp{d}{p_1}, mapping all
   the nodes of $p_2$ into nodes of $p_1$.

  \item the path $n_4//p_2$ does not map into $p_1$.

\end{itemize}

\begin{center}
%%\begin{figure}[h!]
%%\vspace{-0.2cm}
\includegraphics[trim=0mm 149mm 90mm 2mm, clip=true, scale=0.45]{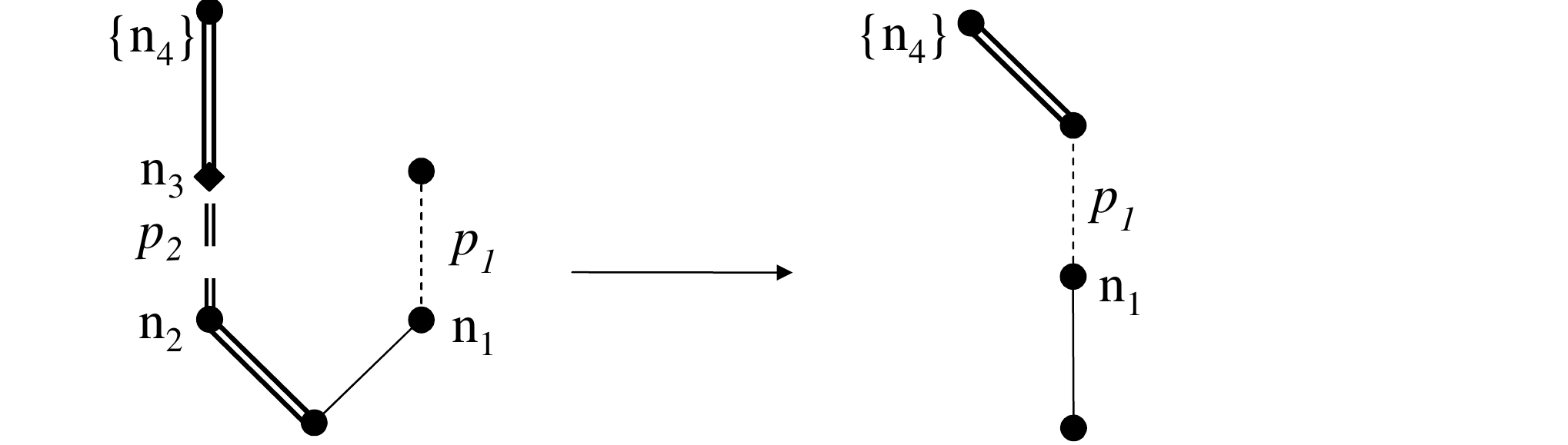}
%\caption{Rule R4.}
%%\vspace{-0.3cm}
%%\end{figure}
\end{center}

\begin{proof}[for Lemma\ref{lemma:dintofr} -  soundness of R4.ii]
Proof similar to R4.i: here the rule's test condition guarantees
that, in any embedding, the beginning of $p_1$ is mapped below
$n_4$.
\end{proof}

\subsection{Rule R5}  This rule triggers if the following conditions hold:

\begin{itemize}
\item  $n_2$ and $n_3$ are collapsible and $p_1 \equiv p_3$,

 \item \pattern{\lambda_d(n_2)[Q]} does not have a root-mapping into \sub{d}{n_2},

 \item for any node $n_4$ in $p_2$ such that
$d'=\collapse{d}{n_4, n_3}$ is not immediately unsatisfiable,
\pattern{\lambda_d(n_2)[Q]} has a root mapping into \sub{d'}{n_2},

\item if there is no path from $n_3$ to a node of $p_2$,
there has to be a root-mapping from
\pattern{\lambda_d(n_2)[Q]} into the pattern
obtained from $\tp{d}{p_2}$ by appending $[Q]$'s pattern, via a
//-edge, below the node \out{\tp{d}{p_2}}.
%
%there exists a mapping
% from the tree pattern associated with the \xp expression
% $\lambda[Q]$ into \sub{d'}{n_2}, where $\lambda$ denotes the label of $n_1$.
 %\item any viable tentative collapse of $n_3$ with a node of $p_1$
% would yield a DAG pattern $d(n_2)$ such that there exists a mapping
% from the tree pattern associated with the \xp expression
% $\lambda_d(n_3)[Q]$ into $d(n_2)$.
 \end{itemize}
%%\vspace{-0.2cm}
\indent ~~~~~ (Special case: $p_1$ and $p_3$ empty.)

\begin{center}
%%\begin{figure}[h!]
%%\vspace{-0.2cm}
\includegraphics[trim=20mm 158mm 95mm 1mm, clip=true, scale=0.44]{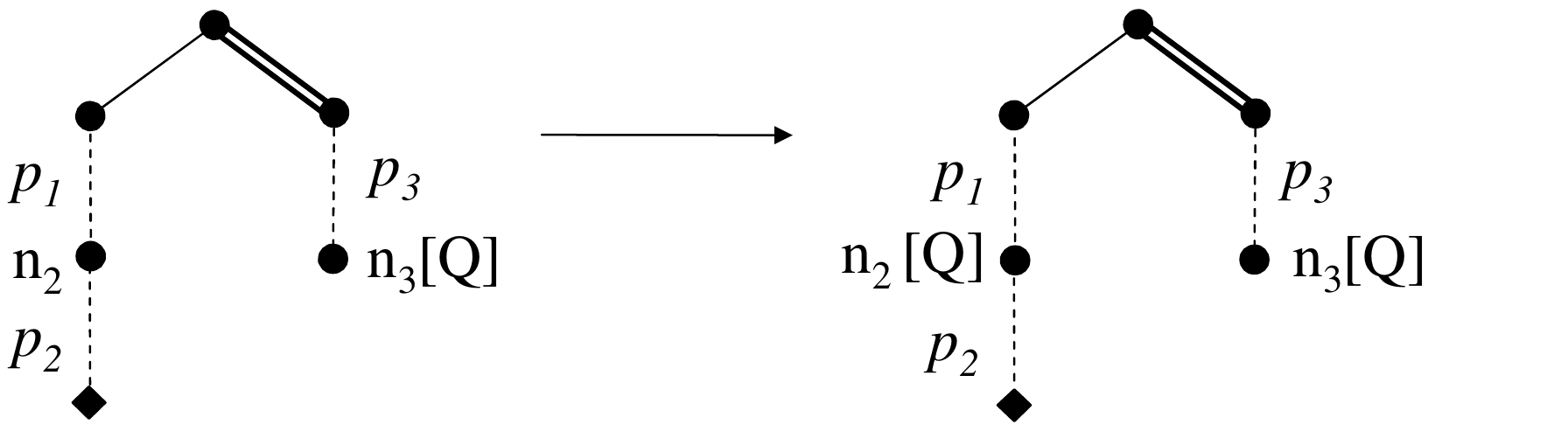}
%\caption{Rule R5.}
%%\vspace{-0.3cm}
%%\end{figure}
\end{center}

%\ioana:{For instance, the DAG obtained by intersecting the views doc/a/b/c and doc//a/b[.//f]/c  will be subject to R5 application, copying the [.//f] predicate on the b-labeled node coming from the first view.}
\begin{example}The DAG pattern that would be obtained by intersecting some two views $\snippet{doc(``L'')/lib/paper/section//\dots}$ and  $\snippet{doc(``L'')//paper[.//caption]//\dots}$ would be subject to R5's application,
copying the predicate $\snippet{[.//caption]}$  on the node labeled $\snippet{paper}$ coming from the former view.
\end{example}
\begin{proof}[for Lemma\ref{lemma:dintofr} -  soundness of R5]
Let $st$ denote the subtree predicate introduced by \f{R5}{d} on $n_2$,
and let $st'$ denote its copy under $n_3$. Given the embedding $e$ of
$d$ in $t$, there are three possible cases:

\begin{enumerate}
\item the image of $n_3$  is the same with the image of $n_2$,
\item the image of $n_3$  is the same with the image of some other node $n_4$ from $p_2$,
\item the image of $n_3$ is below the image of any node from $p_2$.
\end{enumerate}

In the first case, $e$ gives immediately an embedding $e'$ of
\f{R5}{d} in $t$, since $st'$ has already an image in the subtree
rooted at $e(n_3)$.

In the second case, we can first conclude that the DAG pattern
$d'=collapse_d(n_4,n_3)$ must not be immediately unsatisfiable. This
is because $e$ gives also an embedding $e''$ from $d'$ in $t$, one
that maps the node $n_{3,4}$ into $e(n_3)=e(n_4)$. Let us now consider
the tree pattern $p=\pattern{\lambda_d(n_2)[Q]}$, which modulo renaming is
the pattern formed by the main branch node $n_2$ and the predicate
subtree $st$. Since we know that $p$ maps into \sub{d'}{n_2}, by some
mapping $f$, we can also build a containment mapping $\widetilde{f}$
from \f{R5}{d} into $d'$, defined as follows: (a) $n_3$ and $n_4$
have the same image, $n_{3,4}$, (b) $\widetilde{f}=f(n)$ for all the
nodes $n$ of $st$ and (c) $\widetilde{f}$ is the identity mapping for
all the other nodes.  Finally, we obtain the embedding $e'$ as the
composition $e'' \circ \widetilde{f}$.

In the third case,
%% we use the fact that $[Q]$ would be an
%% \reminder{inheritable}???  predicate for $n_2$ in the /-pattern $p$
%% obtained as \tp{d}{p_2} augmented by attaching the subtree of $Q$
%% below $n_2$.
let us consider the DAG pattern $d'$ obtained from $d$ by appending the pattern of $Q$ below \out{\tp{d}{p_2}}, via a //-edge (as described in the rule condition). Note that since the image of $n_3$ under $e$ is below the image of any node from $p_2$, we can easily obtain from $e$ an embedding $e''$ from $d'$ into $t$.

Moreover, by the test condition, we have a root-mapping $f$ from \pattern{\lambda_d(n_2)[Q]} into the modified pattern, $\tp{d'}{p_2}$.  From $f$ we will construct a mapping $\widetilde{f}$ from \f{R5}{d} in $d'$.

Let \textit{st} be new subtree predicate in \f{R5}{d}, corresponding
to $[Q]$ at node $n_2$. We define $\widetilde{f}$ from \f{R5}{d} in
$d'$as follows: (a) $\widetilde{f}$ is the identity function for nodes
outside \textit{st} and (b) $\widetilde{f}(n)=f(n)$ for all the nodes $n
\in \textit{st}$.

Finally, by the composition $e'' \circ \widetilde{f}$ we obtain an
embedding of \f{R5}{d} into $t$.
\end{proof}

\subsection{Rule R6} This rule triggers if the following conditions hold:

\begin{itemize}
 \item $n_3, n_4$ have only one incoming main branch edge, all other
 nodes of $p_1$ and $p_2$ have one incoming and one outgoing main
 branch edge,
 %\item $p_1$ denotes the /-path starting in $n_2$ and $p_2$ denotes
 %the /-path starting in $n_3$,
\item  \tp{d}{p_1} and \tp{d}{p_2} are similar.

\end{itemize}

\begin{center}
%%\begin{figure}[h!]
%%\vspace{-0.2cm}
\includegraphics[trim=10mm 172mm 95mm 1mm, clip=true, scale=0.45]{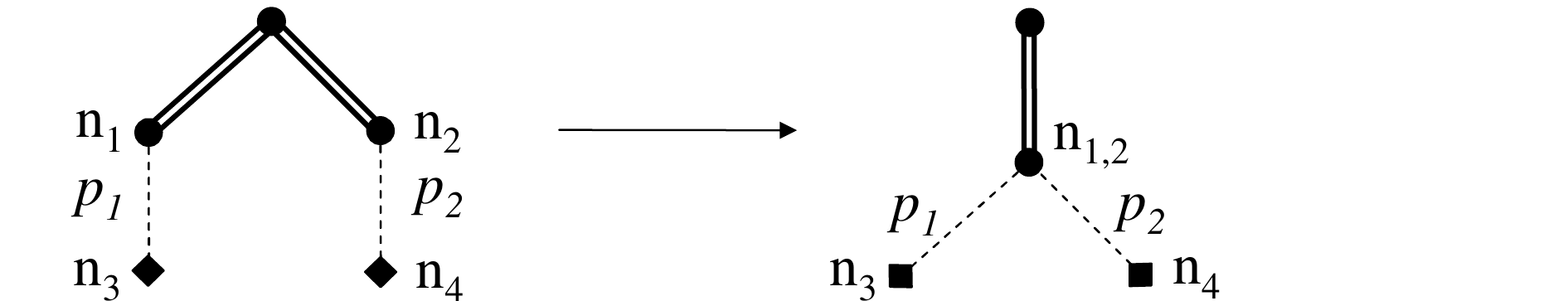}
%\caption{Rule R6.}
%%\vspace{-0.3cm}
%%\end{figure}
\end{center}

\begin{example}
%\ioana:{For instance, the DAG obtained by intersecting the views doc//a/b[.//d]/a/b//d  and doc//a/b/a/b//d will be subject to R6 application, with p1=a/b/a/b and p2=a/b/a/b}
The DAG pattern that would be obtained by intersecting some two views $\snippet{doc(``L'')//lib/paper[.//caption]/section//\dots}$ and  $\snippet{doc(``L'')//lib[.//figure]/paper/section//\dots}$ would be subject to R6's application,
with the paths $p_1$ and $p_2$ corresponding to the $\snippet{lib/paper/section}$ parts of the views.
\end{example}

\begin{proof}[for Lemma\ref{lemma:dintofr} -  soundness of R6]
Let $n_0$ be the parent of $n_1$ and $n_2$. For
convenience, let us first rename $n_1$ by $n_1'$ and $n_2$ by $n_1''$,
let $p_1$ be defined by the sequence of nodes $(n_1', \dots, n_k')$
and let $p_2$ be defined by the sequence of nodes $(n_1'',\dots,
n_k'')$. We know that $\lambda_d(n_i')=\lambda_d(n_i'')$, for each
$i=1,k$.

As mentioned, our rule-rewriting algorithm would behave in the same
way if R6 collapsed the entire $p_1$ and $p_2$ paths. To simplify the
presentation of the rules, we delegated to R1 this task. However, to
simplify the presentation of the proof, we will consider the pattern
$d''$, representing the result of applying R6 followed by these R1
steps.  Let $n_{ii} \in \mbn{d''}$ denote the node that results from
the collapsing of the $(n_i',n_i'')$ pairs. It is easy to see that
there is always a mapping $h$ from \f{R6}{d} into $d''$, given by the
composition of the applications of $f_{\rm R1}$ that merge all these
$(n_i',n_i'')$ pairs.  Hence, it is sufficient to show that for any
embedding $e$ of $d$ into a tree $t$, there is an embedding $e'$ of
$d''$ into $t$, which guarantees that $e' \circ h$ is an embedding of
\f{R6}{d} into $t$.

 Let us consider any embedding $e$ of
$d$ in a tree $t$. We have three possible cases:
\begin{enumerate}
\item the image of $n_1'$ is the same with the image of $n_1''$,
\item the image of $n_1'$ is above the image of $n_1''$, %in the path from $\droot{t}$ to $e(\out{d})$,
\item the image of $n_1''$ is above the image of $n_1'$. %in the path from $\droot{t}$ to $e(\out{d})$.
\end{enumerate}

In the first case, $e$ gives immediately an embedding $e'$ of
$d''$ into $t$, since the nodes of $p_1$ and $p_2$ have the same
images.

In the second case, let $e(p_1)$ denote the image of $p_1$ in
$t$.
%% Since $p_1$ and $p_2$ have equivalent extended prefixes, they
%% have in particular the same /-prefix and the same main branch.
We show that the function $e'$ from $d''$ into $t$ which: (a) for each $i$ maps $n_{ii}$ into $e(n_i')$, and (b)
maps $n$ into $e(n)$ for all the nodes $n$ outside $\tp{d}{n_{11}/\dots/n_{kk}}$,
%% (c) maps any main branch node $n_{i}'$ of $p_2$ into $e(n_i)$,
 can be extended to a full embedding of $d''$ into $t$. For that, we must show that all the predicate subtrees rooted at
$n_{ii}$ nodes can be mapped in the subtree rooted at $e(n_i')$. Since
predicate subtrees from $n_i'$ obviously have an image at
$e(n_i')$, $e'$ is defined by $e$ on these nodes.  What remains is to describe $e'$ over the predicates subtrees that originate at $n_i''$ nodes.

Since \tp{d}{p_1} and \tp{d}{p_2} are similar, we always have a
root-mapping from both /-patterns into any $p_{12}$, as given in
Definition~\ref{def:similar}. Note that the embedding $e$ (as any
other embedding of $d$ in general) imposes an order on the nodes of
$p_1$ and $p_2$, consistent with their /-edges. We can thus always
find an associated pattern $p_{12}$  that has an embedding in the
subtree rooted at $e(n_1')$ in $t$. Let the pattern fixed in this
way be $p$ and let his embedding into the subtree rooted at $e(n_1')$
be $e''$.  Let $f_2$ be the root-mapping of \tp{d}{p_2} into $p$ (by
definition, we can always find such an $f_2$).

We can now define $e'$ over the predicates subtrees of $n_{ii}$'s
that originate at $n_i''$'s using $e'' \circ f_2$.
\eat{
Since predicate subtrees from $n_i'$ obviously have an image at
$e(n_i')$, $e'$ is defined by $e$ on these nodes. Moreover, since
\tp{d}{p_1} and \tp{d}{p_2} are similar, by definition it means that
we can also map
}

By symmetry, the third case can be handled in the same manner.
\end{proof}

%\newpage
\subsection{Rule R7} This rule triggers if the following conditions hold: %\bogdan{Add nodes to figure, n1, n5, n3,n4}

\begin{itemize}

\item the nodes of $p_2$ have only one incoming and one outgoing main branch edge,
\item there exists a mapping from \tp{d}{p_2} into    $\sub{d}{n_1}$,  %\tp{d}{p_1} 
such that the nodes of $p_2$ are mapped into nodes of $p_1$.

% \item there exists a mapping $h$ from \tp{d}{p_2} into a subpattern
% $d(n_1)$ of $d$, for $n_1$ being a node in $p_1$, such that $h$ maps
% all the main branch nodes into nodes of $p_1$.

 \end{itemize}
%%\vspace{-0.2cm}
%(Special case: $p_2$ is a //-edge in parallel with $p_1$.)

\begin{center}
%% \begin{figure}[h!]
% \hspace{2cm}
%% \vspace{-0.2cm}
\includegraphics[trim=10mm 152mm 125mm 10mm, clip=true, scale=0.45]{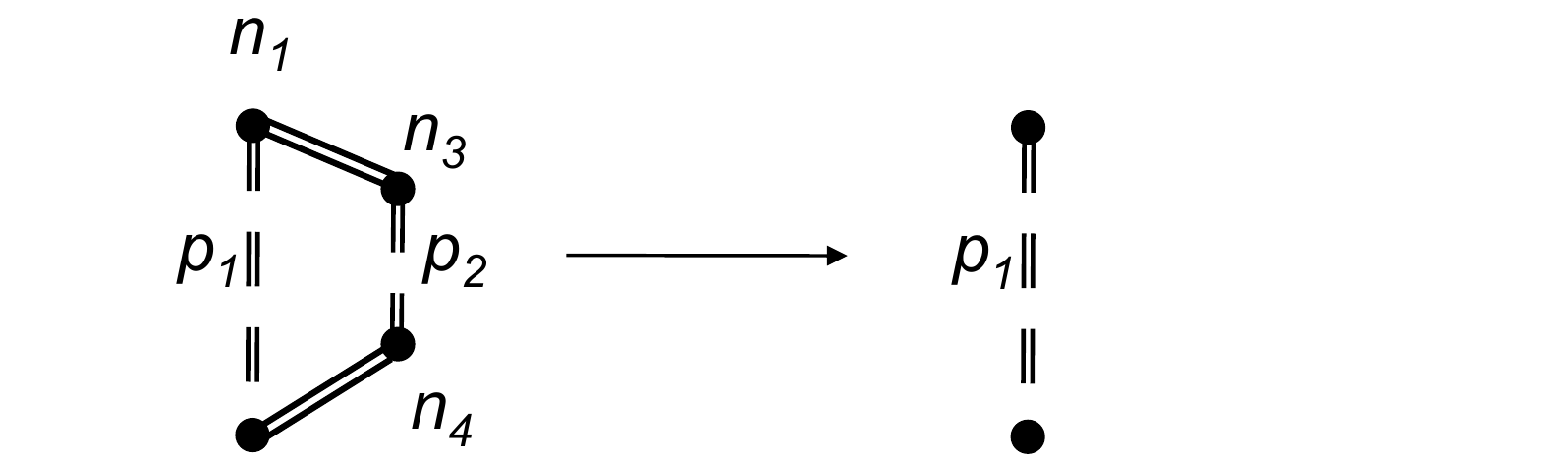}
%\caption{Rule R7.}
%%\vspace{-0.2cm}
%%\end{figure}
\end{center}

%\ioana:{For instance, the DAG obtained by intersecting the views doc/f//b/d/c and doc//b/d//c  will be subject to R7 application, with p1 being the first view, and p2 = b/d.}
\begin{example}Notice the application of this rule in our running example (Figure~\ref{fig:runexample}), with $p_1$ and $p_2$ corresponding to the two paths $\snippet{paper//section//figure}$  in parallel.
\end{example}
\begin{proof}[for Lemma\ref{lemma:dintofr} -  soundness of R7]
 \eat{ The case in which $p_2$ is empty is trivial: any
embedding having $n_1//n_2$ stays an embedding if $n_1//n_2$ is
removed, because the same condition is implied the path connecting
$n_1$ with $n_5$ through $n_2$. <-- needed to show \f{R7}{} is a
mapping.}
With $n_3$ and $n_4$ denoting the end points of $p_2$, 
$n_1$ denoting the common parent of $p_1$ and $p_2$, let $n_2$
denote the image of $n_3$ under the mapping from $\tp{d}{p_2}$ into
$\sub{d}{n_1}$.

Since $n_3$, $n_4$ and $p_2$ are not connected to any other parts of the
DAG, for any embedding $e$ into a tree $t$, the restrictions $e'$ of
$e$ to $\nodes{d}\, \backslash\, \{n | n \in p_2\}$ is a
partial embedding into $t$ (because $n_3$, $n_4$, $p_2$ do not
affect the conditions needed to embed the other nodes of $d$.  %are still satisfied). 
 But $e'$ is a total embedding for \f{R7}{d} and
$e'(\out{\f{R7}{d}}) = e'(\out{d}) = e(\out{d})$.
%\bogdan{The opposite containment seems to be unclear, $e=e' \circ h$ for nodes in \tp{d}{p_2}.}
\end{proof}

\subsection{Rule R8}This rule triggers if the following conditions hold:
\begin{itemize}
\item the nodes of $p_2$ have only one incoming and one outgoing main branch edge,
\item in any possible mapping of $p_{2}$ into $p_{1}$ the image of $n_2$ is $n_1$.
\end{itemize}

\begin{center}
%%\begin{figure}[h!]
%\hspace{1cm}
%%\vspace{-0.2cm}
\includegraphics[trim=0mm 162mm 98mm 1mm, clip=true, scale=0.45]{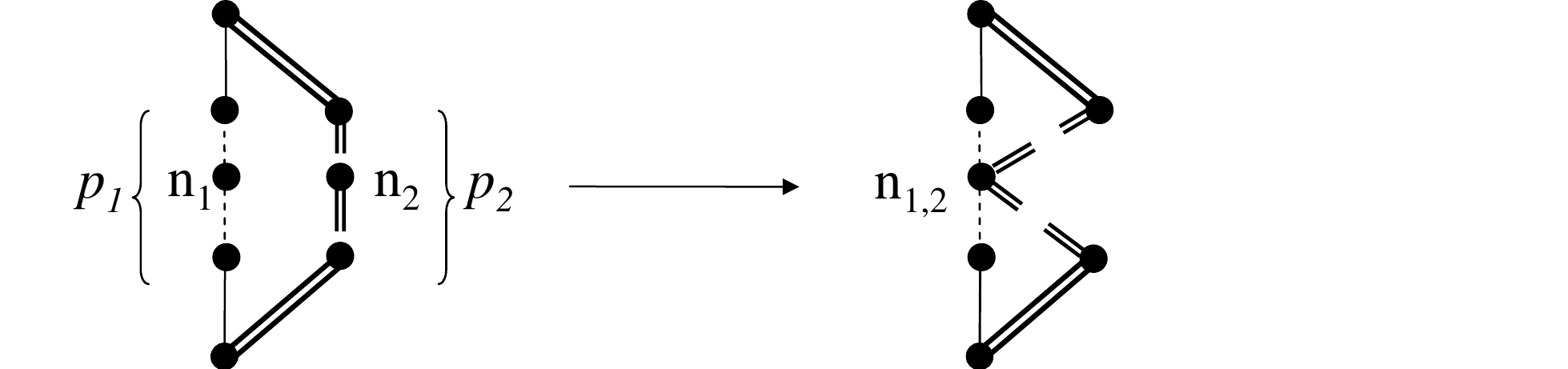}
%\caption{Rule R8.}
%%\vspace{-0.3cm}
%%\end{figure}
\end{center}

%\ioana:{For instance, the DAG obtained by intersecting the views doc/a/b/c/d and doc//a//c//d  will be subject to R8 application, with either a-labeled or c-labeled nodes being collapsed.}
\begin{example}The DAG pattern that would be obtained by intersecting some two views $\snippet{doc(``L'')/lib/paper/section/figure/image}$ and  $\snippet{doc(``L'')//paper[.//caption]//image}$ would be subject to R8's application,
with  $n_1$ and $n_2$ corresponding to  the two nodes labeled $\snippet{paper}$.
\end{example}
\begin{proof}[for Lemma\ref{lemma:dintofr} -  soundness of R8]
 Let $n_0$ denote the common parent of the two branches in parallel, and let $n_3$ denote their common child node. As the branches of $p_1$ and $p_2$ are fragments of main branches
between the nodes $n_0$ and $n_3$, and $n_0/p_1/n_3$ has only /-edges,
we can argue, as for R7, that for any embedding $e$ into a tree $t$,
$p_1$ and $e(p_1)$ are isomorphic. Also, $e(p_2) \subseteq e(p_1)$, as
the nodes of $e(p_2)$ must lie on the same path fragment, between
$e(n_0)$ and $e(n_3)$. Since $p_1 \rightarrow e(p_1)$ is an
isomorphism, there is an inverse mapping $i$ from $e(p_1)$ into $p_1$
which is surjective. But $e(p_1) \supset e(p_2)$, hence $g = i \circ e$
is a mapping from $p_2$ into a suffix of $p_1$. As we assumed the
conditions of R8 to be satisfied, it follows that
$g(n_2) = n_1 \Leftrightarrow e(n_2) = e(n_1)$, for any embedding $e$.
But then, if we take
$$e'(x) = \left\{
\begin{array}{l}
e(n_1) ( = e(n_2)), \textrm{ if } x = n_{1,2}\\
e(x), \textrm{ otherwise}
\end{array}
\right.$$
$e'$ will be an embedding for \f{R8}{d} with $e'(\out{\f{R8}{d}}) = e(\out{d})$.
\end{proof}

\subsection{Rule R9} For any  /-subpredicate $Q$ in $d$, any node $n$ in $p_1$ s.t. the presence of $Q$ as predicate at $n$ on $p_1$ would
 verify the condition of extended skeletons (see Section~\ref{sec:guarantees}), the rule triggers if
\begin{itemize}
\item  the nodes of $p_2$ have only one incoming and one outgoing main branch edge, and
\item for all mappings $\psi_i$ of $p_{2}$ into $p_{1}$, %, with a node $n_i$ s.t. $\psi_i(n_i)=n$,
\item for $d'$ being the pattern obtained from $d$ by collapsing each $n_i \in p_2$ with $\psi(n_i)$, %(all these pairs are collapsible),
\item we have that \pattern{\lambda_d(n)[Q]} has a root-mapping into \sub{d'}{n}.%, for $n'$ being the node that replaced $n$ in $d'$.
\end{itemize}
(Special cases:  $Q$ is attached to a node of $p_2$ itself,  or   $Q$ is a full predicate in $d$, or both.) %, i.e.,  is directly connected to the main branch by a /-edge.)
%\bogdan{n instead of n1 in the figure, Qs in the proof figure, lk-1 also.}
\begin{center}
\includegraphics[trim=0mm 137mm 95mm 1mm, clip=true, scale=0.45]{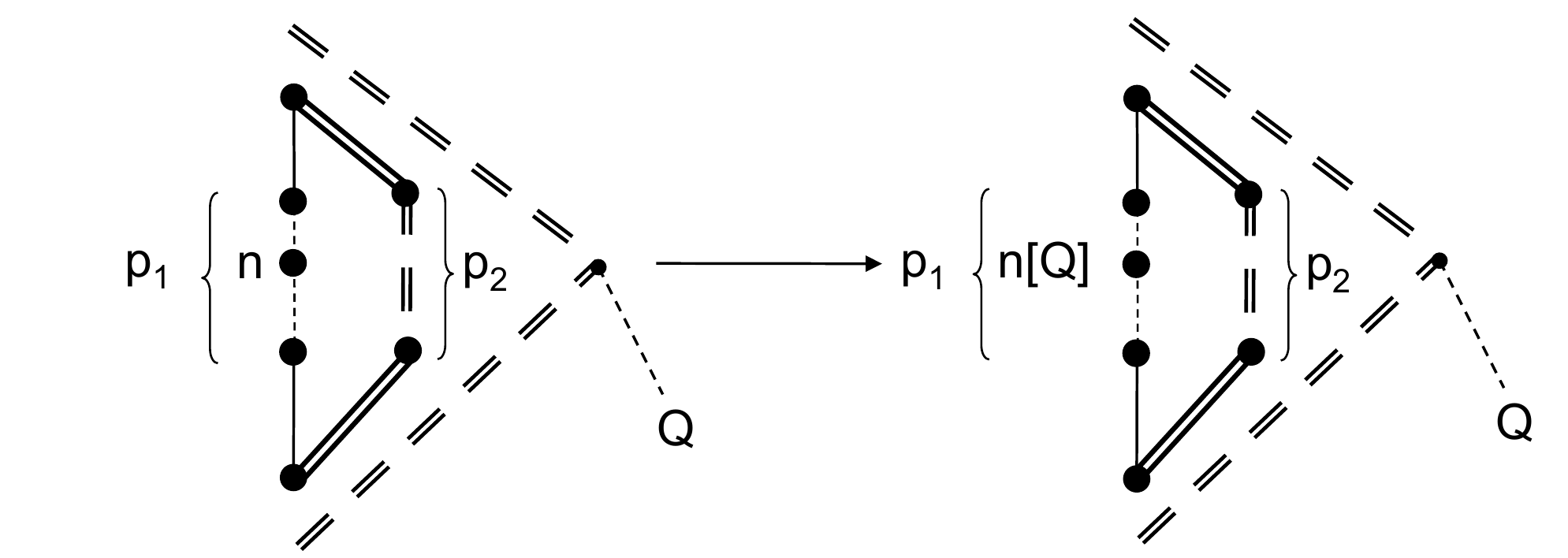}
\end{center}

\textbf{Remark.} This rule was not included in the  extended abstract published in \cite{CautisWebDB08}, but it is required to achieve completeness.

%\ioana:{For instance, the DAG obtained by intersecting the views doc/a/b/c/d and doc//a//c[q]//d  will be subject to R9 application, copying predicate [q] to the first view c-labeled node.}
\begin{example} The DAG pattern that would be obtained by intersecting some two views $\snippet{doc(``L'')/lib/section/section/section[figure]/image}$ and  $\snippet{doc(``L'')//section[figure]/section[figure]//image}$ would be subject to R9's application,
with predicate $Q$ being $\snippet{[figure]}$ and the node $n$ corresponding to the second node labeled $\snippet{section}$ in the former view. Note that only after adding $Q$ on node $n$  R7 can apply,  removing the branch from the latter view and yielding a tree pattern $\snippet{doc(``L'')/lib/section/section[figure]/section[figure]/image}$. %$p_2$ corresponding to the $\snippet{lib/paper/section}$ parts of the views.
\end{example}
\begin{proof}[for Lemma\ref{lemma:dintofr} -  soundness of R9]
The conditions ensure that in any possible interleaving of $d$, in particular in any interleaving of the parts $p_1$ and $p_2$, the predicate $[Q]$ will be present at the position of $n$. Equivalence follows.
\end{proof}

\textbf{Remark.} Some of the rules, such as R3 or R6, could safely collapse several nodes, without changing any of the results in
Section~\ref{sec:guarantees}. We opted for the current version for
ease of presentation.

\begin{example}We illustrate in Figure~\ref{fig:runexample} how  the unfolding
of the intersection of views $v_1$ and $v_2$ from the running example  is rewritten it into a prefix of $q_2$ (see
Figure~\ref{fig:runexample}.(c)).  Then, line~\ref{li:ret-rw} in
\algoname adds the navigation \snippet{/file} and the
intuitive rewriting $r_2$ that we described there is obtained.
\end{example}

\section{Using XML keys}
\label{sec:keys}
We discuss in this section a more general setting in which view-based rewriting can be solved by the techniques of our paper,  even when nodes do not necessarily have persistent Ids. More precisely, in the presence of \emph{key constraints}, the kind of rewrite plans we considered here can still be supported.  For example, key attributes are often present in XML documents, e.g., an SSN attribute could be associated with patient elements in a medical record, and could  play the role of node identifiers in a plan which intersects view results at the level of patient elements.

Several formalisms for specifying integrity constraints on XML data have been consider in standards and research literature~\cite{BunemanDFHT03,HartmannL09}.  In general, they specify in terms of path expressions three components: i) a context, either the root  or a set of nodes, (ii) the target of the key constraint,  namely the nodes that should be uniquely identified within the given context, and (iii) the key nodes, descendants of the target ones, whose values will form the  keys.    The key constraints for which the context is  the root of the document are called absolute, while the others are called relative.  One of the most important problems over key constraints is the one of \emph{constraint implication}, which consists in inferring new constraints from  the existing ones.
 For our purposes, only absolute constraints (explicit or inferred) are relevant, given that query results are always computed in the root context.

In the presence of key constraints,  additional steps of reasoning about node identity  are necessary in order to understand which rewrite plans (involving the intersection operator)  are valid. A naive approach would  be to perform a complete  pass over the data, annotating each node that is the scope of at least one key constraint with a (fresh) node Id. Then, only plans that intersect result sets in which at least one of the sets is composed only of  annotated notes will be valid and executable.   This would allow us to  transparently apply algorithm \algoname on XML data under key constraints, modulo this  final validation step that needs to inspect the actual view results. 

An alternative approach could be to perform the additional reasoning step directly at the query level.  More precisely,  the algorithm \algoname could  be modified to handle also XML data with key constraints by replacing line $1$  with the following one (modification in italic):
\begin{codebox}
%%\li let ${\cal V'} \gets \{ \func{bestcomp}(v,q) \;|\; v \in {\cal V} \textrm{ and }v\textrm{ maps into } q \}$
\li \label{li:init-prefs}$\id{Prefs} \gets \{ (p,\{(v_i,b_i)\}) \;|\; v_i \in {\cal V},
  p \textrm{ a lossless prefix of } q \textrm{ and } p \emph{ \textrm{contained in a target path}}, $
  \zi ~~~~~~~~~~~~~~~~~~~$ b_i \in \mb{p},  \exists  \textrm{ a mapping }  h \textrm{ from } u_i\!=\!\pattern{v_i}
    \textrm{ into } p, h(\out{u_i}) = b_i \}$
\end{codebox}
Testing whether a certain path selects only nodes that are uniquely identified by a key would only require some containment tests (depending on the formalism in which target paths are formulated).

Following these two possible approaches - either by Id-annotating  the XML data based on keys or matching the input query and prefixes thereof against the specified key target paths - all the results formulated under the assumption of persistent Ids remain valid in this new setting and \algoname remains a decision procedure for view-based rewriting.
%\vspace{-0.05cm}
\section{Formal Guarantees of Algorithm \algoname}\label{sec:guarantees}
%\vspace{-0.05cm}
Using Proposition~\ref{l:step-soundness}, we first show that algorithm \algoname
(and \efficient and \allrw) is sound, i.e. it gives no false
positives.
%\vspace{-0.05cm}
\begin{theorem}%%[Soundness]
\label{thm:soundness}
  If algorithm \algoname (or \efficient or \allrw) returns a
  DAG pattern $r$, then $\unfold{r} \equiv q$.
\end{theorem}
\begin{proof}
By construction $d = \dagp{\unfold{r}}$  maps into
$\pattern{q}$, hence, by Lemma~\ref{l:mapping-suff}, $q' \sqsubseteq d$,
where $q'$ is $q$ modified to have its output at the last main branch node of the prefix $p$.%node $b$.

By Proposition~\ref{l:step-soundness}, every rule application preserves
equivalence. Hence the final $d$ is equivalent to the $d$ initially
built.
Let $r'$ be the rewriting returned at line~\ref{li:ret-rw}.
Since $r'$ is just $r$ with, possibly, some more navigation, and
navigation is monotonic, then $d \sqsubseteq p$ guarantees that
$\unfold{r'} \sqsubseteq q$.
We already knew that $d$ had a containment mapping into $q'$. Then,
the last compensation
just ``moves'' the output node lower, and guarantees that $q \sqsubseteq \unfold{r'}$.
Hence $\unfold{r'} \equiv \pattern{q}$.
\end{proof}

%\vspace{-0.1cm}
\eat{
\begin{lemma}%%[Soundness]
\label{l:step-soundness}
  The application of any of the rules from the set R1-R9
  on a DAG $d$ produces another DAG $d' \equiv d$.
\end{lemma}
}
%\vspace{-0.1cm}
%% We recall that the rewriting problem for a query $q$ using a set of
%% views $\cal V$ is to decide whether there exists a rewriting of $q$
%% using $\cal V$.
%% In the following, by saying that an algorithm is \emph{complete} for
%% $\cal L \subseteq \xp$, we mean that it solves the rewriting problem
%% for queries and views in $\cal L$.
Moreover, it is also complete, in the sense described in Section~\ref{sec:preliminaries}.
%\vspace{-0.05cm}
\begin{theorem} \label{th:completeness} \hspace{-3mm} (1) Algorithm \algoname is complete for rewriting \xp. \\
%%\xp\{/,[],//\}.
%Moreover, if the views are in \reminder{our restricted languages},
%the running time is polynomial.
$~~~~~~~~~~~~~~~~~~~~~~~~~~~~~~~~~~$  (2) If the input query $q$ is minimal, \allrw
finds all minimal rewritings.

We say a rewriting $r$ is \emph{minimal} if the DAG pattern corresponding to $r$ has no equivalent strict sub-pattern.\footnote{The focus on minimal rewritings is standard in the literature on rewriting using views.  One reason is that, for a large class of cost models, minimal  size rewritings are provably also  cost minimal.  This class corresponds to 
monotonic cost models, in which the cost of evaluating the entire query is never less than the cost of evaluating any of its subqueries.  While there are exceptions to monotonicity in real-life scenarios, the class is still quite prevalent in 
practice. An additional reason for the traditional focus on minimal rewritings is that the set of non-minimal rewritings is infinite even in the most basic scenarios (relational conjunctive queries and views), rendering a search for
all non-minimal rewritings meaningless.}   
\end{theorem}
\begin{proof}[of Point (1)]  If there is a DAG $d$ equivalent to the query $q$, then $q \sqsubseteq d$,
and then, by Lemma~\ref{l:containment-tree-dag}, there must be a containment
mapping from $d$ into $q$.  

The completeness of algorithm \algoname is not affected by the usage of  \proc{Apply-Rules}, in the following sense:  the result of applying the  containment test directly on the DAG pattern corresponding to the unfolding of $r $ is the same as when applying it to the result of  \proc{Apply-Rules}. This follows from the equivalence preservation property of the rules. For this reason, we do not need to consider \proc{Apply-Rules} in the completeness proof.

%%% IF WE RESTRUCTURE, THE PREVIOUS COMMENT MUST BE REMOVED

Thus, if there is a rewriting
of the form ${\cal I}$ or ${\cal I}/\textit{comp}$ or ${\cal I}//\textit{comp}$,
where $\cal I$ is an intersection,
in particular there must be a root-mapping from \unfold{\cal I}
into $q$.
%(where the output of ${\cal I}$ is defined as the lowest node in the main branch
%which is in ${\cal I}$ ???).
Then, looking for all rewrite plans amounts to testing for all prefixes of \mb{q}
if there is an intersection of views, with possibly some compensation on each
branch, that map their main branches into that prefix, then add the compensation
below the prefix, which will guarantee a containment mapping from the unfolding
into the entire $q$.

Step~\ref{li:init-prefs} in \algoname tries all prefixes of \mb{q} and
for each such prefix it finds all views that map their main branch
inside that prefix. Hence, for each $\cal I$ as before, there will be
a $\cal I'$, built at step~\ref{li:build-plan}, of the form
${\cal I} \cap {\cal J}$, where $\cal J$, possibly empty, is an
intersection of other views (with maybe compensation) that map their
main branches inside the same prefix. But then $q \sqsubseteq {\cal
I}' \sqsubseteq {\cal I}$, so ${\cal I}$ is a rewriting iff ${\cal
I}'$ is one.

We argue next that in order to test whether $r'=\func{compensate}(r,q,\out{p})$ is an equivalent rewrite plan (which amounts here  to testing that $\unfold{r'} \sqsubseteq q$) it is sufficient (and obviously necessary) to test that $d\sqsubseteq p$ (in \algoname at line 7). In other words, if the test $d \sqsubseteq p$ fails, then $r$ (potentially compensated) cannot yield an equivalent rewriting.

The test for $d \sqsubseteq p$ amounts to (a) testing that $d$ is union-free, with some interleaving $i$ such that $i \equiv d$, and (b) testing that $i \sqsubseteq p$. If $d$ is union-free, then the statement \emph{$\unfold{r'} \sqsubseteq q$ iff $d\sqsubseteq p$} can be proven in straightforward manner.

To conclude the proof of completeness, we use Claim~\ref{c:completeness} , which shows   that $d'=\unfold{r'}$ cannot be union-free if $d$ is not union-free (proof given separately hereafter).   
\vspace{-1mm}
\begin{claim}
\label{c:completeness}
$d'$ is union-free only if $d$ is union-free.
\end{claim}
\vspace{-1mm}
Finally, since \algoname is sound (Th.~\ref{thm:soundness}),
any rewrite plan that is output is indeed a rewriting.
%% guarantees that $d$ obtained after line~\label{li:apply-rules} is
%% equivalent to the DAG $d$ from before and that $r$, plus maybe some
%% compensation, forms a rewriting if $d=\unfold{r}$ passes the reverse
%% containment check, at line~\ref{li:rw-cnt-check}.
\end{proof}
\vspace{-2mm}
\begin{proof}[of Point (2)]
  If $q$ is minimal, all the subtrees added by the function
  \func{compensate} are also minimal, hence the rewritings that are
  produced are minimal. It is left to prove that \proc{\allrw} finds
  them all.
  We already saw that if there is a rewriting, \algoname tries one
  which is equivalent to it. \allrw does even more: it really tries all
  the rewrite plans, because it considers all subsets of views that
  map down to a certain position in \mb{q}, at line $\ref{li:mainloop}'$.
%%  More precisely, for all nodes of
%%   $q$ into which views map, it generates all intersections of views
%%   $V_j$ with partial compensations $p_i$ that map their output into
%%   that node.
  Any other compensation that may be needed
  is copied from $q$ in the call to \func{compensate} from the return
  clause (line~\ref{li:ret-rw}). It is easy to prove that a tree pattern is minimal iff one
  cannot drop subtrees from it (which is the definition of minimality
  in \cite{mandhani-suciu}). Therefore, if $q$ is minimal, all its
  subtrees are also minimal.
\end{proof}
\vspace{-2.5mm}
\begin{proof}[of Claim~\ref{c:completeness}] Let $d \equiv i_1 \cup \dots i_k$, for $i_1, \dots, i_k$ being the non-reducible interleavings of $d$ (for each $i_j$ we cannot build \emph{another}, i.e. non-equivalent, interleaving that contains $i_j$). For each $i_j$, let $i_j'$ denote the compensated interleaving $i_j'=\func{compensate}(i_j,q,\out{p})$.

First, it is immediate that $d' \equiv  \bigcup_j i_j'$ by showing containment mappings in both directions.

We show next that if $d$ is not union-free, then $\bigcup_j i_j'$ is not union-free either (there is no query in this union that contains all other queries). We will rely on the following two observations:
\vspace{-1mm}
\begin{enumerate}
\item The compensation applied on $d$'s interleavings will only extend the main branch, but will not bring (or qualify) new predicates for the existing main branch nodes, hence yielding no new mapping opportunities in this sense. This is because the node $\out{p}$ in prefix $p$ is assumed to have already ``inherited'' the rest of $q$ as a predicate ($p$ is a lossless prefix).

    In other words, if there exists a containment mapping between $i_j'$ and $i_l'$, then there must exist a root-mapping between the corresponding $i_j$ and $i_l$.

\item We can partition the interleavings of $d$ into two classes:
 \begin{enumerate}
 \item those that have a ``minimal'' (or certain) result token (essentially obtained after applying R1 rewrite steps on the result tokens of $d$'s parallel branches),
     \item the remaining ones, which by definition must have a result token that cannot map in the result token of interleavings of the first kind (the result token has a longer main branch or some predicates that are not necessarily present in all interleavings\footnote{Since this is the result token and everything else can be put ``above'' it, it is easy to obtain this minimal, certain result token.}).
\end{enumerate}

The first class cannot be empty, the second one may be empty.

\end{enumerate}

%\newpage

 By the first observation, for any two interleavings of $d$, $i_j$ and $i_l$, knowing that $i_j \not \sqsubseteq i_l$,  we can have $i_j' \sqsubseteq i_l'$ only if there exists already a root-mapping $\phi$ from $i_l$ into $i_j$ which fails to be a full containment mapping only because the image of $\out{i_l}$ is not $\out{i_j}$. If such an root-mapping does not exist then surely we have that $i_j' \not \sqsubseteq i_l'$.

 By the second observation, for interleavings of the first kind, their result token will always map in the result token of any other interleaving (of both kinds) such that the image of the output node is the output node in the other interleaving.

Putting everything together, since the first class cannot be empty, in order to prove the claim it is now sufficient to show that a query $i_l'$  corresponding to an interleaving $i_l$ of the second kind  cannot reduce (contain) a query $i_j'$ corresponding to an interleaving $i_j$ of the first kind.   Let $i_j$ and $i_l$ be the two interleavings, with  $t^o_j$ and $t^o_l$ denoting their result tokens, such that there exists a root-mapping  $\phi$ which takes $\out{i_l}$ into some main branch node that is not part of $t^o_j$ (by the definition of the two kinds of interleavings $t^o_l$ cannot map into $t^o_j$).

We show that this leads to a contradiction, namely that $i_j$ \emph{is not minimal}, in the sense that a main branch node of $i_j$ can be removed, obtaining another (simpler) valid interleaving of $d$. %which contains $i_j$ but is not equivalent to it.

Each of the branches in parallel in $d$ (denoted hereafter $x$) must have a containment mapping $\psi_x$ into $i_l$ (and a containment mapping $\tau_x$ into $i_j$ as well), so we obtain by  $\phi \circ \psi_x$ a root-mapping from any branch $x$ into $i_j$. Importantly,  the image of \out{x},  $n=\phi(\psi_x(\out{x}))$, is some main branch node of $i_j$ \emph{above} $t^o_j$.  This means that we can also map (by a root-mapping) each $x$ branch into $i_j$  somewhere higher than the result token $t^o_j$. But this hints that we can in fact simplify $i_j$ into an interleaving having a shorter main branch, yet containing $i_j$. More precisely, this interleaving of the parallel branches $x$, described by a code-mapping pair $(i,f_i)$, can be obtained as follows:

 \begin{enumerate}
  \item take as the code $i$ the main branch of $i_j$ \emph{without} $n$,
   \item define $f_i$ as (a) for the main branch nodes of each $x$ except those of the result token, by the $\phi \circ \psi_x$ mapping, (b) for the main branch nodes of the result token of $x$ by $\tau_x$.
 \end{enumerate}

 It is easy to check that the interleaving obtained in this way has a containment mapping into $i_j$ based on the ``identity'' mapping for the main branch nodes.
\end{proof}
We now consider the computational complexity of \algoname. We can observe that it  runs in worst-case exponential time, 
as it uses a containment check (line~\ref{li:rw-cnt-check}) that is
inherently hard:

\begin{theorem}\label{thm:containment-dag-tree}
\hspace{-3mm}Containment of an  \xpcap query $d$ into an  \xp query $p$ is
coNP-complete in $|d|$ and $|p|$.
\end{theorem}
%% \begin{lemma}\label{l:containment-dag-tree}
%% Containment of a query $d \in \xpcap$ into a query $p \in \xp$ is
%% coNP-complete in $|d|$ and $|p|$.
%% \end{lemma}
\begin{proof}
To show that our problem is in coNP we can use an argument similar to
the one used in \cite{DBLP:journals/lmcs/NevenS06} for the case of
XPath with disjunction. We know that $\textit{nf}(d) = U$, where $U$
is a union of queries from \xp. A non-deterministic algorithm that
decides $d \not \subseteq p$ guesses $u \in U$ (without computing
$U$), making a certain choice at each step of interleaving. Then, it
checks that $u \not \subseteq p$, which can be done in PTIME as $u, p \in
\xp$.

coNP-hardness is proven by reduction from the 3DNF-tautology problem~\cite{DBLP:books/fm/GareyJ79}, which is known to be coNP-complete. We start from a 3DNF
formula $\phi(\bar x) = C_1(\bar x) \vee C_2(\bar x) \vee \dots
C_m(\bar x) $ over boolean variables $\bar x = (x_1,\dots x_n)$,
 $C_i(\bar x)$ being conjunctions of literals.% We denote by
%$\mathcal{C}_l = \{C^{(l)}_j\}_j \subseteq \{C_k\}$ the clauses in
%which the literal $l$ appears.

Out of $\phi$ we build $d \in \xpi$ and
$p \in \xp$ over $\Sigma = \{x_1, \dots, x_n, b, v, true, false, yes\}$ such that $\phi$ is a tautology iff $d
\subseteq p$. Intuitively, $d$ will encode all possible truth assignments for $\phi$.

We build $d$ by intersecting two branches, based on the following gadgets (illustrated in Figure~\ref{fig:cnt_reduction}):
\begin{enumerate}
\item the pattern $x_1[yes]//x_2[yes]//...//x_n[yes]$ (denoted $T$)
%\item $doc(A)//t//b$ (denoted $r$, the right branch in $d$)

\item the pattern $x_1[true]/x_1[false]/x_2[true]/x_2[false]/\dots/x_n[true]/x_n[false]$ (denoted $S$)

%\item $/x_k/a\ \mathcal{P}(x_k, \phi)/x_k/a\ \mathcal{P}(\neg x_k, \phi)$, where,
%for a given literal $l$, $\mathcal{P}(l,
%\phi)$ is a set of predicates $[C^{(l)}_1][C^{(l)}_{2}]\dots$ formed by
%considering all clauses $C^{(l)}_k \in \mathcal{C}_l$. This pattern is denoted $S_k$. By $S$ we denote the pattern $S_1/S_2/\dots/S_n$.

\item for each clause $C_i$, the pattern obtained from $x_1/a/x_1/a/x_2/a/\dots /x_n/a/x_n/a$ by putting the $[yes]$ predicate below each of the $3$ nodes corresponding to the literals of $C_i$ (this pattern is denoted $P_{C_i}$ ). For instance,  for $C_i=(x_1 \wedge \bar{x_2} \wedge x_3)$, we have the pattern $P_{C_i} =x_1/a[yes]/x_1/a/x_2/a/x_2/a[yes]/x_3/a[yes]/x_3/a\dots /x_n/a/x_n/a$.
\item for each clause $C_i$, $Q_{C_i}$ denotes the predicate $[c/c/c/\dots/c/v[.//P_{C_i}]]$, with $m-i +1$ $c$-nodes,

\item for each $C_i$, the predicate $Q_i=[Q_{C_1}, \dots Q_{C_{i-1}}, Q_{C_{i+1}}, \dots, Q_{C_m}]$, that is the list of all $Q_{C_j}$ predicates for $j \neq i$.
    \item the  pattern $c[Q_1]/c[Q_2]/c[Q_3]/\dots c[Q_m]/c$ (denoted $U$)
\end{enumerate}

%trim=10mm 163mm 30mm 2mm, clip=true, scale=0.45

\begin{figure}[t]
\begin{center}
\includegraphics[trim=0mm 50mm 0mm 0mm, clip=true, scale=0.45]{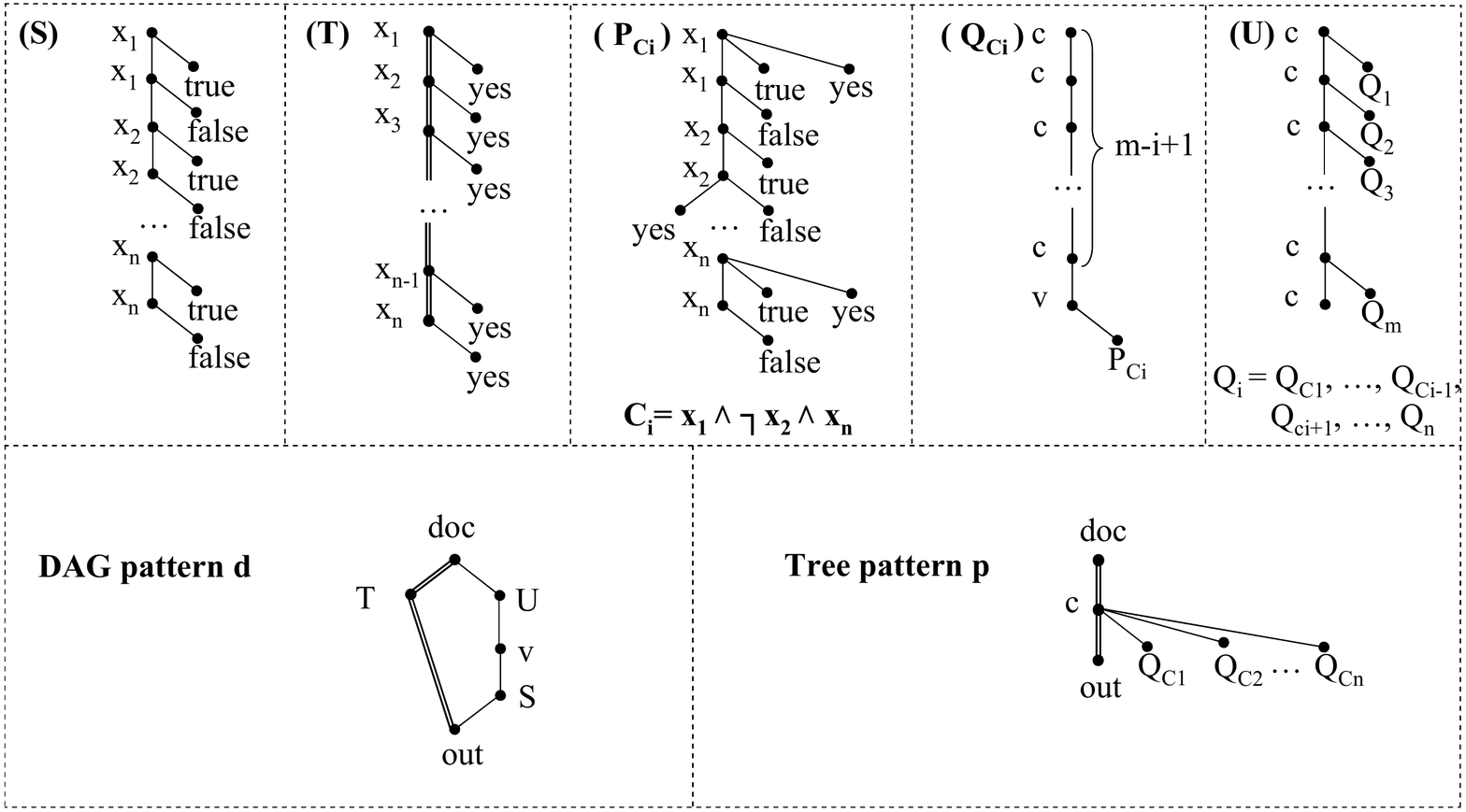}
\end{center}
\vspace{-2mm}
\caption{The patterns used in the reduction of TAUTOLOGY to DAG containment \label{fig:cnt_reduction}}
\vspace{-3mm}
\end{figure}

Now, we define $d$ as $d=(doc(A)//T//out) \cap (doc(A)/U/{\mathbf v}/S/out)$.

An important observation is that $Q_{C_i}$ predicates are not inherited at $c$-nodes. This is ensured by the  $v$ node we introduced in $Q_{C_i}$s and in $U$ (for inheritance, the last $c$-node of $Q_{C_i}$'s branch would have to match with a $v$-node).

Next, when interleaving the two branches in $d$ , $x_k$ in the left branch can either be coalesced with the first $x_k$ (having predicate $true$) in the
second branch, corresponding to the case $x_k =
\textsc{true}$, or with the second $x_k$ (having predicate $false$), corresponding to
$x_k=\textsc{false}$. Hence interleavings correspond to truth assignments for the variables of $\phi$.

Note also that when some clause $C_i$ is made \textsc{true} by a truth assignment (i.e., all its literals are \textsc{true}), then the predicate $P_{C_i}$ will hold at the last $c$ node in the $U$ part, or, put otherwise, the $Q_{C_i}$ predicate will now hold at the $i$th $c$ node in $U$.

Finally, let $p$ be the pattern $doc(a)//c[Q_{C_1}, Q_{C_{2}}, Q_{C_{3}}, \dots, Q_{C_m}]//out$.

We can now argue that $d \sqsubseteq p$ iff $\phi$ is a tautology. The if direction is immediate since in this case each truth assignment makes at least one $C_i$ true. This means that the $Q_{C_i}$ predicate will now hold at the $i$th $c$-node in $d$ and, since all other $Q_{C_j}$ predicates, for $j\neq i$, were already explicitly present at this node, it is now easy to see that there exist now a containment mapping from $p$ into $d$ that takes $p$'s $c$-node into $d$'s $i$th $c$-node.

The only if direction is similar. If for some truth assignment, none of the clauses is \textsc{true} (in the case $\phi$ is not a tautology), then it is easy to check that $p$ will not have a containment mapping into the interleaving corresponding to that truth assignment.
\end{proof}

%%And we prove that the problem itself is hard.
One might hope there is an alternative polynomial time solution for the rewriting problem, which would not require such a potentially expensive containment test. We
prove this is not the case,  showing that that the rewriting problem itself is hard.
\begin{theorem}
\label{thm:coNPcompleteness}
The rewriting problem for queries and views from \xp is coNP-complete.
\end{theorem}
\begin{proof}%({\bf Theorem~\ref{thm:hardness-rewriting-desc} (1)})
%\textbf{Point (1).} First, let us prove the upper bound.
For the coNP-hardness proof, we refer the reader to the proof of Theorem~\ref{th:hardnessUF-desc-2}, which shows an even stronger result, as it deals with a restricted fragment of \xp.

We discuss now the complexity upper-bound. Let $q$ be the input query, let $\cal V$ be the set of views.  First, note that rewriting could be solved using an oracle
for union-freedom, but this does not provide any easy map reduction. 
This is why we prove the following result independently.  

\begin{figure}[t]
 % \hspace{2cm}
\begin{center}
  \includegraphics[trim=8mm 155mm 175mm 3mm, clip=true,
scale=0.5]{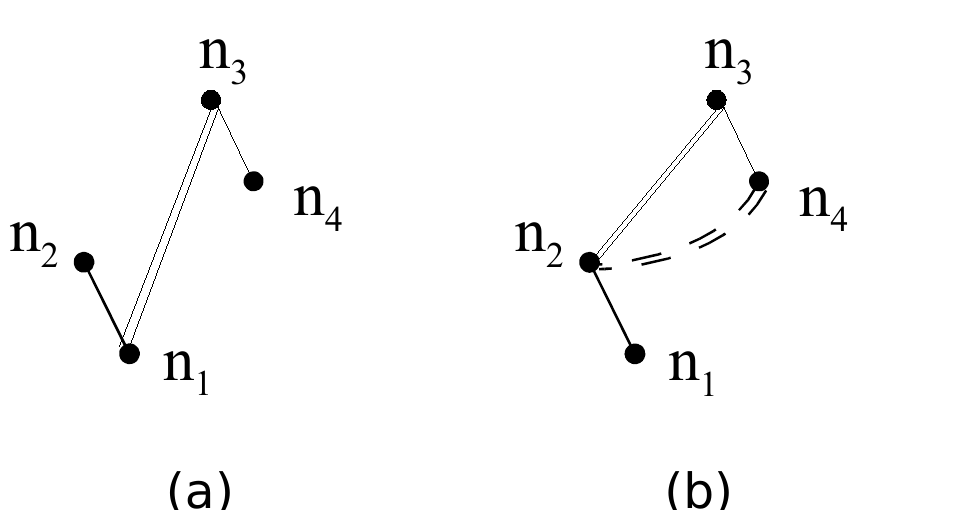}
\end{center}
\caption{Interaction between R2.i and R2.ii \label{fig:termin_R2}}
\vspace{-3mm}
\end{figure}

For each node of \mb{q} in which some of the views ${\cal V}_1
\subseteq {\cal V}$ map, it is enough to guess one interleaving of
$\bigcap_{{\cal V}_1} v_j$ in which $q$ does not map.  If we put
together a polynomial number of polynomially large witnesses, they
make up a polynomial witness for the entire problem. In other words,
one can verify in polynomial time that there is no rewriting.
\end{proof}
%\vspace{-0.05cm}
However, we show next that our rule rewriting procedure is polynomial, hence rewritings can be found efficiently whenever the containment test of \algoname's line $6$ can be done efficiently:
%\vspace{-0.05cm}
\begin{lemma}%%[Termination]
\label{l:termination}
The rewriting of a DAG $d$ using \proc{Apply-Rules} always terminates, and it does so
in $\mathrm{O}(|\nodes{d}|^2)$ steps.
\end{lemma}
\begin{proof} First, let us notice that none of the rules increases the number of
main branch nodes, and in fact R1, R3, R4, R6, R7, R8 always decrease
it, hence the number of times they are applied is less than $|\mbn{d}|
= \mathrm{O}(|\nodes{d}|)$. R5 can also fire only a finite number of
times, as the number of predicates to be introduced is bound by the
initial number of predicates in the pattern $d$, which is in
$\mathrm{O}(|\nodes{d}|)$.

R2 leaves the number of nodes unchanged and may decrease
the number of edges by one or leave it the same. R2.i 
always progresses down main branches, and R2.ii always up,
respectively. The only possibility of going into a loop 
would come from the interaction of R2.i and R2.ii.

Consider the generic case depicted in
Figure~\ref{fig:termin_R2}(a), in which R2.ii would apply for nodes
$n_1$, $n_2$, $n_3$.  (The case in which we apply an R2.i step is
symmetrical). Suppose that $n_3$ also has a /-edge towards a node
$n_4$. There would be a danger of looping if R2.i had been previously
applied to nodes $n_3$, $n_2$, $n_4$, and now it would apply again
because $n_2//n_3$ would be re-introduced. But then, R2.i applied to
those nodes would have introduced an edge $n_4//n_2$ and any later
applications of R2 (or of any other rule) would have maintained $n_2$
reachable from $n_3$, as in Figure~\ref{fig:termin_R2}(b). In this
case, R2.ii would not introduce any //-edge between $n_3$ and $n_2$,
as it is explicitly specified in the rule.

Thus, R2 can fire at most $|\mbn{d}|^2$ times, because at
each step it infers the order, in all interleavings, of a pair of
nodes from \mbn{d} whose ordering was unknown before. So, rewriting
with R1-R9 always terminates in at most $\mathrm{O}(|\nodes{d}|^2)$
steps.

Each of the rules R1-R8 can be tested in polynomial time in straightforward manner. They are mostly based on existence or non-existence of mappings). For some of them (in particular,  for R7, R8 and R9), we will discuss certain implementation choices that may speed-up execution in Section~\ref{sec:implementation}.   Similarity can also be tested in polynomial time, since the number of patterns $p_{12}$ to be considered (Definition~\ref{def:similar}) is linear in the size of the two /-patterns $p_1$ and $p_2$. 

We discuss next how R9 can be tested.

 For the given $n$, $p_1$, $p_2$,  and $Q$ (the number of such predicates is of the order  of $|d|$),  we can check in polynomial time whether the rule triggers  as follows.  $Q$ can be seen as having the following form (Figure~\ref{fig:r9Q}): a /-path $l_1/\dots/l_k$ followed by either (a) one or more //-edges, (b) one or more //-edges and one or more /-edges, or (c) one or more /-edges. In other words, $l_k$ denotes the highest node having either several outgoing edges (of either kind) or one outgoing edge, of  the // kind.

\begin{figure}[t]
\begin{center}
\includegraphics[trim=0mm 130mm 120mm 0mm, clip=true, scale=0.45]{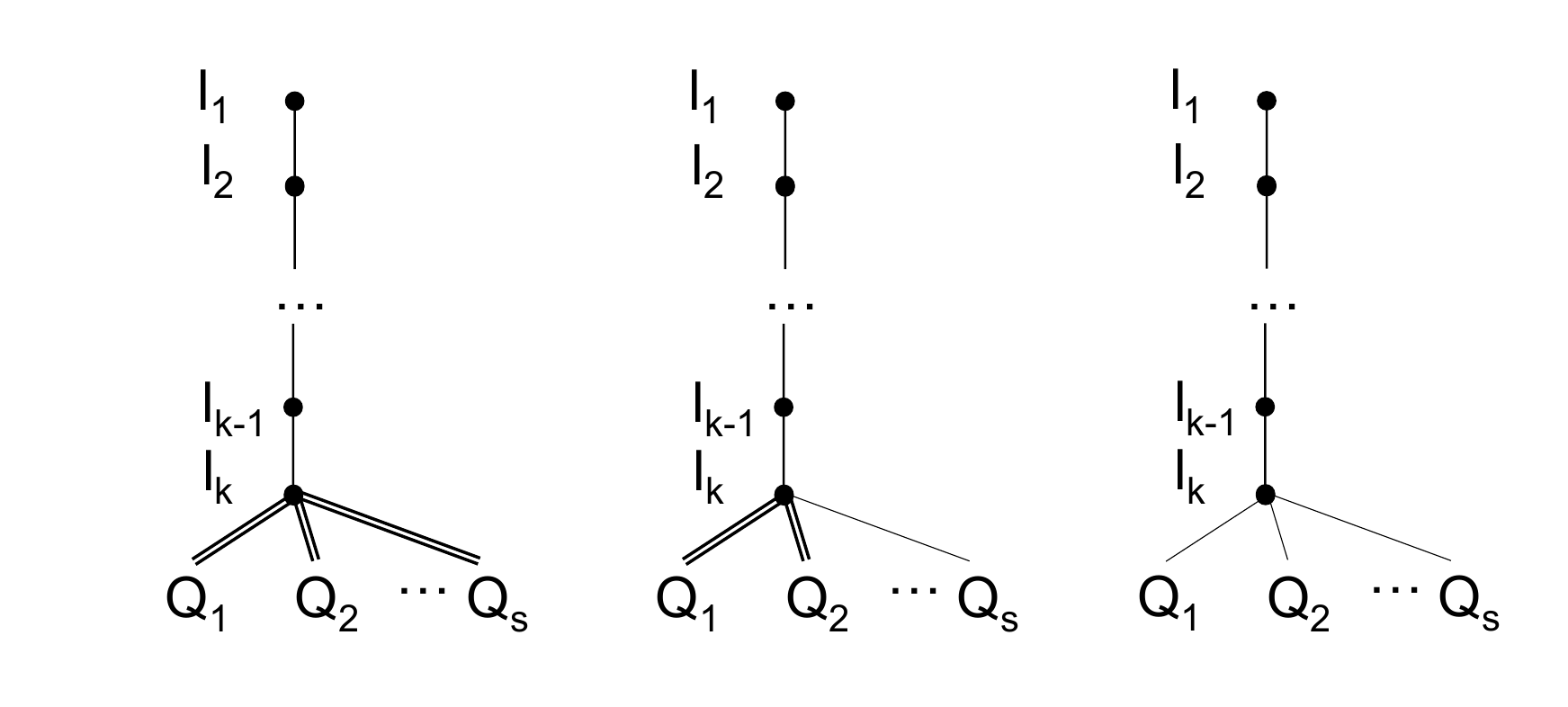}
\vspace{-2mm}
\caption{The possible configurations for predicate subtree Q.\label{fig:r9Q}}
\end{center}
\vspace{-4mm}
\end{figure}

\textbf{Case a.} If  $Q$ is of the first kind, since at node $n$ in $p_1$ the predicate $Q$ would verify \xppes (see the following section), it means that $n$ is followed by a main branch that is incompatible with $l_1/l_2/\dots/l_k$. Let $l_1/\dots/l_{k'}$, for $1\leq k'<k$, be the maximal prefix that is compatible with the  main branch (if one exists). This means that the main branch below $n$ starts by a sequence of labels $l_1/\dots/l_{k'}/l$, where $l \neq l_{k'+1}$.

For $Q$ to hold at $n$ in each  interleaving of $p_2$ with $p_1$, it means that in it we have either:
\begin{enumerate}
\item $Q$ or a predicate into which $Q$ can map attached to $n$ itself (i.e. we do not need the main branch descendants of $n$ and their predicates), or
\item the predicate $l_2/\dots/l_k[Q_1] \dots [Q_s]$ or a predicate into which it can map attached to $n$'s main branch child $n'$ (i.e. we do not need the main branch descendants of $n'$ and their predicates), or 
\item the predicate $l_3/\dots/l_k[Q_1] \dots [Q_s]$ or a predicate into which it can map attached to $n$'s main branch descendant at distance $2$, $n''$ (i.e. we do not need the main branch descendants of $n''$ and their predicates), or  so on, \dots
%    \item the predicate $l_4/\dots/l_k[Q_1] \dots [Q_s]$ or a predicate into which it can map being attached to $n$'s main branch descendant at distance $3$, $n'''$ (i.e. we don't need the main branch descendants of $n'''$ and their predicates),  or so on,  $\dots$

        \item [(k')] the predicate $l_{k'+1}/\dots/l_k[Q_1] \dots [Q_s]$ or a predicate into which it can map attached to $n$'s main branch descendant at distance $k'$, $n^{(k')}$, (i.e. we do not need the main branch descendants of $n^{(k')}$ and their predicates).
\end{enumerate}

Accordingly, in order to test that  $Q$ holds at $n$ in each  interleaving of $p_2$ with $p_1$, we need  to test the \emph{non-existence} of a mapping from $p_2$ into $p_1$ that would not bring a predicate as the ones described above on any of the nodes $n, n', n'', \dots, n^{(k')}$. This test can be done in polynomial time, top-down and one token at a time, by choosing as long as possible for each token of $p_2$  the highest-possible image that does not contribute any predicates like the ones described above.

\textbf{Case b.} This case is similar to the previous since we have the same setting, i.e.,  $n$ is followed by a main branch that is incompatible with $l_1/l_2/\dots/l_k$ and we have at most  a prefix of it $l_1/\dots/l_k'$, for $1\leq k'<k$, that is compatible (if such a prefix exists).

\textbf{Case c.} If $n$ is followed by a main branch that is incompatible with $l_1/l_2/\dots/l_k$, then the same reasoning of the two previous cases applies here as well. Otherwise, for $Q$ to hold at $n$ in each  interleaving of $p_2$ with $p_1$, it means that in each interleaving we have either:
\begin{enumerate}
\item $Q$ or a predicate into which $Q$ can map attached to $n$ itself (i.e. we do not need the main branch descendants of $n$ and their predicates), or
\item predicate $l_2/\dots/l_k[Q_1] \dots [Q_s]$ or one into which it can map attached to $n$'s main branch child $n'$ (i.e., we do not need the main branch descendants of $n'$ and their predicates), or so on, \dots

        \item [(k)] the predicate $l_k[Q_1] \dots [Q_s]$ or a predicate into which it can map being present (as a predicate) on $n$'s main branch descendant at distance $k$, $n^{(k)}$, (i.e. we do not need the main branch descendants of $n^{(k)}$ and their predicates), or
            \item [(k+1)] all the predicates $[Q_1], \dots, [Q_s]$ verified at $n$'s main branch descendant at distance $k+1$, $n^{(k+1)}$.

\end{enumerate}
So a similar test for the non-existence of a mapping has to be done, but with some minor adjustments. Top-down, we will chose a mapping image for each token of $p_2$ into $p_1$, as long as we do not arrive at the position of $n^{(k+1)}$ or below it (i.e. we will chose an image for a token  if it does not overpass this position and does not contribute predicates like the ones described by the items (1) to (k) above). Then, for the remaining suffix of $p_2$, we check the existence of a mapping for it that would (i) not 
contribute predicates like the ones given in conditions (1) to (k), and (ii) would not contribute \emph{all} the predicates of the last condition, i.e.,  that there is a mapping for the remaining part of $p_2$ in the remaining part of $p_1$ s.t. 
among $Q_1, \dots, Q_s$ there is at least one predicate $Q_i$  which will not be verified at $n^{(k+1)}$ after coalescing $p_2$'s nodes with their mapping images.  This can be seen as a recursive call, that can be run for each $Q_i$ individually, and will take us back to the three cases depending on the shape of $Q_i$. (Note that all the predicates $Q_1, \dots, Q_s$ at node $n^{(k+1)}$ on $p_1$ will verify the condition for extended skeletons.)

A dynamic programming approach can be used to perform all these tests in polynomial time, based on the to-be-mapped suffix of $p_2$, the target suffix of $p_1$ and the predicate to be tested (it is not necessary to perform the test several times for a given such triple). 
%we will start by choosing where to map its first token s.t. no predicates like the ones described by the items (1) to (k) above would be added on %$p_1$, and from there on
%For all candidates $n_i$ (there are at most $|p_2|$ many) we do the following: (i) check if there is a mapping of $p_2$ into $p_1$ that takes $n_i$ %into $n$ (this would mean already that $n$ and $n_i$ are collapsible), (ii) check if the tree pattern \pattern{\lambda_d(n)[Q]} has a root-mapping in %the DAG pattern \sub{d'}{n'}. If for all possible candidates $n_i$ we find that such a mapping exists then the rule is triggered and the predicate %$[Q]$ is added on $n$.
\end{proof}

An immediate corollary of Lemma~\ref{l:termination} is the following.
%\vspace{-0.05cm}
\begin{corollary}
\efficient always runs in PTIME.
\end{corollary}
%\vspace{-0.05cm}

To summarize the results so far, we showed that \algoname is complete for \xp, and we gave matching (coNP) complexity bounds for this problem. Moreover, we described a variant of this algorithm (\efficient) which runs in polynomial time, but is only sound.  

We consider next restrictions by which \efficient becomes also
complete, thus turning into a complete and efficient rewriting
algorithm. Note that one may impose restrictions on either the
\xp fragment used by the query and views, or on the rewrite plans
that \algoname deals with. We consider both cases,  by this charting a tight
tractability frontier for this problem. %We start by considering restrictions on the language of queries and views.

The next section shows that \efficient is complete under fairly permissive restrictions on the input query, and this \emph{without restricting the language of views (which remains \xp).}

\section{Tractability frontier - \xp fragment for PTIME} 
\label{sec:frontier1}
We introduce in this section a fragment of \xp that, intuitively,  limits the use of //-edges in predicates, in the following manner: any token $t$ of a pattern $p$ will not have predicates \emph{with //-edges} that  may become redundant in some interleaving $p$ might be involved in, due to descendants of $t$ and their respective predicates.

Let us first fix some necessary terminology. 
%\subsection{\xp fragment for PTIME} 
By a \emph{//-subpredicate} \textit{st} we denote a predicate subtree whose root is connected by a //-edge
to a /-path $p$ that comes from the main branch node $n$ to which $st$ is associated (as in $n[\dots[.//st]]$). %node $m$, as in
%%the expression
%$m[p[.//\textit{st}]]$.
 $p$ is called the \emph{incoming /-path} of $st$ and can be empty.

%% Extending the fragment of skeletons, we consider two richer fragments
%% of \xpp, in which //-edges may also occur in predicates.
%% By \emph{tree skeleton} we denote a tree pattern without //-edges in
%% predicate subtrees. Given a pattern $p$, by its \emph{skeleton} we
%% denote the skeleton pattern $s(p)$ obtained from $p$ by pruning out
%% its //-subpredicates. We denote by $\xpps$ the fragment of skeleton
%% patterns.
By
\emph{extended skeletons} (\xppes) we denote tree patterns having
the following property: for any main branch node $n \neq \out{p}$ and //-subpredicate
$st$ of $n$, there is no mapping (in either direction) between the code of the incoming /-path of
$st$ and the one of the /-path following $n$ in the main branch (where the
empty code is assumed to map in any other code).   For instance,  the patterns $\snippet{a[b//c]/d//e}$ or $\snippet{a[b//c//d]/e//d}$ are extended skeletons, while $\snippet{a[b//c]/b//d}$, $\snippet{a[b//c]//d}$, 
$\snippet{a[.//b]/c//d}$ or $\snippet{a[.//b]//c}$ are not. 

Observe  that the above definition imposes no restrictions on predicates of the output node. This relaxation was not present in~\cite{CautisWebDB08}'s definition of extended skeletons but it is easy to show that it does not affect any of the results that were obtained with the more restrictive definition. This is because there is only one choice for ordering the output nodes in interleavings of an intersection; they are collapsed into one output node. Note that this  \xp sub-fragment  does not restrict in any way the use of descendant edges in the main branch or the use of predicates with  child edges only.  % Obviously, $\xpps \subsetneq \xppes$.
Note also that all the paths given in the running example are from this fragment.

We denote by $\xppescap$ the fragment of $\xpcap$ in which only $\xppes$ expressions are used. For any tree pattern  $v$, by its extended skeleton, we denote the \xppes\ query $s(v)$ obtained by pruning out all the //-subpredicates violating the \xppes\ condition.   This notion can be easily generalized to extended skeletons $s(d)$ for any DAG patterns $d$.

The following two lemmas have the auxiliary role of  allowing us to rewrite input queries from \xppes\  in polynomial time, without imposing any restrictions on the views.
\begin{lemma}
\label{lem:skel-necessary}\hspace{-2mm}
A DAG pattern $d$ is union-free only if its extended skeleton DAG pattern $s(d)$ is so.
\end{lemma}
\begin{proof}
The proof is based on the following property: modulo //-predicates,
the sets of tree patterns \interleave{d} and \interleave{s(d)} are the
same. More precisely, for each $p_i \in \interleave{d}$ there exists
$p_i' \in \interleave{s(d)}$ such that $s(p_i)=p_i'$ and the other way
round. Supposing that $s(d)$ is not union-free, let us assume towards
a contradiction that $d$ is union-free. Let $p_i$ denote the
interleaving such that $p_i \equiv d$ and let $p_i'$ denote the
associated interleaving from $s(d)$, $p'_i = s(p_i)$. Since $s(d)$ is
not union-free, there exists some $p_j' \in \interleave{s(d)}$ such
that $p_j' \not \sqsubseteq p_i'$.  Then there is $p_j \in
\interleave{d}$ such that $s(p_j) \not \sqsubseteq s(p_i)$.  Finally,
since $p_j \sqsubseteq p_i$ we also have $p_j \sqsubseteq s(p_i)$.

\begin{figure}[tb]
\begin{center}
\includegraphics[trim=15mm 20mm 50mm 60mm, clip=true,
scale=0.40]{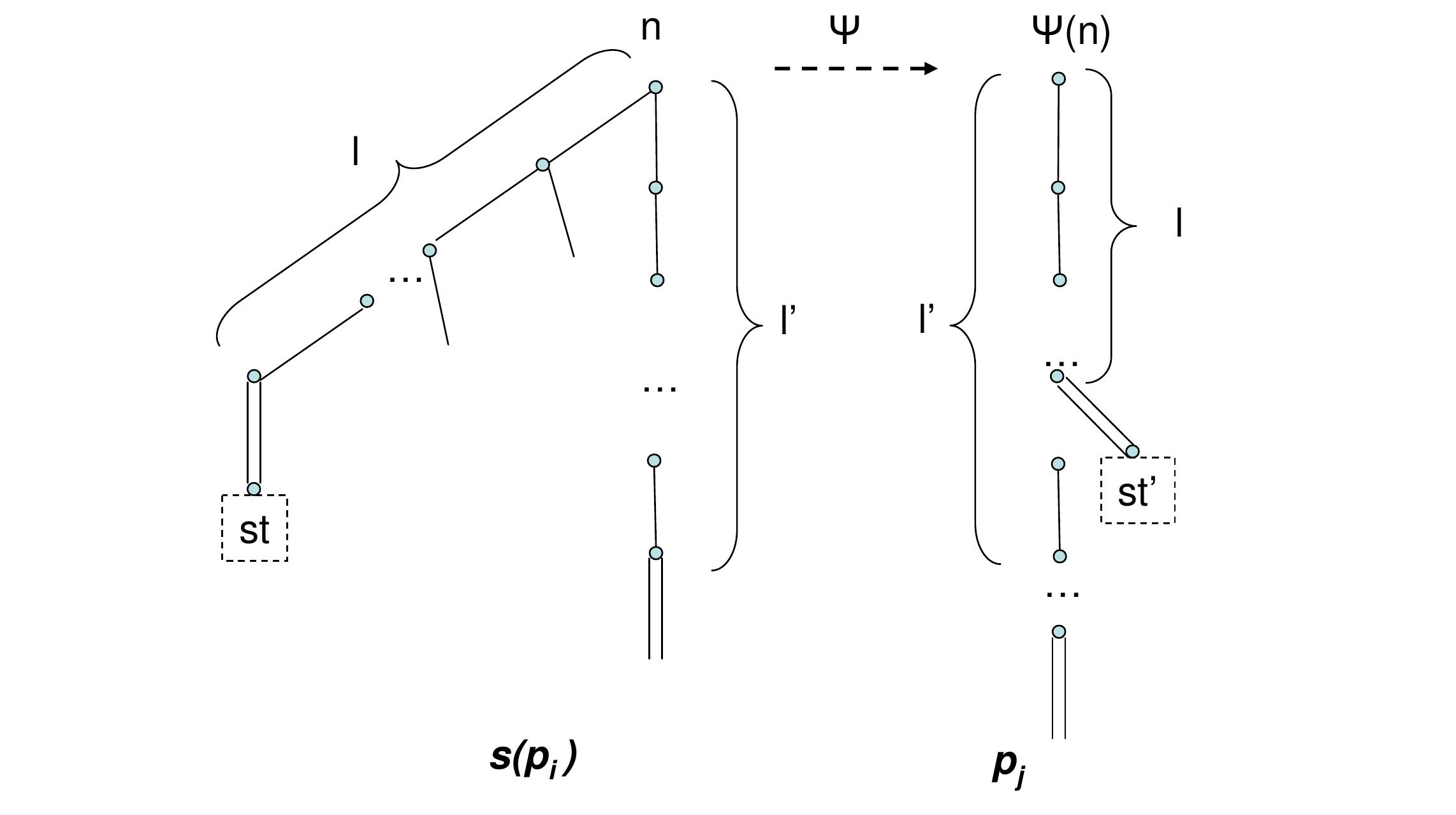}
\end{center}
\vspace{-3mm}
\caption{Extended skeletons and mappings. \label{fig:picture-skel}}
\vspace{-0.3cm}
\end{figure}
Suppose for the sake of contradiction that $p_j \sqsubseteq s(p_i)$
and $s(p_j) \not \sqsubseteq s(p_i)$ both hold. Then any containment
mapping $\psi$ from $s(p_i)$ into $p_j$ should use some of the
predicates of $p_j$ starting with a //-edge (otherwise we would have a
containment mapping from $s(p_i)$ into $s(p_j)$ as well). But this is
not possible since for any //-subpredicate $st$, its incoming /-path
$l$ is incompatible (does not map) with the path $l'$ following the
main branch node (see Figure~\ref{fig:picture-skel}).
\end{proof}

The following result also follows similarly to Lemma~\ref{lem:skel-necessary}.
\begin{lemma} 
\label{l:extskelsuff}
There exists an \xpcap rewriting of an input  query $q \in \xppes$ using  a set of \xp views  iff there exists one using the extended skeletons of the views.
\end{lemma}
\begin{proof}
 The if direction is immediate.  For the only if direction, it suffices to see that since $q$ is an extended skeleton, any containment mapping from $q$ into the unfolding of the rewrite plan  will actually use only parts that are not violating the \xppes condition. This means that a containment mapping from $q$ into this plan also gives a containment mapping from $q$ into the corresponding plan using instead of the original views their extended skeletons. 
\end{proof}
By Lemma~\ref{l:extskelsuff}, assuming \xppes~input queries,  without loss of generality all views can be  assumed in the rest of this section to be  from \xppes~as well (when this is not the case, the views can be substituted by their extended skeletons before the DAG rewriting,   for instance at Step 3 in \algoname).  

A key result of our paper is the following (for readability purposes,  proof  given in Section~\ref{sec:proof1}):

%%always runs in PTIME. It also
\begin{theorem}%%[\xppes]
\label{th:completenessUF-es}
For any pattern $d$ in $\xppescap$, $d$ is union-free iff the algorithm~\proc{Apply-Rules} rewrites $d$
into a tree.
\end{theorem}
 From this, it follows immediately that:
\begin{corollary}[\xppes]
\label{c:completeness-es} Algorithm~\efficient is complete for rewriting $\xppes$ queries using \xp views. 
%$~~~~~~~~~~~~~~~~~~~~~~~~~~~~~~~~~~~~~~~~~$(2) Algorithm~\efficient is complete for \xppes~input queries.
\end{corollary}
We show next that relaxing the extended skeleton restrictions
leads to hardness for  union-freedom and rewriting using views. We consider here as middle-ground between extended skeletons and \xp the fragment  \xppdesc\ obtained from extended
skeletons by  allowing  predicates that are connected by
a //-edge to the main branch (such as in $v_2'$) and freely allowing  //-edges in these predicates. Obviously, $\xppes \subsetneq \xppdesc$.
 We denote by $\xppdesccap$ the fragment of $\xpcap$ in which only $\xppdesc$ expressions are used.

We first prove the following complexity lower bound for union-freedom: 
 \begin{theorem}%%[\xppes]
\label{th:hardnessUF-desc-2}
For a pattern $d$ in $\xppdesc^{\cap}$~, the problem of deciding if $d$ is union-free
is  coNP-hard.
%(2) For a pattern $d = \bigcap_j p_j$, where all $p_j$ are in \xp and
%akin, deciding if $d$ is union-free is coNP-hard.
\end{theorem}
\begin{proof}
\begin{figure*}[t]
\begin{center}
\includegraphics[trim=0mm 0mm 0mm 0mm, clip=true, scale=0.45]{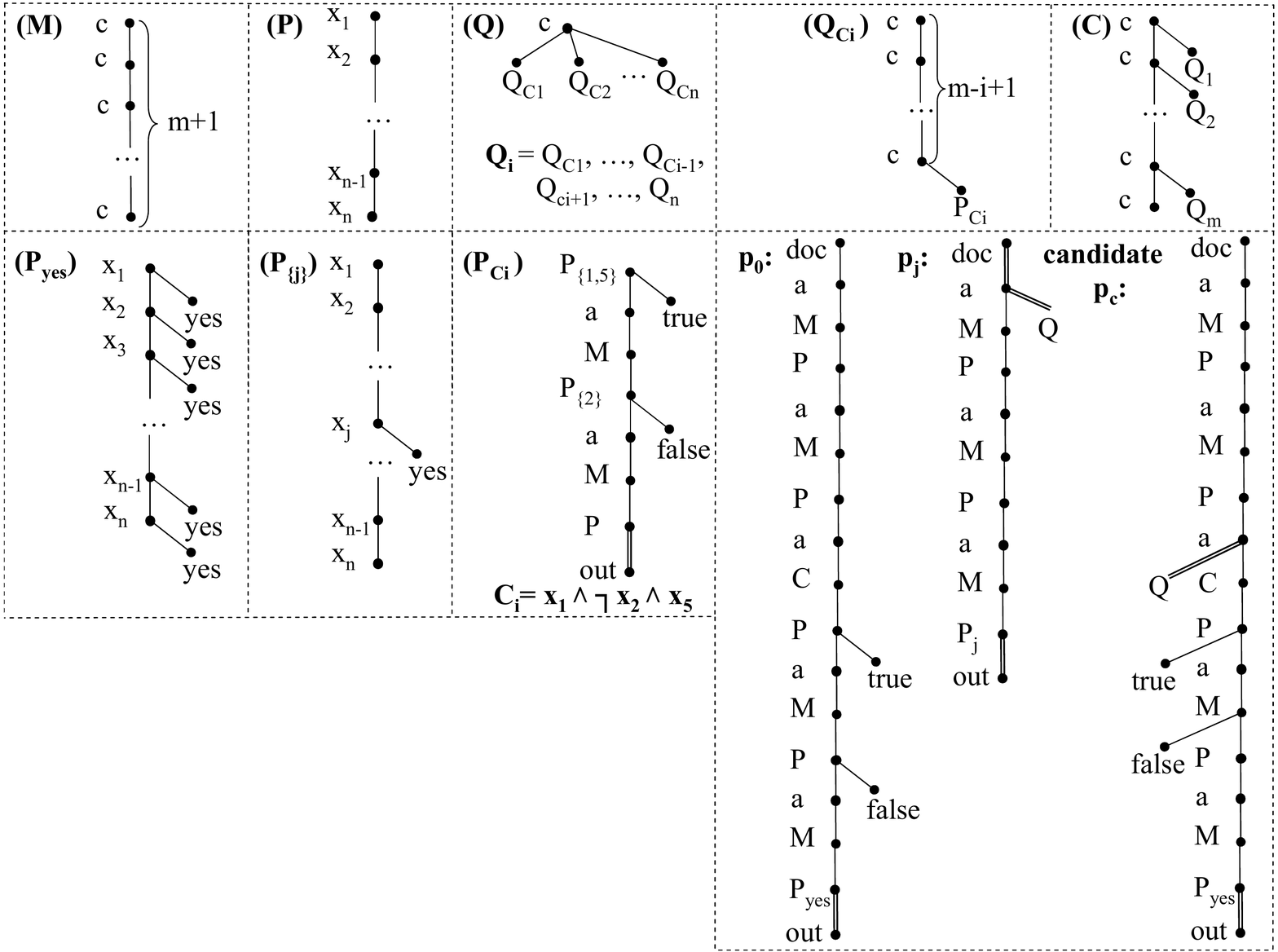}
\end{center}
\caption{The construction for coNP-hardness of union-freedom (\xppdesc). \label{fig:hardness1}}
\vspace{-3mm}
\end{figure*}
We use a  reduction from tautology of 3DNF
formulas, which is known to be coNP-complete. We start from a 3DNF
formula $\phi(\bar x) = C_1(\bar x) \vee C_2(\bar x) \vee \dots
C_m(\bar x) $ over the boolean variables $\bar x = (x_1,\dots x_n)$,
where $C_i(\bar x)$ are conjunctions of literals.

Out of $\phi$, we build patterns $p_0, p_1, \dots, p_n \in \xppdesc$ over
$\Sigma = \{x_1, \dots, x_n, a, c, yes, out\}$ such that the DAG pattern $d=p_0
\cap p_1 \cap \dots \cap p_n$ is union-free iff $\phi$ is a tautology.

We build the patterns $p_0$, $p_1, \dots, p_n$, based on the gadgets $P,  P_{yes}, P_{C_1}, \dots P_{C_m},  M, Q_{C_1}, \dots, Q_{C_m}, C, Q,$ and $P_X$ where $X$ can be any set of one, two or three variables (see Figure~\ref{fig:hardness1}).
 More precisely, these gadgets are defined as follows:

\begin{enumerate}
\item the linear pattern with $m+1$ $c$-nodes, $c/c/\dots/c$ (denoted $M$)
%\item $doc(A)//t//b$ (denoted $r$, the right branch in $d$)

\item the pattern $x_1/x_2/\dots/x_n$ (denoted $P$)

\item the pattern $x_1[yes]/x_2[yes]/\dots/x_n[yes]$ (denoted $P_{yes}$)

%\item $/x_k/a\ \mathcal{P}(x_k, \phi)/x_k/a\ \mathcal{P}(\neg x_k, \phi)$, where,
%for a given literal $l$, $\mathcal{P}(l,
%\phi)$ is a set of predicates $[C^{(l)}_1][C^{(l)}_{2}]\dots$ formed by
%considering all clauses $C^{(l)}_k \in \mathcal{C}_l$. This pattern is denoted $S_k$. By $S$ we denote the pattern $S_1/S_2/\dots/S_n$.

\item patterns $P_X$, where $X$ is a set of variables of size at most $3$, obtained from $P$ by putting a $[yes]$ predicate below the nodes labeled by the variables in $X$.

\item for each clause $C_i$,  the pattern $P_{X_t}[true]/a/M/P_{X_f}[false]/a/M/P//out$, where $X_t$ is the set of positive variables in $C_i$ and $X_f$ is the set of negated variables in $C_i$ (this pattern is denoted $P_{C_i}$ ). For instance,  for $C_i=(x_1 \wedge \bar{x_2} \wedge x_5)$, we have the pattern $P_{C_i} =P_{\{1,5\}}[true]/a/M/P_{\{2\}}[false]/a/M/P//out$.

\item for each clause $C_i$, $Q_{C_i}$ denotes the predicate $[c/c/c/\dots/c[P_{C_i}]]$, with $m-i +1$ $c$-nodes,

\item for each $C_i$, the predicate $Q_i=[Q_{C_1}, \dots Q_{C_{i-1}}, Q_{C_{i+1}}, \dots, Q_{C_m}]$, that is the list of all $Q_{C_j}$ predicates for $j \neq i$.
    \item the  pattern $c[Q_1]/c[Q_2]/c[Q_3]/\dots c[Q_m]/c$ (denoted $C$),  the predicate $Q=[Q_{C_1}, \dots, Q_{C_m}]$.
\end{enumerate}

The $n+1$ patterns are then given  in the last section of Figure~\ref{fig:hardness1}.

First, note that no inheritance of predicates occurs in these patters. $Q_{C_i}$ predicates are not inherited in the $C$ part of $p_0$ because that would require some $x_1$-label to be equated with the $c$-label. Similarly, the $P_{yes}$ part of the main branch does not put implicit $Q_{C_i}$ predicates at $c$-nodes either.

We argue that the candidate interleaving  $p_c$ such that $p_c\equiv d$
is unique: $p_c$ is obtained by the code $i$ corresponding to the main
branch of $p_0$, and the function $f_i$ that maps the first $a$-node (the one with a predicate $[.//Q]$)
of each pattern $p_1, \dots, p_n$ in the same image as the third $a$-node of $p_0$ (the parent of the $C$ part). This
is the interleaving that will yield the ``minimal'' extended skeleton
(namely the one of $p_0$), since nodes with a $[yes]$ predicate are coalesced with $p_0$ nodes having already that predicate. All others would at least have additional $[yes]$ predicate branches and even longer main branches and thus cannot not map into $p_c$. Hence no other interleaving can contain $p_c$.

We show in the following that $p_c$ will contain (and reduce) all other
interleavings of $p_0 \cap \dots \cap p_n$ iff $\phi$ is a tautology. Moreover, it is easy see that $p_c$ contains some interleaving $p$ if and only if its $[.//Q]$ predicate can be mapped at the third $a$-node from the root in $p$.

Note now that $p_c$ will contain all
interleavings $p$ that  for at least some pattern $p_j$ ``put'' its  first $a$-node either below or in the
third $a$-node of $p_0$.  This is because $[.//Q]$ would be either explicitly present at the third $a$-node in $p$ or it would be inherited by this node from some $a$-labeled descendant.   %Note that this is the ``lowest'' choice for these $a$-nodes.

So, the interleavings  that remain be considered are those described by a function
$f_i'$ which takes \emph{all} the first $a$-nodes from $p_1, \dots, p_n$ higher in $p_0$, i.e. in
either the first or the second $a$-node of $p_0$. Each of these interleaving
will basically make a choice between these two $a$-nodes.

For some $p_j$, by choosing to coalesce its first $a$-node with the first $a$-node of $p_0$ we  get an $[yes]$
predicate at the $x_i$ node of the \emph{true} $P$ part of
$p_0$. Similarly, by coalescing with the second $a$-node we get an $[yes]$
predicate at the $x_i$ node of the \emph{false} $P$ part of
$p_0$. So, these $n$ individual choices of where to coalesce  $a$-nodes amount to a truth assignment for the $n$ variables, and in each interleaving the $yes$ predicate will indicate that assignment.

Recall that in order for $p_c$ to contain such an interleaving $p$, it must be possible to map the predicate
$[.//Q]$ of the third $a$-node
of $p_c$  at the third $a$-node of $p$.

We can now argue that $p_0 \cap p_1 \cap \dots \cap p_n \sqsubseteq p_c$ iff $\phi$ is a tautology. The if direction (when each truth assignment $t$ makes at least one clause $C_i$ true) is immediate. For a truth assignment with clause $C_i$ being true, in the corresponding interleaving $p$, the $P_{C_i}$ predicate will hold at the last $c$-node in the $C$  part, hence the $Q_{C_i}$ predicate will hold at the $i$th $c$-node in $C$. Since all other $Q_{C_j}$ predicates, for $j\neq i$, were already explicitly present at this $i$th $c$-node, it is now easy to see that the $[.//Q]$ predicate would be verified at the $a$-labeled ancestor. Hence there exists a containment mapping from $p_c$ into $p$.

The only if direction is similar. If for some truth assignment, none of the clauses is \textsc{true} (in the case $\phi$ is not a tautology), then it is easy to check that $p_c$ will not have a containment mapping into the interleaving $p$ corresponding to that truth assignment. This is because the $[.//Q]$ predicate would not map at the third $a$-node in $p$. 
\end{proof}

We next prove that the coNP lower bound is tight, showing  that union-freedom for  patterns from \xppdesc\  is in coNP (recall that the problem was shown to be tractable for extended skeletons). %For readability purposes, the proof is given in Section~\ref{sec:hardnessUF-desc-2}.
\begin{theorem}%%[\xppes]
\label{th:hardnessUF-desc-1}
For a pattern $d$ in $\xppdesc^{\cap}$ , the problem of deciding if $d$ is union-free
is in  coNP.
%(2) For a pattern $d = \bigcap_j p_j$, where all $p_j$ are in \xp and
%akin, deciding if $d$ is union-free is coNP-hard.
\end{theorem}
\begin{proof}%({\bf Theorem~\ref{th:hardnessUF-desc} (1)})
To show that union-freedom is in coNP we use the following approach: we show that one can always build an interleaving $p_c$ that is the \emph{unique} candidate for $p_c \equiv d$. Then, we can use an argument similar to
the one used in the proof of Theorem~\ref{thm:containment-dag-tree}, to check in coNP if $d\sqsubseteq  p_c$.

If there are no views that have only /-edges in their main branch, we start can with $d$ being a DAG pattern as the one illustrated in Figure~\ref{fig:completeness4}. This is the result of applying  R1 steps until saturation,  on what will be  the root and result tokens of interleavings of $d$.  Let $t_r$ denote the root token (ending with node $n_r$) and let $t_o$ denote the result token (starting with node $n_o$).  We have some branches in parallel $i_1, \dots, i_k, i_{k+1}, \dots, i_{k+l}$, starting and ending at various nodes of $t_r$ and $t_o$, with the first $k$ ones being connected to $n_r$ and $n_o$. We proceed towards building $p_c$. 

Let us first  assume that $l=0$, i.e., all the branches in parallel are connected to $n_r$ and $n_o$. 

 We know by Lemma~\ref{lem:skel-necessary} that $d$ is union-free only if its  extended skeleton $s(d)$ is union-free. This in turn means that there exists some branch in $d$, say $i_1$, such that all other $s(i_j)$ map into $s(i_1)$ (Lemma~\ref{lem:n-comp-skel}). So, all the branches in parallel map their extended skeletons into the branch $i_1$, while predicates starting by a //-edge may not map (these are allowed in \xppdesc, contrary to \xppes).

 It is now easy to see that $p_c$ can only have an extended skeleton of the form $s(t_r)//s(i_1)//s(t_o)$ (this is the minimal skeleton). Note that several choices for mapping each $s(i_j)$ into $i_1$ may be available, and each such mapping can be seen as a way of coalescing $i_j$ nodes with the $i_1$ ones. In fact, interleavings that do not use one of these choices cannot lead to $p_c$ (they will no longer yield the minimal candidate extended skeleton).

Among these coalescing choices that do not modify the candidate extended skeleton, it now easy to pick the one that gives the unique candidate $p_c$.  For each $i_j$ let $\psi_j$ denote the mapping of $s(i_j)$ in $i_1$ that uses the \emph{highest possible image} for each token of  $s(i_j)$. We build $p_c$  from $d$ by transformation steps that coalesce each $n \in \mb{i_j}$ with the node $\psi_j(n) \in \mb{i_j}$.

 We argue that the pattern $p_c$ we obtained  is the only interleaving candidate  for $p_c \equiv d$.
\eat{
 \begin{claim}
 $d$ is union-free only if $p_c \equiv d$.
 \end{claim}
\textbf{Proof of the claim:}  
} This is because all other interleavings that do not modify the candidate extended skeleton will be subsumed (as contained interleavings) by this one. This is because all the predicates that were missing from $i_1$ and were added by coalesce steps must start by a //-edge. If they would be added  below this  first possible image, they would anyway be inherited by the main branch nodes above.

 Now, since we have a clear candidate $p_c$, obtained in polynomial time, we can guess a witness interleaving $p_w$ of $d$ such that $p_w \not \sqsubseteq p_c$ in polynomial time.
 
 We now consider the case when, after the initial phase of R1 steps,  $l\neq 0$, i.e., some of the branches in parallel are not from  $n_r$ to $n_o$. In this case, we  can advance towards the unique candidate interleaving and a new DAG as the one discussed above by  the following coalescing choices. For $j$ from $k+1$ to $k+l$,  for each $i_{j}$,   let $n_{j}^r$ denote the node in $t_r$ that is sibling of  the first node in $i_j$, and let $n_{j}^o$ denote the node of $t_o$ that is ``parent-sibling'' of $i_j$ (they have the same child node).  Let  $pr_j$ denote its maximal token-prefix such that \emph{its extended skeleton} $s(pr_j)$ maps in $\tp{d}{n_{j}^r/\dots/n_r}$.  (Note that some of the predicates - among those starting by a //-edge - may not fully map.)

 Then, for each $i_j$, by $sf_j$ we denote the maximal token-suffix such that $sf_j$ \emph{fully} maps in $\tp{d}{n_o/\dots/n_{j}^o}$. We can thus write each $i_j$ as $i_j=pr_j//m_j//sf_j,~ for~ j=k+1,l+1$. If, for some $i_j$, $pr_j$ and $sf_j$ overlap then in this case the $m_j$ part is considered empty.

Now, we consider the  DAG pattern $d'$ obtained from $d$ by replacing each $i_j$ branch by the branch $m_j$, but connected this time by //-edges to $n_r$  and $n_o$.  We   argue that $d$ is union-free only if $d'$ is union-free, and that if $d'\equiv p_c$, then $p_c$ is also the only interleaving candidate for $p_c \equiv d$. This is because all other interleavings are either subsumed by those of $d'$, or change the extended skeleton of $t_r$, or add some predicates on $t_o$ that will not be present in all interleavings (the candidate must end exactly by $//t_o$). From this point on, we can continue with the initial  line of reasoning, over $d'$.

Finally, the case when at least one of the views has no //-edges in the main branch can be handled  by a similar construction of the unique candidate interleaving. 
\end{proof}

%% Please note that we have a PTIME Turing reduction from union-freedom
%% to the rewriting problem, just by using an oracle for union-freedom in
%% order to answer the containment test in
%% \algoname. However, this does not provide any easy map reduction. 
It is now easy to show that the same  complexity lower bound holds for the rewriting problem. 
%% %\vspace{-1ex}
\eat{
\begin{theorem}\label{thm:hardness-rewriting-desc-1}
For queries and views from
\xp, the rewriting problem is in coNP.
%(2) For a query and views from $\xpp$, deciding the existence of a
%rewriting $r$ such that the patterns intersected in \unfold{r} are akin is
%coNP-complete.
\end{theorem}
\begin{proof}%({\bf Theorem~\ref{thm:hardness-rewriting-desc} (1)})
%\textbf{Point (1).} First, let us prove the upper bound.
For each node of \mb{q} in which some of the views ${\cal V}_1
\subseteq {\cal V}$ map, it is enough to guess one interleaving of
$\bigcap_{{\cal V}_1} v_j$ in which $q$ does not map.  If we put
together a polynomial number of polynomially large witnesses, they
make up a polynomial witness for the entire problem. In other words,
one can verify in polynomial time that there is no rewriting.
\end{proof}
}
\begin{theorem}\label{thm:hardness-rewriting-desc-2}
 For queries and views from
$\xppdesc$~, the rewriting problem is  coNP-hard.
%(2) For a query and views from $\xpp$, deciding the existence of a
%rewriting $r$ such that the patterns intersected in \unfold{r} are akin is
%coNP-complete.
\end{theorem}
\begin{proof}
We can use the same construction as in the proof
of Theorem~\ref{th:hardnessUF-desc-2}, for a reduction from tautology
of a formula $\phi$. We define $n+1$ views, $v_0 =p_0$, $v_1=p_1,
\dots, v_n=p_n$. We define $q$ as $q=p_c$, for $p_c$ being the unique
candidate interleaving for the DAG $d=\dagp{\unfold{v_0 \cap v_1 \cap
\dots \cap v_n}}$. Moreover, it is easy to see that the only rewrite
plan that has chances to be a rewriting is $r = v_0 \cap v_1 \cap
\dots \cap v_n$ (the output node of each view can only be mapped in
the output node of $q$). From this, it follows that $r$ is an
equivalent rewriting iff $d \equiv q$ iff $d$ is union-free iff the
formula $\phi$ is valid. This shows that deciding the existence of an
equivalent rewriting $r$ for queries and view from \xppdesc\ is
coNP-hard.
\end{proof}
%
%%%%%%%%%%%%%%% END OF PART to be rewritten.
\eat{
\vspace{-3mm}
\noindent  \textbf{Point (2).}
The fact that the problem is in coNP was already discussed
(see proof of Theorem~\ref{thm:hardness-rewriting-desc}(1)).

For coNP-hardness, we will use the same construction as in the proof
of Theorem~\ref{th:hardnessUF-desc}(2), for a reduction from tautology
of a formula $\phi$.
%%BOGDAN: to rewrite when the 4.7(2) is available.
% We define two views, $v_1 =p_1$ and $v_2=p_2$,
%with $p_1$ and $p_2$ as in Figure~\ref{fig:hardness2}. We define $q$
%as $q=p_c$, where $p_c$ (see Figure~\ref{fig:hardness3}) is the unique
%candidate interleaving for the DAG $d=\dagp{\unfold{v_1 \cap
%v_2}}$. Moreover, it is easy to see that the only rewrite plan that
%could be a rewriting is $r = v_1 \cap v_2$, since $p_1$ and $p_2$ map
%only in \out{p_c}. Also note that the two views are akin. From this,
%it follows that $r$ is an equivalent rewriting iff $d \equiv q$ iff
%$d$ is union-free iff the formula $\phi$ is a tautology. This shows that
%deciding if there exists an equivalent rewriting $r$ such that the
%intersected patterns in the unfolding are akin is coNP-hard.
\end{proof}
}

\section{Tractability frontier: Rewrite-plans for PTIME}
\label{sec:frontier2}
We also identify a large class of rewrite plans that lead to
PTIME completeness.

We say that two (or several) tree patterns are \emph{akin} if their
root tokens have the same main branch codes. For instance, while the
views $v_1$ and $v_2$ from our example are not akin, $v_1$ is
akin to:
%\vspace{-0.15cm}
$$v_2': \snippet{doc(``L'')//figure[.//caption//label]//subfigure/image[ps].}$$
Under the assumption of akin views, we can relax the syntactic restrictions  of the \xppes\ fragment for tractability and accept
the class of patterns $\xppdesc$.
\eat{
obtained from extended
skeletons by  allowing  predicates that are connected by
a //-edge to the main branch (such as in $v_2'$) and freely allowing  //-edges in these predicates. Obviously, $\xppes \subsetneq \xppdesc$.
 We denote by $\xppdesccap$ the fragment of $\xpcap$ in which only $\xppdesc$ expressions are used.} 
Our main result for restricted rewrite plans is  following (for readability purposes, the proof is given in Section~\ref{sec:completenessUF-desc-akin}):

\begin{theorem}\label{th:completenessUF-desc-akin}
For DAG patterns of the form $d = \bigcap_j p_j$, where all $p_j$ are in \xppdesc~ and akin,
$d$ is union-free iff the algorithm~\proc{Apply-Rules} rewrites $d$
into a tree.
\end{theorem}
 From this, it follows that:
%% Is there any more readable way to state this theorem?? -- Nicola
\begin{corollary}[\xppdesc]
\label{c:completeness-//}
\efficient always finds a rewriting for \xppdesc\ ,
%%(and it does it in PTIME),
provided there is at least a rewriting $r$ such that the patterns intersected in
\unfold{r} are akin.
\end{corollary}

Once again, we can show that dropping the restriction of akin patterns leads to hardness for  both union-freedom and rewriting using views.
 \begin{theorem}%%[\xppes]
\label{th:hardnessUF-desc-3} For a pattern $d = \bigcap_j p_j$, where all $p_j$ are in \xp and
akin, deciding if $d$ is union-free is coNP-hard.
\end{theorem}
\begin{proof} Similar to the proof of Theorem~\ref{th:hardnessUF-desc-2}.
\end{proof}
It follows easily that the same lower-bound holds   for the rewriting problem as well.
\begin{theorem}\label{thm:hardness-rewriting-desc-3}
For queries and views from
\xpp, deciding the existence of a rewriting $r$ such that the patterns intersected in \unfold{r} are akin is
coNP-hard.
\end{theorem}
\begin{proof} Similar to the proof of Theorem~\ref{thm:hardness-rewriting-desc-2}.
\end{proof}

\section{Nested Intersection}\label{sec:nestedinter}
%\bogdan{say what this new formalism captures (what kind of rewritings.)}
We have considered so far the  \xpcap formalism, for rewritings that may first compensate the views, may then perform an intersection, and finally may compensate the result of this intersection.  Such rewritings may not be sufficient for certain input queries, as illustrated in the following example. 
\begin{example} 
Let us assume that we have a view $v_1$, that
retrieved all sections from papers :%%\vspace{-0.5cm}
$$v_1: \snippet{doc(``L'')//paper//section}$$
The result of $v_1$ is stored in the cache as a materialized view, rooted at 
an element named $v_1$.
%% and containing nodes that preserved the identifiers
%% of the original document \snippet{L}.
Later, the query processor had to answer another XPath $v_2$ looking
for sections having at least one theorem:
%%appearing in sections with theorem statements:
%%\vspace{-0.15cm}
$$v_2: \snippet{doc(``L'')//section[theorem]}$$%%\vspace{-0.15cm}$$
The result of $v_2$ is not contained in that of $v_1$, so it is also
executed and its answer cached.

A third view has to be processed then, looking for library images in figures:
$$v_3: \snippet{doc(``L'')/lib//figure/image}$$
and $v_3$ has to be computed as well, since it cannot be answered using the results of $v_1$ and $v_2$.

Let us now look at an incoming query $q$, asking for all
 images in figures that appear in sections of papers with theorems:
%%\vspace{-0.15cm}
$$q: \snippet{doc(``L'')/lib//paper//section[theorem]//figure/image}$$%%\vspace{-0.15cm}$$
However, by intersecting the results of the first two views  we get the right section elements. Then, we can further navigate inside them, and intersect their image descendants with the images selected by $v_3$. So  one can build a
rewriting equivalent to $q$:
%%\vspace{-0.15cm}
%% $r_2$: \snippet{doc(``$v_1$'')/$v_1$/image/file $\cap$ doc(``$v_2$'')/$v_2$/image/file}
$$r_2: \snippet{(doc(``$v_1$'')/$v_1$ $\cap$ doc(``$v_2$'')/$v_2$)//figure/image $\cap$ doc(``$v_3$'')/$v_3$}$$
\end{example}
In the following, we will consider an extension \xpint of $\xp$ with respect
to intersection, which includes \xpcap, and allows us to handle rewritings with arbitrary many levels of intersection and compensation over views (i.e., \emph{nested intersections}). The grammar of \xpint is obtained from that of \xpcap by adding the rule
%\vspace{-1mm}
\begin{eqnarray*}
\jpath &::=& \ipath~|~(\jpath) | ~\jpath \cap \cpath  |  ~(\jpath)/\rpath~|~(\jpath)//\rpath, 
\end{eqnarray*}
having the straightforward semantics.

Let us notice that \xpint queries can also be represented by DAG
patterns having the particular property that if there are two distinct main
branches from a node $n_1$ to another node $n_2$, then $n_1$
must be the root of the DAG.

As before,  given a query
$r \in \xpint$ over the view documents $D_{\cal V}$, we define
\unfold{r} as the \xpint query obtained by replacing in $r$ each
$\doc{``V"}/V$ with the definition of $V$. The notion of rewriting also extends naturally.
\begin{definition}\label{def:equiv-nested-rw}
 For a given XML document $D$, an \xpint query $q$ and \xpint views ${\cal
 V}$ over $D$, an \xpint-\emph{rewrite plan} of $q$ using ${\cal V}$ is
 a query $r \in \xpint$ over $D_{\cal V}$.

 If $\unfold{r} \equiv q$, then we also say r is
 an \xpint-\emph{rewriting}.
\end{definition}
We can show that this more general flavor of rewritings is no more expensive to evaluate (over view documents) than the \xpcap ones. 
\begin{lemma}\label{l:nested-rwplan-eval}
An \xpint-rewrite plan can be evaluated over a set of view documents
$D_{\cal V}$ in polynomial time in the size of $D_{\cal V}$.
\end{lemma}
\begin{proof}[Idea] 
This a direct consequence of Courcelle's theorem~\cite{Courcelle90} and generalizations thereof~\cite{FlumFG02}, on evaluating in linear time (data complexity) monadic second-order (MSO) formulas on trees or bounded tree-width structures. MSO is known to capture the \emph{navigational XPath} outputting sets of nodes~\cite{GottlobK02},  a fragment that strictly subsumes \xpint. 
\end{proof}
Equally important,  the rewriting problem stays in the same
complexity class, even if the expressivity of the rewriting language
increases.
\begin{theorem}\label{thm:hardness-xpint-rewriting}
Deciding the existence of an \xpint-rewriting for a query and views from
$\xppdesc$ is coNP-complete.
\end{theorem}
\begin{proof}Similar to the proof of Theorem~\ref{thm:coNPcompleteness}.
\end{proof}

%\bogdan{Say what remains open, what is likely hard.}
%% Let query $q$ be a query in \xp whose pattern is $p$ and $n_1, n_2 \in
%% \mb{p}$.  Let $a$ be the node immediately following $n_1$ in
%% \mb{p}. We denote by $\fragment{n_1}{n_2}$ the subgraph of $p$
%% formed by the nodes on \mb{p} from $a$ to $n_2$, together with all
%% their predicates.

%\bogdan{put running example for this section.}

For the purposes of this section, we introduce the notion of \emph{rewriting graphs}, which are similar
to DAG patterns, with the following differences:

\begin{itemize}
\item in addition to $\Sigma$ nodes, they  have \emph{view nodes} with labels of
the form $\doc{``v"}$ where $v$ is a symbol from a set of views $\cal
V$; a view node cannot have any incoming edge.
\item they do not have a distinguished node \mbox{{\small \sc ROOT}}.
\end{itemize}
 We also generalize the
function \unfold{r} to return the
\xpint query that corresponds to the following transformation: replace every  $\doc{``v"}/v$ with the query that defines $v$, coalescing the newly introduced node corresponding to  \out{v} with its unique main branch child.
This leads to a DAG pattern, since the views are assumed to be evaluated over the same, unique document.

Starting from the pattern of the input query $q$ and the views $\cal
V$, we will build  step by step  a rewriting graph \candRw{q}{\cal V}
that we call the rewriting candidate. We will then prove that the
obtained candidate is \emph{minimally containing} w.r.t. $q$, i.e.
\candRw{q}{\cal V} is contained in any query over the view documents that contains $q$. This
guarantees completeness in the sense that if an \xpint rewriting
exists, \xpath{\candRw{q}{\cal V}} is also a rewriting.

%For a DAG pattern $p$, we designate by $\cup_i p_i = p$ the union
%of all its interleavings.

\vspace{1em}
\noindent {\bf Algorithm BuildRewriteCandidate}\\
\noindent
input: query $q \in \xp$, set of views $\cal V$ defined by \xp queries\\
output: candidate rewriting \candRw{q}{\cal V}

\begin{enumerate}
\item set $r$ to $\pattern{q}$

\item let $S$ $\gets$
$\{ (v_i,o)\; | \textrm{ if }% d_i=\dagp{\unfold{V_i}}\equiv nf(d_i)=\cup_j d_{ij}
%\textrm{ and } p \equiv nf(p)= \cup_t p_t,\\
%~\hspace{4em}
\textrm{there is a mapping } h
%\textrm{ from every } 
\textrm{ from } v_i \textrm{ into  } q, o = h(\out{v})\}$

%% \item let $S$ $\gets$
%% $\{ (V_i,o_{i,k})\; | \textrm{ there is a mapping } h
%% \textrm{ from } d_i=\dagp{\unfold{V_i}}   \textrm{ into } p,\\
%% ~\hspace{4em} o_{i,k} = h(\out{d_i}) \}$

\item\label{it:loop-add-paths} foreach $(v,o)$ in $S$ \\
\indent add to $r$ a new view-node
     labeled $\doc{``v"}$, with a child labeled $v$ connected by a /-edge to $o$

%%     \item\label{it:notlast} if $i < |L|$ then
%%       \begin{itemize}
%%         \item let $(d',o')$ be the element at position $i+1$ in $L$
%%         \item if $o \neq o'$ then append
%%           \fragment{o}{o'}
%%           to $r$ .
%%       \end{itemize}
%%     \item\label{it:last} else append \fragment{o}{\out{p}} to
%%       $r$ %% the last part of p

\item keep in $r$ only paths accessible
  starting from view-nodes %$v \in \cal V$
\item if \out{q} is not in $r$, then \textbf{fail}
\item set $\out{r}$ to \out{q}
\item \textbf{return}  \candRw{q}{\cal V} := r  %~~~~~(this is \candRw{q}{\cal V} )
\end{enumerate}

%\vspace{1em}
\noindent {\bf Algorithm NestedRewrite}

\vspace{-2mm}
\begin{enumerate}
\item \candRw{q}{\cal V} $\gets$ {\bf BuildRewriteCandidate}($q, \cal V$)
\item\label{it:equiv-test} if $\unfold{\candRw{q}{\cal V}} \equiv q$ then output $\candRw{q}{\cal V}$ \\
      else {\bf fail}
\end{enumerate}

As the step (7) of \algoname, the equivalence test at step (2) uses Lemma~\ref{lem:cap_sub_cup}.
\eat{
and  the following one (similar to Lemma~\ref{lem:equiv_tree_union}).
\begin{lemma}\label{l:cnt-union-tp}
Let $p = \cup_i p_i$ and $q = \cup_j q_j$ be two finite unions of tree
patterns. Then $p \sqsubseteq q$ iff $\forall i, \exists j$ s.t. $p_i
\sqsubseteq q_j$.
\end{lemma}
}
We show next that our  algorithm  always produces the
\emph{minimally containing rewriting}, in the following sense.

\begin{lemma}\label{l:necessary-containment}
For $q \in \xp$, $\cal V$ a set of \xp views and $r'$ a
rewriting graph, if $q \sqsubseteq \unfold{r'}$, then
$\unfold{\candRw{q}{\cal V}} \sqsubseteq
\unfold{r'}$.
\end{lemma}
\begin{proof}
Let $d$ be $\dagp{\unfold{\candRw{q}{\cal V}}}$, $d'$ be
$\dagp{\unfold{r'}}$ and $p$ be $\dagp{q}$. As $q \sqsubseteq
\unfold{r'}$.  We can write $d = \cup_i d_i$, $d' = \cup_j d'_j$ and
$p = \cup_k p_k$ as the union of their interleavings.  By
Lemma~\ref{l:cnt-union-tp}, there are containment mappings $h_i$ from
each $d_i$ into some $p_k$ and $h'_j$ from every $d'_j$ into some
$p_k$.
%%structural induction

We  show by structural induction that there is a containment
mapping from every $d'_j$ into some $d_i$.

(1) Suppose $r'$ is of the form $\doc{V}/V$, where $V \in \cal
    V$. Since every $d'_j$ maps into some $p_k$ and $\candRw{q}{\cal
    V}$ is built using all possible mappings of the views, there is at
    least one occurence of $V$ in $d$ connected to a node
    $o$. Moreover, since $q \sqsubseteq\unfold{r'}$, there must be
    such an occurrence in which $o$ actually corresponds to \out{p}.
    Hence there is trivially a containment mapping from every $d'_j$
    into some $d_i$.

(2) Consider now the case in which $r'$ corresponds to a query of the
    form $\alpha/x$, where every interleaving $\alpha_j$ of the
    pattern $\unfold{\alpha}$ has a mapping $m^1_j$ into some $d_i$
    and $x$ is a relative path. Let $n_\alpha$ be the last node on the
    main branches of $\unfold{\alpha}$. Let $d'_j$ be $\alpha_j/x$. By
    construction, the subgraph of $d_i$ accessible starting from
    $m^1_j(n_\alpha)$, call it $t_{ij}$ is a tree that is isomorphic
    to a subtree \id{st} of $h'_j(d'_j)$, where $h'_j$ is the mapping
    from $d'_j$ into some interleaving $p_k$. Hence the identity
    function \id{id} is a mapping from \id{st} into $t_{ij}$,
    therefore $m^1_j$ can be extended to a mapping $m^2_j$ from $d'_j$
    into $d_i$. If \out{d'} is in the pattern of $\alpha$, then from
    the containment of $q$ into $\alpha/x$ we can also infer the
    containment of $q$ into $\alpha$ and the induction hypothesis
    guarantees that $m_1$ can be chosen such that it is a containment
    mapping. Otherwise, \out{d'} is part of the pattern of $x$ and
    $h'_j(\out{d'}) = \out{p}$. But then $m^2_j(\out{d'}) =
    h'_j(\id{id}(\out{d'})) = \out{d}$ and $m^2_j$ is a containment
    mapping.

(3) Suppose that $r'$ is of the form $(\alpha\ \cap\ \doc{V}/V)$,
    where every interleaving $\alpha_j$ of the pattern of
    $\unfold{\alpha}$ has a mapping $m^1_j$ into some $d_i$. Let
    $n_\alpha$ be the last node in the main branches of
    $\dagp{\unfold{\alpha}}$ and $n'_\alpha = m^1_j(n_\alpha)$, unique
    for all $j$. Then again we can find $\doc{V}/V$ in $d$ that was
    added to \candRw{q}{\cal V} by a set of mappings $\{m'_l\}$ into
    $p$ that agree with $h'_j$ on $\unfold{\doc{V}/V}$ and such that
    $m'_l(n'_\alpha) = \out{p}$, $\forall l$.  Hence $n'_\alpha$ is
    the output node of $d$ and $m^1_j$ can be extended to a mapping
    $m^2_j$ from $d'_j$ into $d_i$ such that $m^2_j(\out{d'}) =
    n'_\alpha = \out{d}$.

Then by Lemma~\ref{l:mapping-suff},  we have that $d \sqsubseteq d'$.
\end{proof}

\begin{theorem}\label{th:mincontaining}
For any query $q \in \xp$ and any \xp rewriting $r'$ such that
$\unfold{r'} \sqsupseteq q$, $\unfold{\candRw{q}{\cal V}} \sqsubseteq
\unfold{r'}$.
\end{theorem}
%\begin{proof}
%Follows from Theorem~\ref{thm:dag-xpint-equiv} and
%Lemma~\ref{l:necessary-containment}.
%\reminder{We need to extend theorem~\ref{thm:dag-equiv} to \xpint}
%\end{proof}

\noindent {\bf Remark.} Note that  algorithm BuildRewriteCandidate runs in PTIME,
but NestedRewrite is  worst-case exponential, which is the best we can hope for, given the hardness result of Theorem~\ref{thm:hardness-xpint-rewriting}.

It is remains open whether for \xpint the complexity of union-freedom and view-based rewriting drops to PTIME under the input query restrictions (extended skeletons) or the rewrite plan ones discussed in the previous section. This represents one of our main directions for future research  on XML view-based query rewriting, and challenging one. 
%(** to be detailed: number of mappings with distinct images, sorting...
%complexity of the containment test... **)

%%Bogdan: commented out this extensions for now:
%%\vspace{1em}
%\reminder{Some ideas for extensions:}
%
%Using Lemma~\ref{l:containment-tree-dag}, we can extend algorithm
%{\bf BuildRewriteCandidate} to work for queries from \xpint, provided
%the views are still in \xp (i.e. no intersection in the views).
%
%One problem, which shows up also in {\bf BuildRewriteCandidate}, is
%that it is hard to tell how redundant the rewriting is. If the query
%has intersection, it can even get worse, that is there could be
%compensation containing intersection which is kept without being
%needed.
%
%If both queries and views are from \xpint, then we cannot decide
%containment based on mappings. One option would be to keep using
%the same algorithm {\bf BuildRewriteCandidate}, which would be
%an incomplete procedure in this case. Or, if we want completeness,
%we need an algorithm for deciding containment of DAGs, and that
%would probably be very expensive.
%

\section{Optimization opportunities}
\label{sec:implementation}
The implementation of the \algoname algorithm revealed several possible refinements, and we discuss in this section the most relevant high-level optimizations that enabled performance improvements in practice. We stress that, by this implementation,  our main goal   was to provide  a proof-of-concept prototype, illustrating advantages and scalability.  Many other directions for optimization, including low-level aspects and supporting data structures, remain to be explored. 

 %a set of adjustments/refinements of the algorithms described in previous sections. Moreover, high and low level optimizations may yield dramatic performance improvements in practice.

\subsection{Optimizing plans}

We start by introducing the theoretical foundations  of a PTIME technique for partially minimizing redundancy in rewrite plans, without paying the price of full minimization. We show that, at line~\ref{li:rw-build-compensations} of \algoname, for a given view participating in a plan,  we can always choose  one unique  compensation instead of all the valid compensations.  To this end, for a view $v$ that maps into $p$, $\func{bestcomp}(v,p)$ returns $v'$ formed by $v$ plus some compensation, such that $v'$ is contained in  all other compensated versions of $v$. 

\begin{lemma} For a given prefix $p$ of the input query and a given view $v$, $\func{bestcomp}(v,p)$ is well-defined.  It is sufficient to use in algorithm \algoname, for the corresponding rewrite plan, only $v$'s $\func{bestcomp}(v,p)$ compensation.
\end{lemma}
\begin{proof}[Idea]
It is immediate that compensations that are subsumed (contain) by others will be redundant in the rewrite plan and can be safely discarded  (this could be performed by applying R7 to saturation on the original plan).  However, the best compensation is uniquely defined for each prefix and view, and can be constructed and used directly in the plan, in this way avoiding expensive rule testing and application. 
More precisely, the compensated view $\func{bestcomp}(v,p)$ corresponds to the longest possible compensation -- in terms of main branch nodes of $p$ -- that can be applied on $v$ in \algoname, which in turn corresponds to the highest possible mapping image of $v$ in $p$.   It can be easily verified that the $\func{bestcomp}(v,p)$ compensation obtained in this way will be contained in all the other compensations of $v$.
\end{proof}

We  also identified as potentially beneficial a test  for the existence of a sub-plan equivalent to the input query formed by \emph{akin} views, even for  input queries outside the fragment \xppdesc. Recall that \algoname is complete for \xppdesc\  if limited to  plans using akin views.   Intuitively, rule application on the akin plan is likely  to advance faster towards a conclusion. Note however that whenever the resulting  DAG is not a tree we cannot conclude whether the initial plan is  union-free, and this plan has to be processed as well.
 
Regarding the simplification of plans, due to the  compensation steps,  many redundant predicates may be present on nodes of the initial plan. These can be  pruned out, before the \proc{Apply-Rules} subroutine,  in this way making the various mapping tests for rule applications lighter. 

Finally, we can use some efficient tests to detect plans that cannot be equivalent to the input query and can be discarded before the \proc{Apply-Rules} subroutine. For example, in the case of input queries with only /-edges in the main branch, a view having that same main branch must be available. Similar tests on the main branches of the root and output tokens of the query and the plan's views can be  used to discard plans.

\subsection{Optimizing rule applications}

First, the implementation of the rewrite rules uses as a key building block the ``bottom-up''  computation of mappings. While for mappings into tree structures we can rely directly on a dynamic programming approach as in~\cite{DBLP:journals/jacm/MiklauS04},   extending this approach for mapping of trees into  DAGs requires to maintain the DAG nodes in topological order.

Second, complementing the steps discussed in the proof of Lemma~\ref{l:termination} -- for testing in polynomial time the premises of the rules  --  other optimizations are possible in the DAG rewrite phase.    Regarding rule R1,  note that the only DAG changes that may trigger it, besides R1 itself,  can come from the rules R2, R6 and R8. Therefore, we can (i) start by applying  R1 to saturation, and (ii)  go  through the rest of the rules and reapply rule R1 to saturation only after one of these three rules triggered. Rule application  can be tested efficiently when the various paths involved  can be easily  identified, due to either being ``single incoming edge, single outgoing edge'' or being /-paths;   for identifying /-paths,   one must however apply to saturation rule R1, before any  tests involving such paths.   The test for R7 is slightly more involved, since part of its input, namely the $p_1$ candidates,  is not easily identifiable. However,  it is sufficient to test the existence of such candidates and to  handle them implicitly, contrary to $p_2$ candidates which can be found easily by the properties required by the rule.  Given a $p_2$ candidate, any mapping of $p_2$ nodes in the DAG rooted at node $n_1$ will also determine at least one such path $p_1$, and we only need to keep track of the allowed images for nodes of $p_2$ (between $n_1$ and the common descendant of $n_1$ and $p_2$'s nodes).

\eat{
Then, in order to implement rule R7, what we do is
(i) we mark as "allowed" for main-branch nodes mapping all main branch nodes in the subdag delimited by n1 and n2, with the exception of nodes in p2, n1 and n2.
(ii) we test whether a mapping exista from TP(p2) to SP(n1) \textbf{with the additional condition of mapping main branch nodes of p2 in allowed nodes}. The mapping computation is easily adaptable to this additional constraint.
If such mapping exists we can infer the existence of some p1 and safely remove p2, according to the rule.
}

Third, regarding the tests for collapsible nodes, the formal definition does not necessarily lead to the most efficient implementation. We can in fact bypass the computation of tentative DAGs and simply compare  the incoming and outgoing  /-paths for the tested nodes.  
%If these paths are the same up to a common ancestor or descendant or until the shortest of the two sizes if such "meeting point" doesn't exist, we can conclude to the satisfiablity of the resulting DAG, without explicitely applying the formal unsatisfiability condition. This proves to greatly improve the execution speed for rule R2.

Fourth, in order to efficiently check the applicability of rule R8, instead of the  na\"{i}ve test for each possible pair of nodes $(n_1, n_2)$, we can rely on two parallel mapping computations: one that is bottom-up, ending at the common ancestor node, and one that is top-down, ending that the common descendant node. Then, the intersection of the valid mappings for nodes of $p_2$ can reveal those that are relevant for R8, having only one possible image in $p_1$. 
%introduce an additional notion of mapping, "top-down", which complements the "bottom-up" approach. The intersection of the valid mappings bottom-up and top-down will the provide the set of valid mapping for a given node in p2. When this set contains exactly one element, rule R8 may be applied.

\subsection{Other optimization opportunities}
Some general adjustments that proved to be useful for the overall performance include dedicated data structures, such as adjacency lists for incoming and outgoing main branch edges, predicates, child and descendant edges -- as many of the rules involve iterating on specific children types -- as well as lists of topologically-sorted nodes,  built with the candidate plan and  updated only when needed, after certain rules were applied.

%Among future improvements, the pre-computation/update-only-when-needed is a general optimization direction that still holds a lot of potential (mapping matrixes, paths). Also, on a higher level, selectively applying rule R1 may be generalized to a graph of possible rule transitions (whether, in a given context, rule Ri may be applied after rule Rj) and to refinements on the order of rule applications, yielding better overall performance.

\section{Experiments}
\label{sec:experiments}
We performed our experiments on an Intel(R) Core(TM) i7-2760QM@2.40GHz machine, with  8G of RAM and the Ubuntu 11.10 operating system. 

We evaluated the performance and scalability of the \algoname algorithm,  focusing on two main aspects: 
\begin{itemize}
\item the \emph{rewrite time}, i.e.,  the time necessary to find an equivalent rewriting, when one exists, and 
\item the improvements on \emph{evaluation time}, i.e., the comparison between the evaluation time of  the input query over the data, on one hand,  and the  rewrite time cumulated with the evaluation time of the rewriting, over the view documents, on the other hand.  
\end{itemize}
In the space of analysis, we looked at how these two performance indicators vary with the size and the type of input queries -- w.r.t. the various XPaths fragment that were discussed in the paper,  the size of the view set that may give a rewriting, and the size of the input data.

\subsection{Documents, queries and views}
Our experimental setup was guided by: (a) our focus on measuring rewrite time as well as improvements on evaluation time, (b) the intention to stress-test our implementation for performance evaluation purposes.
We thus needed:
\begin{itemize}
\item queries and views spanning the \xp fragments analyzed in our theoretical study,
\item a set of documents the queries and views would apply to,
\item \emph{the ability to scale} query, view and document sizes, for performance assessment.
\end{itemize}   

Given our needs in terms of variation of query and view structure, number and size, we could not benefit from existing benchmarks or real-life settings publishing queries and views.  Therefore we designed \emph{our own synthetic query and views generator}, suiting our testing purposes. Starting from a given XML input document, this generator produces queries and views over that document (i.e. yielding a non-empty result), controlling their structure, number and size, as well as pair-wise containment.

While our synthetic queries and views generator can be plugged on any XML document, our need to scale with the document size  limited the usefulness of existing XML documents. We have therefore adopted in our experiments the 
extensively  cited XMark document generator ~\cite{DBLP:conf/vldb/SchmidtWKCMB02}.  This generator allows varying the size of its output, while ensuring similar structure and properties across the XML documents it produces. 

Three input documents were generated with the XMark generator, of sizes $41$KB, $91$MB and $18$GB. On each of the documents, we used our custom generator to produce input queries and view sets.

We generated, for each of the documents, input queries of main branch size $5$, $7$ and $9$ (the XMark documents have a maximal depth of $11$).  We considered input queries from three categories: \xppes, beyond \xppes\ but in \xppdesc, and in \xp but beyond  \xppdesc, i.e., without restrictions.  

%%% INVESTIGATE THE FOLLOWING TWO ARGUMENTS FOR WHY WE DID NOT CONSIDER LARGER DOCS: 1) SPEEDUP ONLY INCREASES WITH DOC SIZE (CONSISTENT WITH EXPERIENCE FFROM PREVIOUS WORK ON REWRITING USING V
%%% VIEWS) AND 2) THE SAXON PARSER CRASHES ON LARGER DOCS (VERIFY).

For  input queries,  each possible pairing of main-branch size and category for input queries  was  used to generate $10$ random input queries,  for a total of $90$ input queries. Importantly, these queries were generated from the data, in a way  that ensures that they all  have  a\emph{ non-empty result} on  the three  input documents. This was to avoid meaningless evaluation time measurements and to preclude the case when an alternative detection of unsatisfiability would shortcut the rewrite time.

For the generation of views, to each of the input queries we associated five randomly generated view sets of variate size, namely  consisting of $40$, $80$, $160$, $320$,  or $640$  views, for a total of $450$ view sets. We had the following guidelines in the generation of view sets:

\begin{enumerate}
\item  While  we wanted many views, we wanted to control the percentage of views that would be useful in  a rewriting; more precisely, all the view sets  consisted of  $10$\% useful views (for the rewriting), while the remaining $90$\% were useless (i.e., they did not  map in the input query)\footnote{We adopted this $10\%-90\%$ ratio as a reasonable one for most practical scenarios.}.

\item We wanted  views that were not equivalent to the input query nor a prefix thereof, and did not allow single-view rewritings.  This was to exercise precisely the non-trivial, multiple-view rewritings that our algorithm achieves. 

\item We wanted only view sets that gave an equivalent rewriting for the input query they were associated to. This was in order to be able to evaluate rewritings and 
to check their benefit over the evaluation of the input queries.  %since our algorithm is quick in the alternative case of non-existing rewritings the algorithm is quick in detecting this situation and we wanted to stress test the more productive rewriting cases.  

\item For the fragment \xppes, no restrictions were imposed on the views, consistent with our theoretical results (completeness does not depend on view restrictions). 
%\item we wanted views that were guaranteed to have a non-empty result. 

\end{enumerate}
Note that, although all views have a non-empty result by construction, the size of their result and their selectivity could  vary significantly and were not controlled in the generator. Other aspects that were not controlled by our query generator were (i) the overall size of input queries (only the size of the main branch was chosen), (ii) the overall size of the views, and (iii) the overall size of the plans built and tested in \algoname.

\subsection{\algoname vs. \efficient}
As a first experiment, through the random generation of sets of views, we took a first step towards understanding how often one may lose completeness in practice, if interleavings are not computed.  In other words, we wanted to quantify how often, at line (7) in \algoname, a positive containment test has to deal with a DAG pattern in the left-hand side.  (Recall that the only  difference introduced by the \efficient variant of the algorithm is that the containment test is done only if the DAG transformations yield a tree pattern.)  This is important for input queries from  \xppdesc\  and \xp, as the computation of interleavings  -- potentially exponentially many --  is expected to represent the main overhead in the search for a rewriting. 

To this end,  the random generation flow was the following: (a) a set of views would be generated, for a given input query and a given set size, and (b) the rewrite plans would be constructed and each would be tested for equivalence in \algoname, \emph{within a limit of 30 minutes of execution time}. 

This experiment gave us valuable insight : within the time limit of this random process, yielding $450$ sets of views, we obtained no view set that did provide an equivalent rewriting, \emph{but only by performing interleavings,  after \proc{Apply-Rules}} (line ($7$) in \algoname).  In the process, for all the generation tentatives,  a large majority of the allotted time was spent in computing interleavings, when  a tree pattern was  not outputted by \proc{Apply-Rules}. Moreover, in many of these cases, the time limit was met without reaching a conclusion. This outcome, on one hand,  confirms the fact that the computation of interleavings is very costly. More importantly,  while this represents just  a preliminary set of results, it does suggest that  one  could ``turn off'' the computation of interleavings: while we know that interleavings are needed for completeness, it seems that we may not have a sufficiently important gain in finding rewritings by computing them; if supported by further evidence,  we can suggest avoiding the computation of interleavings,  using \efficient instead of \algoname.

%It's a way to quantify how much completeness we loose. 
%and this rewriting could be found by \efficient. 
%We can interpret as significant  evidence that computing all interleavings would be too costly, if at all feasible, in most practical scenarios. 

Consequently,  with the $450$ sets of views we obtained, we continued our empirical evaluation using the \efficient variant of our algorithm.

\subsection{Rewrite time}
For this set of measurements, for each input query size and category, for the corresponding $10$ input queries, we recorded the average time to find a rewriting for each possible size of the view set (among the $5$ sizes given previously). This allowed us to understand how the rewrite time using \efficient varies  with the input query size and  category, on one hand, and with the size of the view set, on the other hand. 

\begin{figure}[t]
\hspace{-1.9cm}
\vspace{-4mm}
  \includegraphics[trim=0mm 125mm 0mm 0mm, clip=true,
scale=0.6]{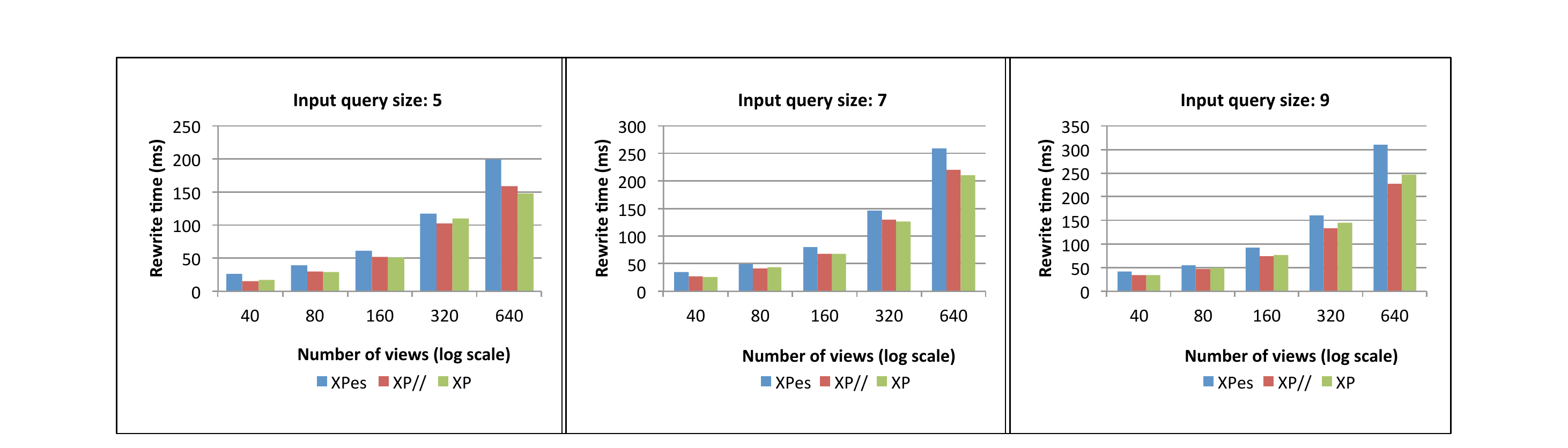}
\vspace{-2mm}
\caption{Rewrite time results.\label{fig:running-time}}
\vspace{-3mm}
\end{figure}

We present our  measurements for the running time of  \efficient in Figure~\ref{fig:running-time}. There,  we give one set of results (a sub-figure) for each query length. In each sub-figure, we  give five groups of three  columns. A group corresponds to one possible size of the view set, and in each group the first column corresponds to \xppes\ input queries, the second column corresponds to  \xppdesc\ input queries, and the third column corresponds to input queries without  restrictions. 

We can draw several important conclusions from the results of in Figure~\ref{fig:running-time}. First, our proof-of-con\-cept prototype of \efficient  can process efficiently, in a fraction of a second, queries of significant size -- up to $9$ nodes in the main branch, with $3-4$ predicates in average on each main branch node and with predicates of average depth of $3$ -- and view sets of significant size as well (order of hundreds).  Note that the measurements follow closely a linear progression with respect to the size of the view set. With respect to varying  the query size, the observed progression is even less pronounced  -- for example, for queries without restrictions, from 110ms to 210ms to 250ms. 

%It is thus notable that  the complexity upper-bound characterization that we gave for the rewriting problem (quadratic in the size of the query and plan) is  less visible in the  experiments. %of rewriting using multiple views.

\begin{figure}[t]
\hspace{-1cm}
\vspace{-4mm}
  \includegraphics[trim=0mm 30mm 0mm 0mm, clip=true,
scale=0.6]{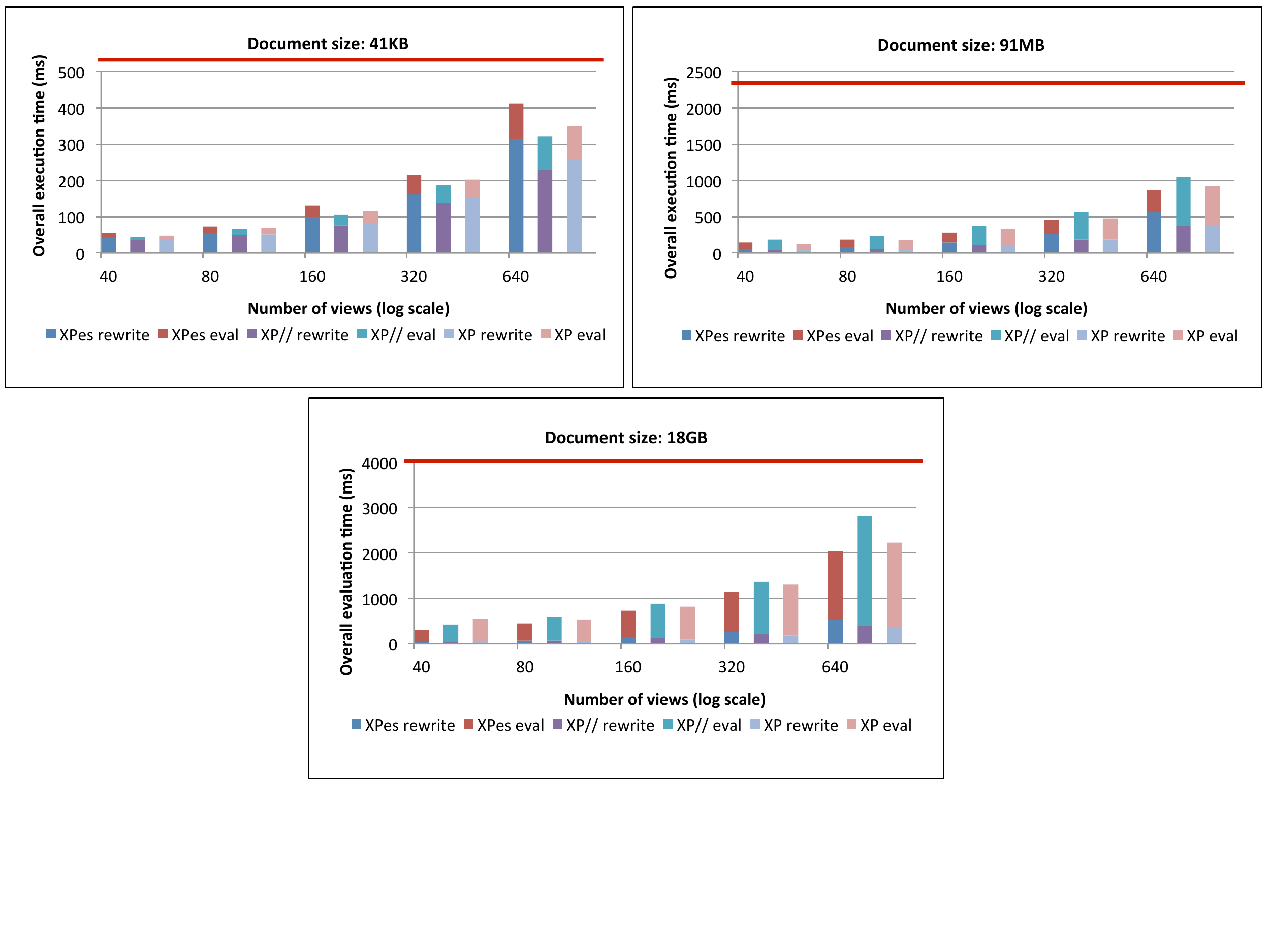}
\vspace{-1mm}
\caption{Global time results (rewrite time plus evaluation time).\label{fig:execution-time}}
\vspace{-4mm}
\end{figure}

\subsection{Evaluation time}
Regarding evaluation time,  we  compared the time necessary to evaluate the initial (input) queries over the input documents with the time necessary to build, test and then evaluate the rewriting over the view documents. We present these  measurements in Figure~\ref{fig:execution-time}.  For the sake of conciseness, we present only the results for the input queries of  maximal size ($9$ nodes in the main branch).  We give one set of results (a sub-figure) for each document size. As before, in each sub-figure, we  give five groups of three  columns, with  one group for each possible size of the view set.  %in each group, the first column corresponds to the \xppes\ input queries, the second column corresponds to the \xppdesc\ input queries, and the third column corresponds to the input queries with no restrictions.  
 As for rewrite time, we measured the average  time in each setting. Since the time necessary to run the input query over the input documents does not depend on the views, it is represented by a horizontal  line in the plot. 

Query evaluation was done using the SAXON query engine\footnote{http://saxon.sourceforge.net/},  which we extended with the Id-based JOIN functionality across multiple documents (the view documents), as SAXON's ability to perform this task was  incomplete.

A first important aspect to be noted in  Figure~\ref{fig:execution-time} is that,  over all input documents,  the time necessary to eva\-lu\-ate  the rewriting is smaller than the one for the input query, for all sizes of view sets. Moreover, the evaluation time based on view documents exhibits a linear progression and, overall, remains quite low, roughly $0.5$ seconds for the smallest document, around  $1$ second for the intermediary document, and around $2.5$ seconds for the largest document.

One can note the intuitive trend indicating that the larger the set of views in the rewriting, the less

\noindent important the performance benefit  over the original query plan. (Note that we measured the  plans  consisting of all the useful views.) In our results,  this trend stems from the way we set up the experiments, doubling at each step the number of views that were applicable in a rewriting (while this seems to be an unlikely scenario in practice, it represents a  suited stress test for our  algorithm). In our experimental configuration, many views meant, inevitably, more data and many opened documents, hence the overhead related to managing them, which for SAXON  starts being  noticeable.

% Moreover,  the evaluation time based on view documents follows a linear progression. In the extreme case of $1280$ views, one can note that the difference between the ratio of the two execution strategies decreases with the increase of the input document size. Also, in this case, the evaluation time based on views remains quite low, around $1$ second for the smallest document, around  $3$ seconds for the second document, and around $5$ seconds for the largest document. 

Within one group of columns, the differences in evaluation time  based on views between the three categories of queries are mainly due to the variations in terms of selectivity and view documents' size. For instance, on the smallest document, the views generated for extended skeletons were significantly less selective, yielding view documents  almost two times larger than the ones corresponding to the other two categories. Similar differences could be observed for the second document,  between views for the \xppes\ and \xppdesc\ queries on one hand, with larger view documents, and the views for  arbitrary queries on the other hand.  %In the case of the largest document, the differences between the measurements in the three categories are less important. 

Finally,  we also considered the execution time of \efficient over the view sets without the $10\%$ useful views, measuring the ratio between the direct evaluation of the query and the overhead of unsuccessful \efficient  runs. Without listing the precise measurements, we mention here  that we observed negligible overhead in all cases -- up to 2\% --   even over the smallest document.

\subsection{Discussion}
Our main conclusions from this experimental evaluation are the following:
\begin{itemize} 
\item \efficient scales to large sets of views, with rewrite time under one second, in all cases. Recall ours is an initial prototype and there is certainly room for further optimization, therefore these are very encouraging results. Moreover, the rewrite time represented a small percentage of the evaluation time.  At the same time,  there are  many  scenarios (e.g., with security views) where rewriting is not done for performance purposes, and in which the comparison between rewrite time and evaluation time is immaterial.

%note that the rewrite time is the only aspect that  matters in many scenarios, such as the security ones, where the rewriting could rely on a large number of views.  

\item The evaluation of the rewriting, including the rewrite time, is significantly more efficient than the evaluation of the input query, under assumptions that are widely-applicable in practice. Despite the fact that views were  generated randomly, without means to control their selectivity or  how they may ``cover'' the input query, the rewriting was evaluated two to three times faster than the input query,  even for hundreds of views.  Note that in the random generation process we do control the percentage of useful views and, for stress testing the algorithm,  this was increased  exponentially;  this is certainly not what one would expect in practice.   Finally,  the evaluation time depends undoubtedly on the particular query engine that is used, and it is not clear whether the one that we relied on had an optimal behavior when handling many opened, large documents.

\item   The benefit of our tractable techniques for equivalently rewriting XPath queries using multiple views is significant.  % as the alternative na\"{i}ve approach of computing interleavings becomes rapidly unfeasible in practice. 
 Indeed,  interleavings are key to achieving completeness (as our theoretical results show) yet, at the same time, our experiments show that the benefit of computing interleavings is limited; therefore, \efficient is a good candidate for practical, performance-oriented scenarios, even when completeness cannot be guaranteed (e.g., beyond \xppes).  
\end{itemize}

\section{conclusion}\label{sec:conclusion}
We considered  the problem of rewriting XPath queries using multiple views, characterizing the complexity of the intersection-aware rewriting problem.
More precisely, our work identified the tightest restrictions under which an XPath
query can be rewritten in PTIME using an intersection of views and
we propose an algorithm that works for any type of
identifiers.  A side effect of this research was to establish a similar tractability
frontier for the problem of deciding if an intersection of XPaths can
be equivalently rewritten as an XPath without intersection or
union. We  extended our formal study of the view-based rewriting problem for XPath to  more complex rewrite plans, with no limitations on the number of intersection and navigation steps inside view outputs they may employ. We also presented a proof-of-concept implementation of our  techniques and a thorough experimental evaluation, focusing on  scalability   and the comparison between the evaluation time of the initial query plan and the evaluation time of the rewriting,  using the view documents.

An important open question  remains to   provide  a more accurate characterization for the complexity of XPath rewriting with multiple levels of intersection and compensation (\xpint rewritings).  In particular,  an important question is whether view-based rewriting for \xppes\ input queries
allows tractable, sound and complete procedures when rewrite plans are from \xpint.

\bibliographystyle{acmsmall}
\bibliography{biblio-tcs}

\begin{thebibliography}{}

\bibitem[\protect\citeauthoryear{Afrati, Damigos, and Gergatsoulis}{Afrati
  et~al\mbox{.}}{2011}]{Afrati2011}
{\sc Afrati, F.}, {\sc Damigos, M.}, {\sc and} {\sc Gergatsoulis, M.} 2011.
\newblock Union rewritings for {XPath} fragments.
\newblock In {\em Proceedings of the 15th Symposium on International Database
  Engineering}. IDEAS '11. 43--51.

\bibitem[\protect\citeauthoryear{Amer-Yahia, Cho, Lakshmanan, and
  Srivastava}{Amer-Yahia
  et~al\mbox{.}}{2002}]{DBLP:journals/vldb/Amer-YahiaCLS02}
{\sc Amer-Yahia, S.}, {\sc Cho, S.}, {\sc Lakshmanan, L. V.~S.}, {\sc and} {\sc
  Srivastava, D.} 2002.
\newblock Tree pattern query minimization.
\newblock {\em VLDB J.\/}~{\em 11,\/}~4, 315--331.

\bibitem[\protect\citeauthoryear{Aravogliadis and Vassalos}{Aravogliadis and
  Vassalos}{2011}]{AravogliadisV11}
{\sc Aravogliadis, P.} {\sc and} {\sc Vassalos, V.} 2011.
\newblock On equivalence and rewriting of {XPath} queries using views under
  {DTD} constraints.
\newblock In {\em DEXA (2)}. 1--16.

\bibitem[\protect\citeauthoryear{Arion, Benzaken, Manolescu, and
  Papakonstantinou}{Arion et~al\mbox{.}}{2007}]{DBLP:conf/vldb/ArionBMP07}
{\sc Arion, A.}, {\sc Benzaken, V.}, {\sc Manolescu, I.}, {\sc and} {\sc
  Papakonstantinou, Y.} 2007.
\newblock Structured materialized views for {{XML}} queries.
\newblock In {\em VLDB}. 87--98.

\bibitem[\protect\citeauthoryear{Balmin, {\"O}zcan, Beyer, Cochrane, and
  Pirahesh}{Balmin et~al\mbox{.}}{2004}]{DBLP:conf/vldb/BalminOBCP04}
{\sc Balmin, A.}, {\sc {\"O}zcan, F.}, {\sc Beyer, K.~S.}, {\sc Cochrane, R.},
  {\sc and} {\sc Pirahesh, H.} 2004.
\newblock A framework for using materialized {{XPath}} views in {{XML}} query
  processing.
\newblock In {\em VLDB}. 60--71.

\bibitem[\protect\citeauthoryear{Benedikt, Fan, and Geerts}{Benedikt
  et~al\mbox{.}}{2005}]{DBLP:conf/pods/BenediktFG05}
{\sc Benedikt, M.}, {\sc Fan, W.}, {\sc and} {\sc Geerts, F.} 2005.
\newblock {{XPath}} satisfiability in the presence of {{DTD}s}.
\newblock In {\em PODS}. 25--36.

\bibitem[\protect\citeauthoryear{Benedikt, Fan, and Kuper}{Benedikt
  et~al\mbox{.}}{2005}]{BenediktTCS05}
{\sc Benedikt, M.}, {\sc Fan, W.}, {\sc and} {\sc Kuper, G.} 2005.
\newblock Structural properties of {{XPath}} fragments.
\newblock {\em Theor. Comput. Sci.\/}~{\em 336,\/}~1, 3--31.

\bibitem[\protect\citeauthoryear{Berglund, Boag, Chamberlin, Fern{\'a}ndez,
  Kay, Robie, and Sim{\'e}on}{Berglund et~al\mbox{.}}{2007}]{xpath2}
{\sc Berglund, A.}, {\sc Boag, S.}, {\sc Chamberlin, D.}, {\sc Fern{\'a}ndez,
  M.~F.}, {\sc Kay, M.}, {\sc Robie, J.}, {\sc and} {\sc Sim{\'e}on, J.} 2007.
\newblock {{XML}} path language ({{XPath}}) 2.0.

\bibitem[\protect\citeauthoryear{Boag, Chamberlain, Fern{\'a}ndez, Florescu,
  Robie, and Sim{\'e}on}{Boag et~al\mbox{.}}{2007}]{xquery}
{\sc Boag, S.}, {\sc Chamberlain, D.}, {\sc Fern{\'a}ndez, M.~F.}, {\sc
  Florescu, D.}, {\sc Robie, J.}, {\sc and} {\sc Sim{\'e}on, J.} 2007.
\newblock {{XQuery}} 1.0: An {{XML}} query language.

\bibitem[\protect\citeauthoryear{Bray, Paoli, Sperberg-McQueen, Maler, and
  Yergeau}{Bray et~al\mbox{.}}{2006}]{xml}
{\sc Bray, T.}, {\sc Paoli, J.}, {\sc Sperberg-McQueen, C.}, {\sc Maler, E.},
  {\sc and} {\sc Yergeau, F.} 2006.
\newblock Extensible markup language ({{XML}}) 1.0 (fourth edition).

\bibitem[\protect\citeauthoryear{Buneman, Davidson, Fan, Hara, and Tan}{Buneman
  et~al\mbox{.}}{2003}]{BunemanDFHT03}
{\sc Buneman, P.}, {\sc Davidson, S.~B.}, {\sc Fan, W.}, {\sc Hara, C.~S.},
  {\sc and} {\sc Tan, W.~C.} 2003.
\newblock Reasoning about keys for {XML}.
\newblock {\em Inf. Syst.\/}~{\em 28,\/}~8, 1037--1063.

\bibitem[\protect\citeauthoryear{Cautis, Abiteboul, and Milo}{Cautis
  et~al\mbox{.}}{2007}]{DBLP:conf/pods/CautisAM07}
{\sc Cautis, B.}, {\sc Abiteboul, S.}, {\sc and} {\sc Milo, T.} 2007.
\newblock Reasoning about {{XML}} update constraints.
\newblock In {\em PODS}. 195--204.

\bibitem[\protect\citeauthoryear{Cautis, Deutsch, and Onose}{Cautis
  et~al\mbox{.}}{2008}]{CautisWebDB08}
{\sc Cautis, B.}, {\sc Deutsch, A.}, {\sc and} {\sc Onose, N.} 2008.
\newblock {XPath} rewriting using multiple views: Achieving completeness and
  efficiency.
\newblock In {\em WebDB}.

\bibitem[\protect\citeauthoryear{Cautis, Deutsch, Onose, and Vassalos}{Cautis
  et~al\mbox{.}}{2011}]{DBLP:journals/tods/CautisDOV11}
{\sc Cautis, B.}, {\sc Deutsch, A.}, {\sc Onose, N.}, {\sc and} {\sc Vassalos,
  V.} 2011.
\newblock Querying {XML} data sources that export very large sets of views.
\newblock {\em ACM Trans. Database Syst.\/}~{\em 36,\/}~1, 5.

\bibitem[\protect\citeauthoryear{Chandra and Merlin}{Chandra and
  Merlin}{1977}]{ChandraM77}
{\sc Chandra, A.~K.} {\sc and} {\sc Merlin, P.~M.} 1977.
\newblock Optimal implementation of conjunctive queries in relational data
  bases.
\newblock In {\em STOC}. 77--90.

\bibitem[\protect\citeauthoryear{Chen and Rundensteiner}{Chen and
  Rundensteiner}{2002}]{DBLP:conf/webdb/ChenR02}
{\sc Chen, L.} {\sc and} {\sc Rundensteiner, E.~A.} 2002.
\newblock {XCache}: {{XQuery}}-based caching system.
\newblock In {\em WebDB}. 31--36.

\bibitem[\protect\citeauthoryear{Clark and DeRose}{Clark and
  DeRose}{1999}]{xpath1}
{\sc Clark, J.} {\sc and} {\sc DeRose, S.} 1999.
\newblock {{XML}} path language ({{XPath}}).

\bibitem[\protect\citeauthoryear{Courcelle}{Courcelle}{1990}]{Courcelle90}
{\sc Courcelle, B.} 1990.
\newblock Graph rewriting: An algebraic and logic approach.
\newblock In {\em Handbook of Theoretical Computer Science, Volume B: Formal
  Models and Sematics (B)}. 193--242.

\bibitem[\protect\citeauthoryear{Deutsch and Tannen}{Deutsch and
  Tannen}{2003}]{DBLP:conf/vldb/DeutschT03}
{\sc Deutsch, A.} {\sc and} {\sc Tannen, V.} 2003.
\newblock {MARS}: A system for publishing {{XML}} from mixed and redundant
  storage.
\newblock In {\em VLDB}. 201--212.

\bibitem[\protect\citeauthoryear{Fallside and Walmsley}{Fallside and
  Walmsley}{2004}]{xsd}
{\sc Fallside, D.~C.} {\sc and} {\sc Walmsley, P.} 2004.
\newblock {{XML}} {Schema} part 0: Primer second edition.

\bibitem[\protect\citeauthoryear{Fan, Geerts, Jia, and Kementsietsidis}{Fan
  et~al\mbox{.}}{2007}]{DBLP:conf/icde/FanGJK07}
{\sc Fan, W.}, {\sc Geerts, F.}, {\sc Jia, X.}, {\sc and} {\sc Kementsietsidis,
  A.} 2007.
\newblock Rewriting regular {XPath} queries on {XML} views.
\newblock In {\em ICDE}. 666--675.

\bibitem[\protect\citeauthoryear{Flum, Frick, and Grohe}{Flum
  et~al\mbox{.}}{2002}]{FlumFG02}
{\sc Flum, J.}, {\sc Frick, M.}, {\sc and} {\sc Grohe, M.} 2002.
\newblock Query evaluation via tree-decompositions.
\newblock {\em J. ACM\/}~{\em 49,\/}~6, 716--752.

\bibitem[\protect\citeauthoryear{Gao, Wang, and Yang}{Gao
  et~al\mbox{.}}{2007}]{DBLP:conf/dexa/GaoWY07}
{\sc Gao, J.}, {\sc Wang, T.}, {\sc and} {\sc Yang, D.} 2007.
\newblock {MQTree} based query rewriting over multiple {{XML}} views.
\newblock In {\em DEXA}. 562--571.

\bibitem[\protect\citeauthoryear{Garey and Johnson}{Garey and
  Johnson}{1979}]{DBLP:books/fm/GareyJ79}
{\sc Garey, M.~R.} {\sc and} {\sc Johnson, D.~S.} 1979.
\newblock {\em Computers and Intractability: A Guide to the Theory of
  NP-Completeness}.
\newblock W. H. Freeman.

\bibitem[\protect\citeauthoryear{Gottlob and Koch}{Gottlob and
  Koch}{2002}]{GottlobK02}
{\sc Gottlob, G.} {\sc and} {\sc Koch, C.} 2002.
\newblock Monadic queries over tree-structured data.
\newblock In {\em LICS}. 189--202.

\bibitem[\protect\citeauthoryear{Groppe, B{\"o}ttcher, and Groppe}{Groppe
  et~al\mbox{.}}{2006}]{DBLP:conf/icde/GroppeBG06}
{\sc Groppe, S.}, {\sc B{\"o}ttcher, S.}, {\sc and} {\sc Groppe, J.} 2006.
\newblock {{XPath}} query simplification with regard to the elimination of
  intersect and except operators.
\newblock In {\em ICDE Workshops}. 86.

\bibitem[\protect\citeauthoryear{Hartmann and Link}{Hartmann and
  Link}{2009}]{HartmannL09}
{\sc Hartmann, S.} {\sc and} {\sc Link, S.} 2009.
\newblock Efficient reasoning about a robust {XML} key fragment.
\newblock {\em ACM Trans. Database Syst.\/}~{\em 34,\/}~2.

\bibitem[\protect\citeauthoryear{Hidders}{Hidders}{2003}]{DBLP:conf/dbpl/Hidde%
rs03}
{\sc Hidders, J.} 2003.
\newblock Satisfiability of {{XPath}} expressions.
\newblock In {\em DBPL}. 21--36.

\bibitem[\protect\citeauthoryear{Katsifodimos, Manolescu, and
  Vassalos}{Katsifodimos et~al\mbox{.}}{2012}]{Katsifodimos2012}
{\sc Katsifodimos, A.}, {\sc Manolescu, I.}, {\sc and} {\sc Vassalos, V.} 2012.
\newblock Materialized view selection for {XQuery} workloads.
\newblock In {\em Proceedings of the 2012 ACM SIGMOD International Conference
  on Management of Data}. SIGMOD '12. 565--576.

\bibitem[\protect\citeauthoryear{Lakshmanan, Wang, and Zhao}{Lakshmanan
  et~al\mbox{.}}{2006}]{DBLP:conf/vldb/LakshmananWZ06}
{\sc Lakshmanan, L. V.~S.}, {\sc Wang, H.}, {\sc and} {\sc Zhao, Z.} 2006.
\newblock Answering tree pattern queries using views.
\newblock In {\em VLDB}. 571--582.

\bibitem[\protect\citeauthoryear{Mandhani and Suciu}{Mandhani and
  Suciu}{2005}]{mandhani-suciu}
{\sc Mandhani, B.} {\sc and} {\sc Suciu, D.} 2005.
\newblock Query caching and view selection for {{XML}} databases.
\newblock In {\em VLDB}. 469--480.

\bibitem[\protect\citeauthoryear{Manolescu, Karanasos, Vassalos, and
  Zoupanos}{Manolescu et~al\mbox{.}}{2011}]{Manolescu11}
{\sc Manolescu, I.}, {\sc Karanasos, K.}, {\sc Vassalos, V.}, {\sc and} {\sc
  Zoupanos, S.} 2011.
\newblock Efficient {XQuery} rewriting using multiple views.
\newblock In {\em ICDE}. 972--983.

\bibitem[\protect\citeauthoryear{Miklau and Suciu}{Miklau and
  Suciu}{2004}]{DBLP:journals/jacm/MiklauS04}
{\sc Miklau, G.} {\sc and} {\sc Suciu, D.} 2004.
\newblock Containment and equivalence for a fragment of {{XPath}}.
\newblock {\em J. ACM\/}~{\em 51,\/}~1, 2--45.

\bibitem[\protect\citeauthoryear{Neven and Schwentick}{Neven and
  Schwentick}{2006}]{DBLP:journals/lmcs/NevenS06}
{\sc Neven, F.} {\sc and} {\sc Schwentick, T.} 2006.
\newblock On the complexity of {{XPath}} containment in the presence of
  disjunction, {{DTD}s}, and variables.
\newblock {\em Logical Methods in Computer Science\/}~{\em 2,\/}~3.

\bibitem[\protect\citeauthoryear{Onose, Deutsch, Papakonstantinou, and
  Curtmola}{Onose et~al\mbox{.}}{2006}]{DBLP:conf/sigmod/OnoseDPC06}
{\sc Onose, N.}, {\sc Deutsch, A.}, {\sc Papakonstantinou, Y.}, {\sc and} {\sc
  Curtmola, E.} 2006.
\newblock Rewriting nested {{XML}} queries using nested views.
\newblock In {\em SIGMOD}. 443--454.

\bibitem[\protect\citeauthoryear{Popa, Deutsch, Sahuguet, and Tannen}{Popa
  et~al\mbox{.}}{2000}]{DBLP:conf/sigmod/PopaDST00}
{\sc Popa, L.}, {\sc Deutsch, A.}, {\sc Sahuguet, A.}, {\sc and} {\sc Tannen,
  V.} 2000.
\newblock A chase too far?
\newblock In {\em SIGMOD Conference}. 273--284.

\bibitem[\protect\citeauthoryear{Robie, Chamberlin, Dyck, and Snelson}{Robie
  et~al\mbox{.}}{2010}]{xpath3}
{\sc Robie, J.}, {\sc Chamberlin, D.}, {\sc Dyck, M.}, {\sc and} {\sc Snelson,
  J.} 2010.
\newblock {{XML}} path language ({{XPath}}) 3.0.

\bibitem[\protect\citeauthoryear{Schmidt, Waas, Kersten, Carey, Manolescu, and
  Busse}{Schmidt et~al\mbox{.}}{2002}]{DBLP:conf/vldb/SchmidtWKCMB02}
{\sc Schmidt, A.}, {\sc Waas, F.}, {\sc Kersten, M.~L.}, {\sc Carey, M.~J.},
  {\sc Manolescu, I.}, {\sc and} {\sc Busse, R.} 2002.
\newblock {XMark}: A benchmark for {XML} data management.
\newblock In {\em VLDB}. 974--985.

\bibitem[\protect\citeauthoryear{Tang and Zhou}{Tang and
  Zhou}{2005}]{DBLP:conf/xsym/TangZ05}
{\sc Tang, J.} {\sc and} {\sc Zhou, S.} 2005.
\newblock A theoretic framework for answering {{XPath}} queries using views.
\newblock In {\em XSym}. 18--33.

\bibitem[\protect\citeauthoryear{Tang, Yu, \"Ozsu, Choi, and Wong}{Tang
  et~al\mbox{.}}{2008}]{waterlooicde08}
{\sc Tang, N.}, {\sc Yu, J.~X.}, {\sc \"Ozsu, M.~T.}, {\sc Choi, B.}, {\sc and}
  {\sc Wong, K.-F.} 2008.
\newblock Multiple materialized view selection for {{XPath}} query rewriting.
\newblock In {\em ICDE}.

\bibitem[\protect\citeauthoryear{ten Cate and Lutz}{ten Cate and
  Lutz}{2007}]{1265541}
{\sc ten Cate, B.} {\sc and} {\sc Lutz, C.} 2007.
\newblock The complexity of query containment in expressive fragments of
  {{XPath}} 2.0.
\newblock In {\em PODS}. 73--82.

\bibitem[\protect\citeauthoryear{Wang, Li, and Yu}{Wang
  et~al\mbox{.}}{2011}]{Wang2011}
{\sc Wang, J.}, {\sc Li, J.}, {\sc and} {\sc Yu, J.~X.} 2011.
\newblock Answering tree pattern queries using views: a revisit.
\newblock In {\em Proceedings of the 14th International Conference on Extending
  Database Technology}. EDBT/ICDT '11. 153--164.

\bibitem[\protect\citeauthoryear{Wu, Theodoratos, and Wang}{Wu
  et~al\mbox{.}}{2009}]{Wu2009}
{\sc Wu, X.}, {\sc Theodoratos, D.}, {\sc and} {\sc Wang, W.~H.} 2009.
\newblock Answering {XML} queries using materialized views revisited.
\newblock In {\em Proceedings of the 18th ACM conference on Information and
  knowledge management}. CIKM. 475--484.

\bibitem[\protect\citeauthoryear{Wu, Gucht, Gyssens, and Paredaens}{Wu
  et~al\mbox{.}}{2009}]{DBLP:conf/bncod/WuGGP09}
{\sc Wu, Y.}, {\sc Gucht, D.~V.}, {\sc Gyssens, M.}, {\sc and} {\sc Paredaens,
  J.} 2009.
\newblock A study of a positive fragment of path queries: Expressiveness,
  normal form, and minimization.
\newblock In {\em BNCOD}. 133--145.

\bibitem[\protect\citeauthoryear{Xu and {\"O}zsoyoglu}{Xu and
  {\"O}zsoyoglu}{2005}]{DBLP:conf/vldb/XuO05}
{\sc Xu, W.} {\sc and} {\sc {\"O}zsoyoglu, Z.~M.} 2005.
\newblock Rewriting {{XPath}} queries using materialized views.
\newblock In {\em VLDB}. 121--132.

\bibitem[\protect\citeauthoryear{Yang, Lee, and Hsu}{Yang
  et~al\mbox{.}}{2003}]{CachingVLDB03}
{\sc Yang, L.~H.}, {\sc Lee, M.~L.}, {\sc and} {\sc Hsu, W.} 2003.
\newblock Efficient mining of {{XML}} query patterns for caching.
\newblock In {\em VLDB}. 69--80.

\end{thebibliography}
\appendix

\section*{APPENDIX} \setcounter{section}{0}

\section{Proof of Theorem~\ref{th:completenessUF-es} (\xp fragment for PTIME)}
\label{sec:proof1}
The proof is organized as follows. We first show that $\proc{Apply-Rules}$ is complete  over  DAG patterns in which the root and the output node are connected by a path having only /-edges (Lemma~\ref{l:onetokenq_nviews}).  

We then consider the complementary case when all the branches in parallel (the compensated views) have at least one //-edge in the main branch. For clarity, we prove completeness progressively, starting with the case of intersecting two views under certain restrictions: their root tokens have the same main branch, their result tokens have the same main branch as well (Lemma~\ref{lem:2-comp-skel}). We then extend to the case of arbitrary many views, with these restrictions (Lemma~\ref{lem:n-comp-skel}). Then we consider the general case, which will rely on results proven for the limited cases. 

We start with DAG patterns -- unfoldings of the rewrite plan -- in which the root and the output node are connected by a path having only /-edges, i.e. there one of the views involved in the intersection has only one token.
\begin{lemma}\label{l:onetokenq_nviews}
%  Given 2 extended skeletons $v_1$, $v_2$, where $v_1$ has only one
%  token, then $\dagp{v_1 \cap v_2}$ is union-free iff \apprules
%  rewrites it into a tree.
For $n$ \xppes\ patterns $v_1, v_2, \dots, v_n$, where $v_1$ has
  only one token, $\dagp{v_1 \cap v_2 \cap \dots \cap v_n}$ is
  union-free iff \apprules rewrites it into a tree.
\end{lemma}
\begin{proof}
Without loss of generality,  we can consider that  $v_2, v_3, \dots v_n$ have more than one token.

%The case in which $v_2$ has only one token is trivial: we can just
%apply rule R1 for all edges. Let us assume that $v_2$ has at least one
%//-edge in the main branch i.e. it has more than one token.

Suppose that the output of the rewriting algorithm, call it $d'$, is
not a tree pattern. It is easy to see that there is a subpattern \textit{sd} in
$d'$ that is not a tree and it has a //-edge, otherwise R1 (plus maybe
R7) would have reduced that subgraph to a tree. Then there must be a
node $n_1$ in $d'$ such that there are at least $2$ main branch paths ($p_1$ and $p_2$) going out
of $n_1$: one starting with a /-edge $n_1/n_2$ and another one
starting with a //-edge $n_1//n_3$, such that the 2 paths meet again
starting from a node $n_4$. We can also infer that on the $n_1//n_3$
branch, the last edge before $n_4$ is a //-edge $n_5//n_4$, as
in Figure~\ref{fig:subgraph_onetokencase}, otherwise
R1.ii would have applied or the pattern would have been unsatisfiable.

\eat{

\begin{figure}[h]
\begin{center}
\includegraphics[trim=0mm 130mm 190mm 0mm, clip=true,
scale=0.60]{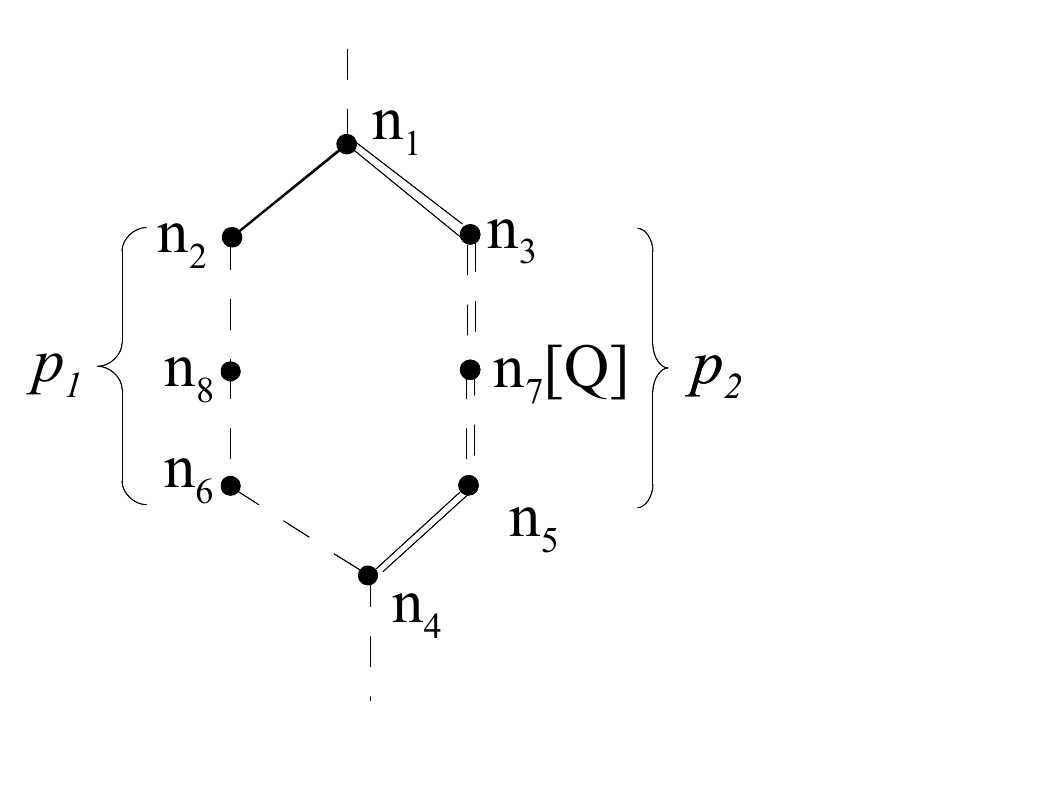}
\caption{Subgraph \textit{sd} (when one query has 1 token). \label{fig:subgraph_onetokencase}}
\end{center}
\end{figure}
}

Since R2 did not apply, it means that $n_2$ and $n_3$ are collapsible and, moreover, that the root token of $p_2$ is a prefix of $p_1$ yet  there is
at least one predicate attached to a node of it that  that does not map in the corresponding position in $p_1$.

Besides these two branches, other branches may be present in parallel with the /-only main branch in $d'$ of which $p_1$ is part of. %, not necessary starting at $n_1$.

In general, since R7 does not apply, it means that $p_2$ does not fully map into $p_1$. This means that an interleaving $i$  of $d'$ s.t. $d'\equiv i$ (witness for union-freedom), built by some choice $\psi$ of mapping $p_2$'s main branch nodes into $p_1$'s main branch nodes, must for at least some node $n_7$ in $p_2$ and  predicate $Q$ attached to it collapse $n_7$ with $n_8=\psi(n_7)$ of $p_1$, even though  \pattern{\lambda_d(n_7)[Q]} did not map into \sub{d'}{n_8}.

We will show that there exists at least one other interleaving $i'$ of $d'$ in which the predicate $Q$ does not hold at $n_8$. Since R7  did not apply, it means that $Q$ did not hold at $n_8$ before any possible application of R9 and it still does not hold after R9 steps. But this basically tells us there there exists an interleaving in which $Q$ does not hold at $n_8$. This implies that $d$ is not union-free.
\end{proof}

We now consider the case when all the branches in parallel (the compensated views) have at least one //-edge in the main branch.  

We start by proving the following  lemma.
\begin{lemma}
\label{lem:2-comp-skel}
For two \xppes\ patterns $v_1, v_2$ s.t. their root and result tokens have the same main branch, the DAG pattern $d=dag(v_1
\cap v_2)$ is union-free iff \apprules rewrites $d$ into a
tree.
\end{lemma}
\begin{proof}
Let us first consider what rule steps may apply in order to refine
$d$. First, since we are dealing with patterns with root and result tokens having the same main branch, R1 steps
will first apply, coalescing the root and result tokens of the two
branches. At this point, the only rules that remain applicable are R6
and R7.  This is because we do not have nodes with incoming (or
outgoing) /-edges and //-edges simultaneously and R5 will only apply
to predicates starting by a //-edge.

We argue now that $d$ is union-free iff rule R7 applies on it.

Note that since we only have 2 parallel branches an application of
rule R7 would immediately yield a tree. So the \emph{if} direction is
straightforward. For the \emph{only if} direction, if R7 does not
apply this translates into \emph{$(\dagger)$ there is no mapping (not necessarily root-mapping)
between the intermediary part of $v_1$ and the intermediary part of
$v_2$.}

Assuming that $(\dagger)$ holds, rule R6 remains the only option. So,
possibly after some applications of R6, followed by applications of R1
collapsing entire tokens, we obtain a refined DAG $d$ as illustrated
in Figure~\ref{fig:completeness1} (only the main branches of $d$ are
illustrated). $p_r$ has the common main branch following the root (may
have several tokens if R6 applied) and $t_o$ denotes the result
token. $t_1$ and $t_2$ denote the two sibling /-patterns for which R6
no longer applies. As $t_1$, $t_2$ are \emph{dissimilar} we have
that $t_1 \not \equiv t_2$.

\begin{figure}[t]
  \centering
  \subfloat[DAG pattern $d$ for Lemma~\ref{lem:2-comp-skel}.]{\label{fig:completeness1}\includegraphics[trim=50mm 65mm 90mm 25mm, clip=true,
scale=0.40]{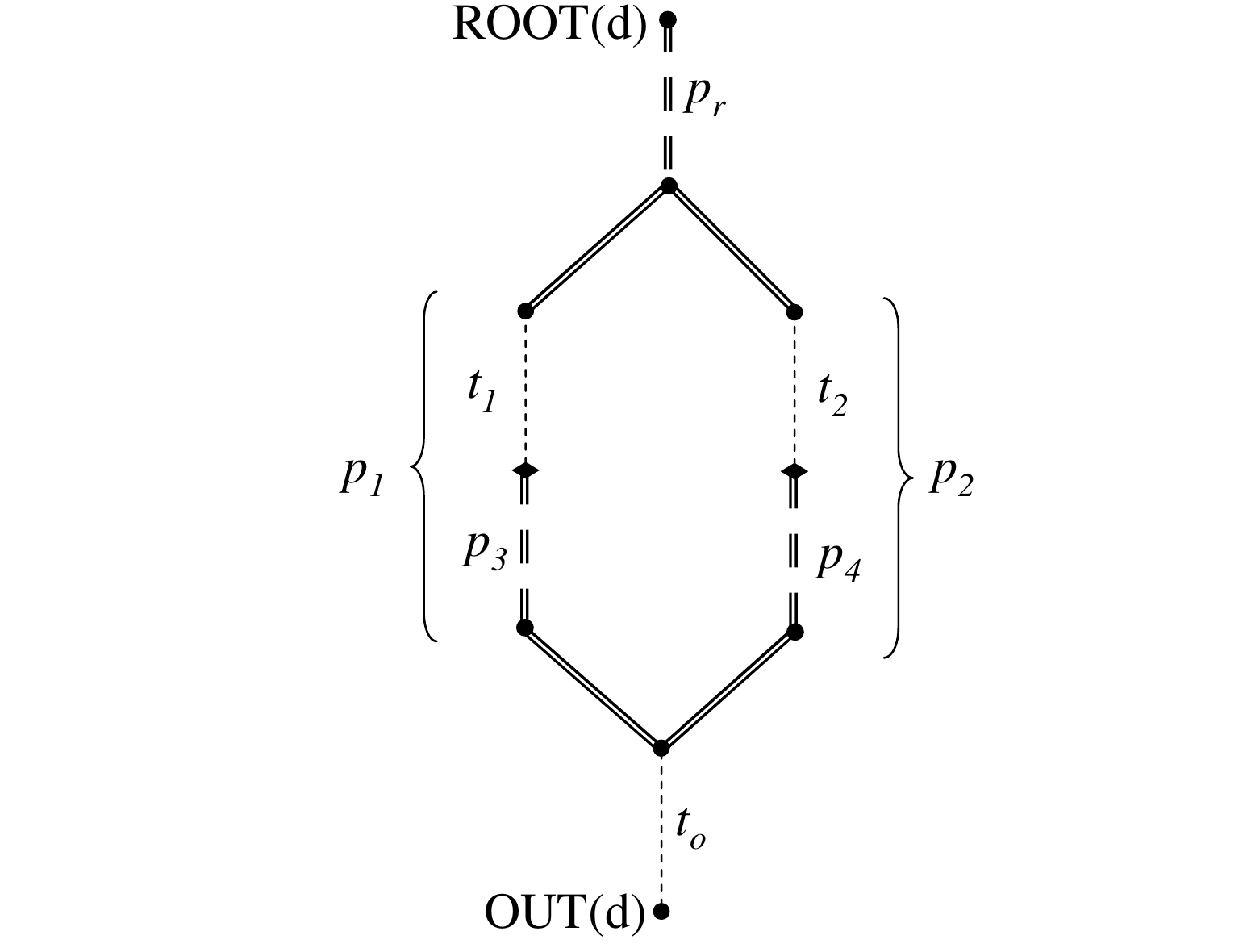}}                
 \subfloat[DAG pattern $d$ for Lemma~\ref{lem:n-comp-skel}.]{\label{fig:completeness3}\includegraphics[trim=50mm 85mm 90mm 25mm, clip=true,
scale=0.40]{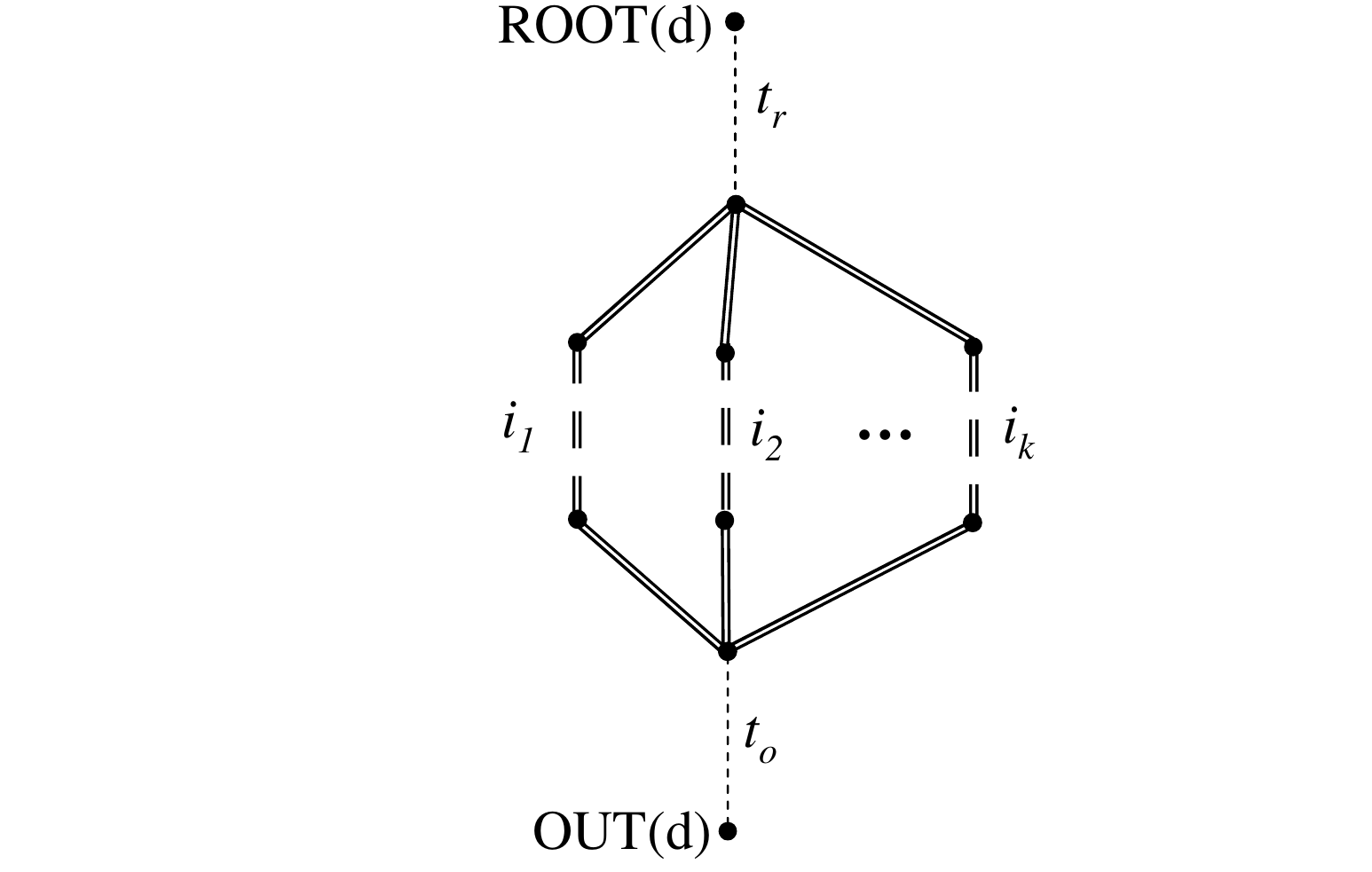}} 
 \caption{DAG patterns for the proof of  Theorem~\ref{th:completenessUF-es}.}
   \label{fig:contour1}
   \vspace{-3mm}
 \end{figure}

We show that $d$ is not union-free by the following approach: we build
two interleavings, $p'$ and $p''$, that do not contain one another,
and then show that by assuming the existence of a third interleaving $p$
that contains both we obtain the contradiction $t_1 \equiv t_2$.

We continue with the following observation which follows easily from the  \xppes\ restriction on  usage of //-edges  in predicates: given two /-patterns $t_1$ and $t_2$ from \xppes, if $t_1$ does not map in $t_2$  then, for any tree pattern $q$ of the form $\dots//t_2'//\dots$ with $t_2'$ being an isomorphic copy of $t_2$, we have that $t_1$ does not map into $t_2'$.

The following steps will implicitly use this  observation.

Let $p_1=t_1//p_3$ denote the left branch and let $p_2=t_2//p_4$
denote the right branch in $d$. Because of $(\dagger)$, there is no
mapping (not necessarily root-mapping) between $p_1$ and $p_2$.

Let $sf_1$ be the maximal token-suffix of $p_1$ that can map into
$p_2$, and let $pr_1$ be the remaining part (i.e., a
token-prefix). Note that $pr_1$ cannot be empty,  so we can write it 
as $p_1=pr_1//sf_1$.

Similarly, let $sf_2$ denote the maximal token-suffix of $p_2$ that
can map into $p_1$, and let $pr_2$ denote the remaining part,
non-empty as well. So we can write $p_2$ as $p_2=pr_2//sf_2$.

We build $p'$ and $p''$ as follows:
\begin{eqnarray*}
 p' & = & p_r'//pr_2'//p_1'//t_o'=p_r'//pr_2'//pr_1'//sf_1'//t_o'\\
  p''&=& p_r''//pr_1''//p_2''//t_o''=p_r''//pr_1''//pr_2''//sf_2''//t_o''
\end{eqnarray*}
where the $\#'$, $\#''$ parts are isomorphic copies of the $\#$ parts of $d$.

Note that $pr_1'$ (resp. $pr_2''$) starts by token $t_1' \equiv t_1$  (resp. $t_2'' \equiv t_2$).

These two queries are obviously in $interleave(d)$. Moreover, there
can be no containment mapping between $p'$ and $p''$ since, by the way
$sf_1$ and $sf_2$ were defined, $pr_1'$ (resp. $pr_2''$) could only
map in $pr_1''$ (resp. $pr_2'$).

So neither $p'$ nor $p''$ can be the interleaving that reduces all the
others.  We show in the following that no other interleaving $p$ of
$d$ can reduce both $p'$ and $p''$ unless $t_1 \equiv t_2$.

Let us assume that such a $p$ exists. Without loss of generality, let
$p$ be of the form $$p=p_r//m//t_o.$$ (interleavings that are not of
this kind will not remain in the normal form of $d$).

We assume a containment mapping $\phi'$ from $p$ to $p'$ and another one
  $\phi''$ from $p$ to $p''$. Obviously, $v_1,v_2$
must have containment mappings into $p$, since $p \equiv v_1 \cap
v_2$. In particular, their sub-sequences $p_1$ and $p_2$ have images
in the $m$ part of $p$. Let $\psi'$ and $\psi''$ be these containment
mappings.

With a slight abuse of notation, let $\psi'(pr_1)$ denote the minimal
token-prefix of $m$ within which the image under $\psi'$ of the $pr_1$
part of $v_1$ occurs. $\psi'(pr_1)$ is well defined because $v_1$ and
$v_2$, and hence $p_1$ and $p_2$, are in \xppes, hence the image of a
token of $p_1$ and of its predicates is included into a token of
$m$. In other words, $\psi'(pr_1)$ starts with the root token of $m$
and ends with the token into which the output token of $pr_1$
maps. Similarly, let $\psi''(pr_2)$ denote the minimal token-prefix of
$m$ within which the image under $\psi''$ of the $pr_2$ part of $v_2$
occurs.

We can thus write $p$ in two forms, as  
\begin{eqnarray*}
p& =& p_r//\psi'(pr_1) \dots \psi'(sf_1)//t_o \\
p & = &p_r//\psi''(pr_2)\dots\psi'(sf_2)//t_o
\end{eqnarray*}
Next, we argue that in the containment mapping $\phi''$ of $p$ in
$p''$, we must have $\phi''(\psi'(pr_1))=pr_1''$. Similarly, we must
have that $\phi'(\psi''(pr_2))=pr_2'$. This follows easily from the
way $sf_1$ and $sf_2$ were defined. (For instance, no node of
$\psi'(pr_1)$ can map below $pr_2''$ in $p''$, otherwise
$\textit{sf}_1$ would not be maximal. And $\mb{p_r}=\mb{p_r''}$ and
$|\psi'(pr_1) \geq pr_1''|$, hence no node
of $\psi'(pr_1)$ can map higher than $pr_1''$ either,
otherwise $p_r$ would not map into $p''$.) And it then implies that
$\psi'(pr_1) \equiv pr_1$ and $\psi''(pr_2) \equiv pr_2$.

But since $m$ starts by both the token-prefix $\psi'(pr_1)$ and by
$\psi''(pr_2)$, hence by token-prefixes $pr_1$ and $pr_2$,   $pr_1$ and $pr_2$ should at least start by the same token. \emph{Hence
$t_1 \equiv t_2$, which is a contradiction.
}

In other words, we showed  that $d$ is union-free iff, after a  sequence
of R1 steps,  R7 applies, transforming the pattern into a
tree. As we know from Lemma~\ref{l:termination} that \apprules also
terminates, it follows that $d$ is union-free iff \apprules rewrites
$d$ into a tree.
\end{proof}

For $d=v_1 \cap v_2$, where $v_1$ and $v_2$ are two
skeleton patterns such that their root and result tokens have the same main branch, with the previous notations, we can also easily prove the
following.

\begin{lemma}
\label{cl:shape}
All the interleavings of \nf{d} are of the form $p_r//\dots//t_o$.
\end{lemma}
We know so far that the intersection $d$ of two
skeleton queries such that their root and result tokens have the same main branch is union-free iff \apprules rewrites $d$ into a
tree. Moreover, this happens iff there is a mapping between the
intermediary part of one into the intermediary part of the
other. If $d$ is not union-free, the result is a union of
queries having the same root and result tokens, as described in Lemma~\ref{cl:shape}.

We now consider intersections of more than two patterns.

\begin{lemma}
\label{lem:n-comp-skel}
Given  \xppes\ patterns $v_1, \dots, v_n$ s.t. their root and result tokens have the same main branch, the DAG pattern
$d=dag(v_1 \cap \dots\cap v_n)$ is union-free iff~\apprules
rewrites $d$ into a tree.
If the skeletons are of the form $v_i = p_r//p_i//t_o$, $1\leq i \leq n$,
then $d$ is union-free iff
there is a query among
them, $v_j$, having an intermediary part $p_j$ such that all other
$p_i$ map into $p_j$.
\end{lemma}
\begin{proof}
We prove this by induction on the number of
patterns (Lemma~\ref{lem:2-comp-skel} covers $n=2$).

As in the case of Lemma~\ref{lem:2-comp-skel}, we first rewrite $d$ by
rule R1, coalescing the root and result tokens of the parallel
branches. At this point, the only rules that remain applicable are R6
and R7.

Let us now assume that some run of \apprules terminates on $d$ without
outputting a tree. Then, it is easy to check that \apprules will also
stop in the particular run, in which we start by applying only R7
until it does not apply anymore.

We show in the following that $d$ resulting from this run is not
union-free.

We continue with $d$ obtained, as said previously, possibly after some
applications of R7 that removed some of the branches in parallel,
yielding a DAG pattern as the one illustrated in
Figure~\ref{fig:completeness3}. $2\leq k\leq n$ denotes the number of
remaining branches in parallel and $i_1,
\dots i_k$ denote these branches. Without loss of generality, let these be the intermediary
parts of $v_1, \dots, v_k$ respectively.

Note that we are now in a setting in which $d \equiv v_1 \cap \dots \cap v_k$ and the following holds: \emph{$(\dagger)$ there is no mapping  between the intermediary parts of any  of $v_1, \dots, v_k$.}

Next, starting from the DAG pattern $d$ in Figure~\ref{fig:completeness3}, by $(\dagger)$, only rule R6 is applicable. %Since the Algorithm  %without But before applying R6, note that since $(\dagger)$ holds any sequence of applications of rule R6 over $d$ would terminate without outputting %a  tree. This is because none of the parallel branches can be entirely ``consumed'' during this phase. So we know the Algorithm does not output a %tree, and if we prove that $d$ is not union-free than its completeness.

For convenience, we  assume that R6 steps are applied by a slightly different strategy: we take an R6 step only if it applies to \emph{all} the parallel branches simultaneously. At some point, this process will stop as well and we obtain a refined $d$ as illustrated in Figure~\ref{fig:completeness2} (only the main branches are given).  Let  $p_i=t_i//r_i$  denote the branches in parallel. Note that the $t_i$ tokens cannot all be equivalent (recall that in \xppes similar patterns must be equivalent).

\begin{figure}[t]
  \centering
  \subfloat[DAG pattern $d$ for Lemma~\ref{lem:n-comp-skel}.]{\label{fig:completeness2}\includegraphics[trim=50mm 65mm 70mm 25mm, clip=true,
scale=0.40]{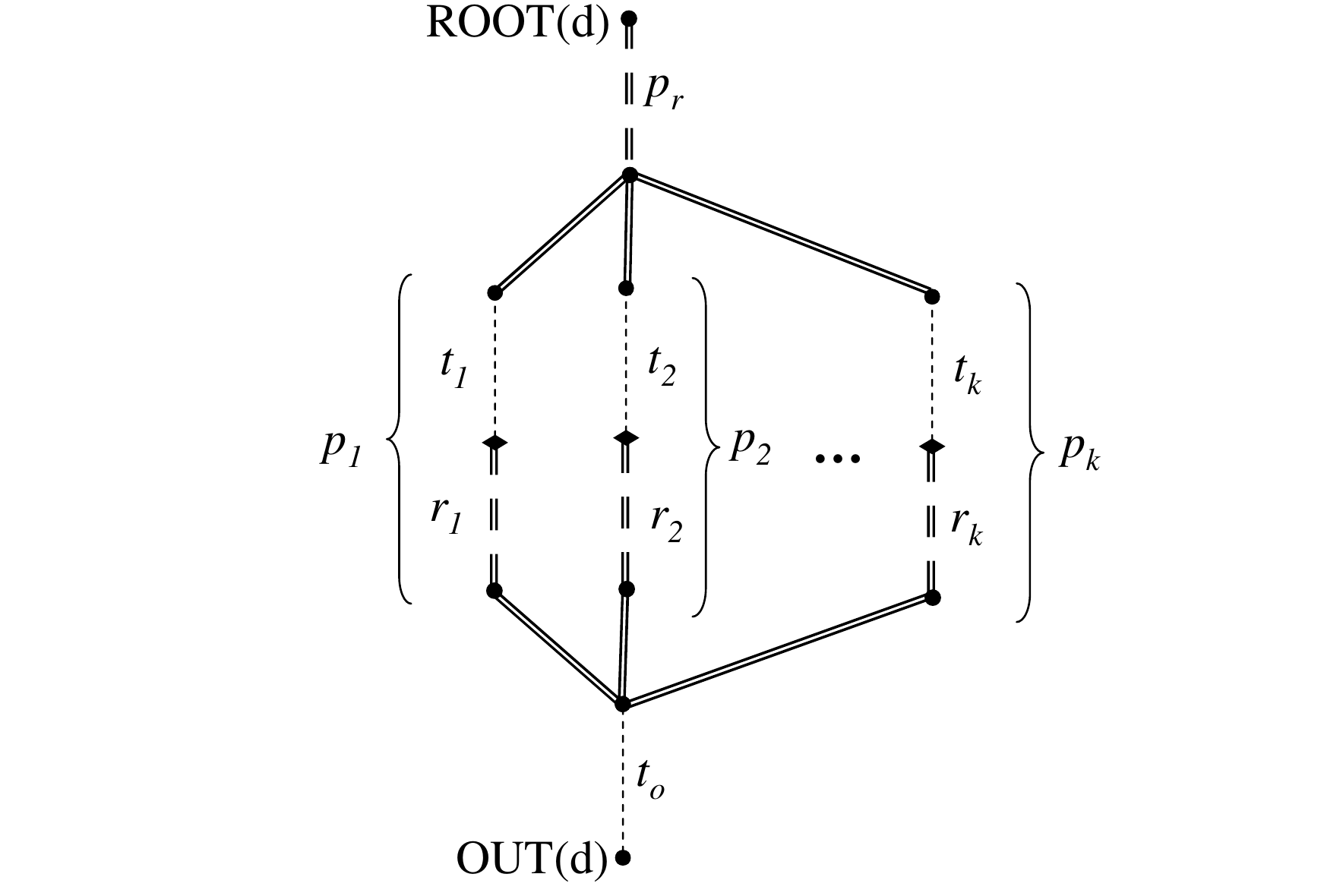}}                
 \subfloat[DAG pattern $d$ for Lemma~\ref{l:onetokenq_nviews}.]{\label{fig:subgraph_onetokencase}\includegraphics[trim=0mm 130mm 210mm 0mm, clip=true,
scale=0.50]{images/sd_onetokencase-bis}} 
 \caption{DAG patterns for the proof of  Theorem~\ref{th:completenessUF-es}.}
   \label{fig:contour2}
     \vspace{-3mm}
 \end{figure}

Let us assume towards a contradiction that \emph{$d$ is union-free}
and let $q$ be the interleaving such that $q\equiv d$. Without loss of generality, let $q$
be of the form $q=p_r//t//m//t_o$,  where $t$ is the token
immediately following $p_r$ (the $m$ part might be empty). 

Without loss of generality, 
let us also  assume that $t_1 \not \equiv t$ (we know that there must be at
least one such token among $t_1, \dots, t_k$.) We  show that by
assuming $q\equiv d$ we obtain the contradiction  $t \equiv
t_1$.

For $p_1$ chosen in this way, let $d'$ denote the DAG pattern obtained
from $d$ by removing its $p_1$ branch. Introducing for each $i$ the
pattern $v_i'=p_r//p_i//t_o,$ by $(\dagger)$ all incomparable, note
that $d'$ can be seen as $d'=\dagp{v_2' \cap \dots \cap v_k'}$ and
note also that $d \equiv d' \cap v_1' = d' \cap (p_r//t_1//r_1//t_o)$. 
By the inductive hypothesis, $d'$ is not union-free, i.e., 
there is some $x \geq 2$ and some patterns $q_1, \dots, q_x$, which are some incomparable
interleavings of $d'$ (such their root and result tokens have the same main branch), all of the form $p_r//\dots//t_o$ (by
induction, from Lemma~\ref{cl:shape}), such that
$d' \equiv q_1 \cup \dots \cup q_x$. 

So can conclude that $d \equiv v_1' \cap (q_1 \cup q_2 \cup \dots \cup q_m) = (v_1' \cap q_1) \cup (v_1' \cap q_2) \cup \dots \cup (v_1' \cap q_m)$.

Note now that we cannot have $v_1' \sqsubseteq q_i$, for any $q_i$,
since this would mean that $v_1' \sqsubseteq v_2', \dots, v_k'$, in
contradiction with ($\dagger$).

We proceed by  an exhaustive case
analysis:

\textbf{\emph{ Case 1:}}  \emph{for all $q_i$, we have $q_i \not \sqsubseteq v_1'$.}

In this case, each intersection of two given above will not be
union-free. This follows easily from Lemmas~\ref{lem:2-comp-skel} and~\ref{cl:shape}, since $v_1'$ and $q_i$ have the same root tokens and result
tokens (there is no containment mapping between them, so  there can
be no mapping between their intermediary parts).

Hence any interleaving resulting from some DAG pattern $d_i=dag(q_1
\cap v'_i)$ cannot even reduce all the other interleavings of $d_i$, so
$d$ cannot be union-free in this case, since $d=\cup_i d_i$. This case
can be thus discarded.

\textbf{\emph{ Case 2:}} \emph{at least two interleavings of $d'$, say $q_1$ and $q_2$, are such that $q_1 \sqsubseteq v_1$ and
$q_2 \sqsubseteq v_1$.}

We can thus reformulate $d$ as $d \equiv q_1 \cup  q_2 \cup (v_1' \cap q_3) \cup \dots\cup (v_1'\cap q_m)$. 
Now, each DAG pattern $v_1' \cap q_j$ is not union-free and, moreover, their interleavings cannot contain
$q_1$ or $q_2$ (since $q_1, q_2 \not \sqsubseteq q_j$ in the first place). Also, obviously, $q_2 \not \sqsubseteq q_1$ and $q_1 \not \sqsubseteq q_2$. So again $d$ can not be union-free and this case can be discarded as well.

\textbf{\emph{Case 3:}} \emph{exactly one of the interleavings of $d'$ ,
call it $q_1$, is contained in $v_1'$ ($q_1 \sqsubseteq v_1'$).}

In this case, $d$ can be reformulated as
$ d \equiv q_1 \cup(v_1' \cap q_2) \cup \dots\cup (v_1'\cap q_m)$ 
and cannot be union-free unless it is in fact equivalent to
$q_1$. This means that for all other $q_i$'s we must have $v_1' \cap
q_i \sqsubseteq q_1$. Of course, $q_1$ should be equivalent
(isomorphic modulo minimization, by Lemma~\ref{lem:equiv-iso}) to $q$,
the interleaving of $d$ for which we supposed $d \equiv q$, i.e.
$q_1\equiv p_r//t//m//t_o$. 

We continue by assuming for instance that $v_1' \cap q_2 \sqsubseteq q_1$.

Recall that $v_1'$ is of the form
$v_1'=p_r//p_1//t_o$ 
and let $q_2$ be of the
form $q_2=p_r//m_2//t_o$. Since $q_2 \not \sqsubseteq v_1'$ and they have the same root
 and result tokens, there is no mapping from $p_1$ into
$m_2$. Consequently, let $sf_1$ denote the maximal token-suffix of
$p_1$ that can map into $m_2$, and let $pr_1$ denote the remaining
part (i.e., a token-prefix). Since $pr_1$ cannot be empty,  we
can write $v_1'$ as
 $v_1'=p_r//pr_1'//sf_1//t_o$ where $pr_1'$ is an isomorphic copy of $pr_1$.

Let us now consider the interleaving $u$ of $v_1' \cap q_2$, of the
form $u=p_r//pr_1''//m_2//t_o$  where $pr_1''$ is an isomorphic copy
of $pr_1$ as well.

As we assumed that $u \sqsubseteq v_1'\cap q_2 \sqsubseteq q_1$,
there must exist a containment mapping
$\psi$ from $q_1$ to $u$.

Since $q_1 \sqsubseteq v_1'$, let $\phi$ be a containment mapping from
$v_1'$ into $q_1$.
So we have
$v_1' \stackrel{\phi}{\longrightarrow} q_1 \stackrel{\psi}{\longrightarrow} u$.

In particular, $\phi$ must map the $pr_1'//sf_1$ part of $v_1'$ in the
$t//m$ part of $q_1$. With a slight
abuse of notation, let $\phi(pr_1')$ denote the minimal token-prefix
of $t//m$ within which the image under $\phi$ of $pr_1'$ occurs. In
other words $\phi(pr_1')$ starts with the root token of $t$ and ends
with the token into which the output token of $pr_1'$ is mapped.
(Again, $\phi(pr_1')$ is well defined because all patterns are
skeletons and tokens can only map strictly inside tokens.)

We can thus write $q_1$ as $q_1=p_r//\phi(pr_1') \dots
\phi(sf_1)//t_o$. 

Next, we argue that in the containment mapping $\psi$ of $q_1$ into
$u$, we must have $\psi(\phi(pr_1'))=pr_1''$. (This follows easily
from the definition of $sf_1$.)  And this implies that $\phi(pr_1')
\equiv pr_1'' \equiv pr_1$. Hence $q_1$ and $v_1'$ start by  some
common non-empty token-prefix. Since one of them starts by $t$ and the
other by $t_1$ this means in the end that\emph{ $t \equiv t_1$, which is  a
contradiction.}

\textbf{Remark.} We can also generalize Lemma~\ref{cl:shape} as follows: the interleavings of \nf{d} are of the form $p_r//\dots//t_o$ (see Figure~\ref{fig:completeness2}).
\end{proof}
So we know for now that \apprules is complete for the
case of DAG patterns that are defined as the intersection of
skeleton queries when their root and result tokens have the same main branch.  Such an intersection is union-free iff there is
a query $v_i$ among them having an intermediary part into
which all the other intermediary parts map. If this is not the case,
the DAG is equivalent to a union of interleavings having the same
root tokens and result tokens.

We are now ready to give sum up the results so far and conclude the completeness proof for \xppes.

\begin{proof}[Summing-up]
We will show that, given $n$ (extended) skeletons $v_1, \dots, v_n$,  all having several tokens,   \apprules is complete for deciding union-freedom
for the DAG pattern $$d=\dagp{v_1 \cap \dots\cap v_n}.$$

We first rewrite $d$ by R1 steps. We obtain after this phase a DAG
pattern $d$ in which the root token of $d$ may have several main
branch nodes with outgoing //-edges. Similarly, the result token may
have several nodes with incoming //-edges. If this is not the case,
neither for the root token nor for the result token, then we know that
the algorithm \apprules is in this case complete by Lemma~\ref{lem:n-comp-skel}.

Let us assume that some run of \apprules ends without a tree. We
can easily prove that in this case the following run \apprules
would also stop without yielding a tree:
\begin{itemize}

\item first refine by rules R2, R3 and R4 the root token and the result token
 w.r.t. their outgoing/incoming //- edges,

\item then rewrite out some of the branches in parallel by applying R7.
\end{itemize}
%f this particular strategy yields a tree, it is easy to check that any other strategy would yield a tree as well. \reminder{some details}

We continue assuming that we do not obtain a tree by the above run. At
this point, $d$ is a DAG pattern as the one illustrated in
Figure~\ref{fig:completeness4}, where $t_r$ denotes the root token
(ending with node $n_r$) and $t_o$ denotes the result token (starting
with $n_o$). Rules R2, R3 and R4 no longer apply, hence each //-edge
outgoing from a node of $t_r$ that is ancestor of $n_r$ cannot be
refined into connecting it to a lower node in $t_r$. Similar for
//-edges incoming for nodes of $t_o$ that are descendants of $n_o$.

\eat{

\begin{figure}[h]
\hspace{0.5cm}
\includegraphics[trim=50mm 85mm 0mm 25mm, clip=true,
scale=0.40]{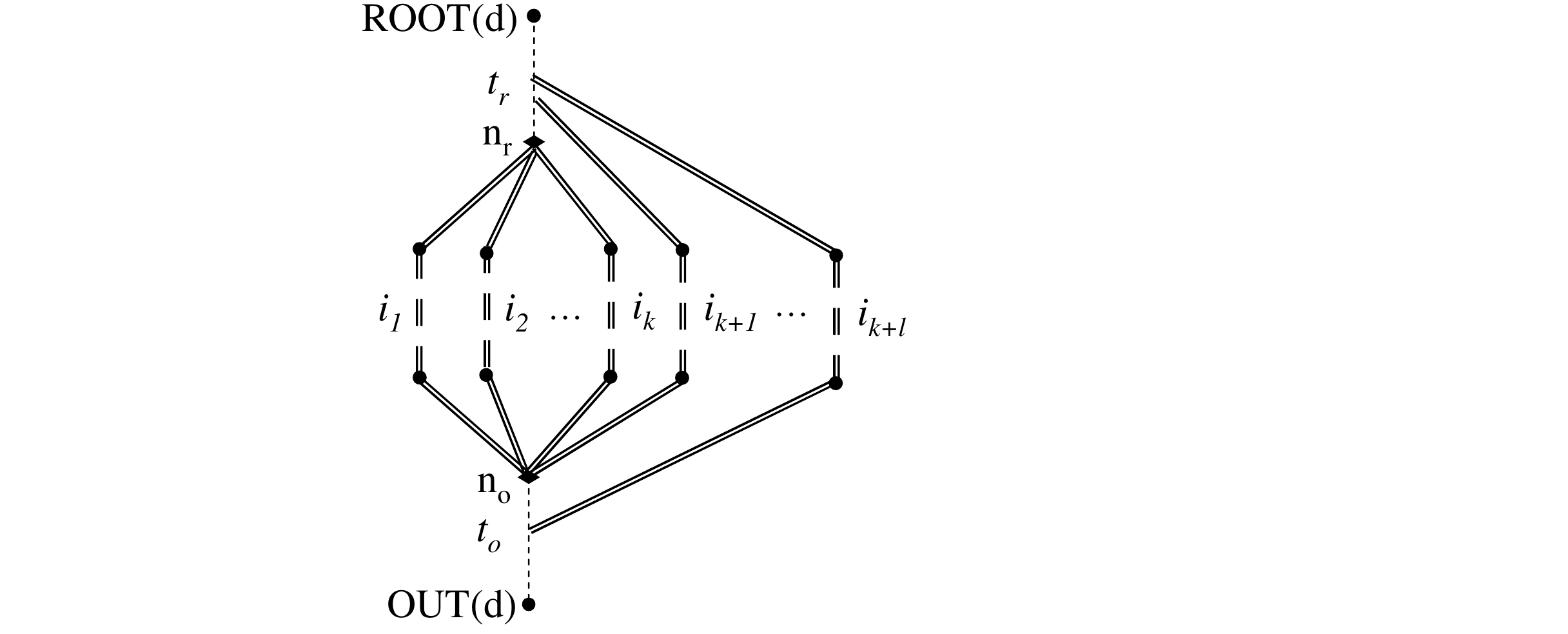}
\caption{DAG pattern $d$. \label{fig:completeness4}}
%%\vspace{-0.3cm}
\end{figure}
}

The intermediary branches $i_1, \dots, i_k$ denote those that start
from $n_r$ and end at $n_o$ (we use this notation, even if there may
be no such $i_1, \dots, i_k$ and $k=0$). The other branches in
parallel, $i_{k+1},\dots, i_{k+l}$, denote those that do not obey both
conditions.  If $l=0$, i.e. there are no such branches, we fall again
in the case handled by Lemma~\ref{lem:n-comp-skel}, for which the algorithm is complete. We continue with the assumptions that $k \geq
0$ and $l \geq 1$ as well.%, and R7 steps to not apply.
%
%Since on $d$ we can now apply only  (a) rules R6 and R7 between some of the $i_1,\dots,i_k$ branches, (a) rules R6 and R7 between some of the %$i_{k+1},\dots,i_{k+l}$ branches, or (c) rule R6 between some $i_r$ and $i_p$, for $r=1, k$ and $p=k+1, k+l$.

%For the $d$  given in Figure~\ref{fig:completeness4}, we can now prove the following:

We next prove that $d$ is not union-free.
%
%\begin{claim}
%\label{cl:claim}
%$d$ is not union-free.
%\end{claim}
%Obviously, if $d$ is not union-free the Algorithm will not output a tree (since it is sound). With the proof of this claim the proof of Lemma~\ref{lem:n-skel} is complete.
%\\

%We now prove  Claim~\ref{cl:claim}.  Let us assume that the above rewriting strategy does not go through.

%Note that it can fail at two phases: either at phase 1) or at phase 2). We consider in the following both these cases, showing that $d$ is not %union-free in either case.

We introduce some additional notation. For each $i_{j}$, $k+1 \leq j
\leq k+l$, such that $i_j$ starts above $n_r$, let $n_{j}^r$ denote
the node in $t_r$ that is sibling of the first node in $i_j$ (i.e.,
$n_{j}^r$ and the first node in $i_j$ have the same parent node, a
node in $t_r$). Note that $n_{j}^r$ is ancestor-or-self of $n_r$.  Let
$n_{j}^o$ denote the node of $t_o$ that is ``parent-sibling'' of $i_j$
(they have the same child node). $n_{j}^o$ is defined if $i_j$ ends below
$n_o$ and it is descendant-or-self of $n_o$.

For each $i_j$, by $pr_j$ we denote its maximal token-prefix that can
map in $\tp{d}{n_{j}^r/\dots/n_r}$. Similarly, for each $i_j$ by
$sf_j$ we denote the maximal token-suffix that can map in
$\tp{d}{n_o/\dots/n_{j}^o}$.

Note that $pr_j$ and $sf_j$ cannot overlap since in this case $i_j$
would have been rewritten away by R7. 

We can thus write each $i_j$ as
$i_j=pr_j//m_j//sf_j,~ \textrm{for}~ k+1\leq j \leq l+1$. 

Now, we consider a second
DAG pattern $d'$ obtained from $d$ by replacing each $i_j$ branch by  $m_j$, connected now  by //-edges to $n_r$ and $n_o$
(Figure~\ref{fig:completeness5}), instead of the parent of
$n_{j}^r$ and the child of $n_{j}^o$.

\eat{
\begin{figure}[h]
\hspace{0.5cm}
\includegraphics[trim=50mm 85mm 0mm 25mm, clip=true,
scale=0.40]{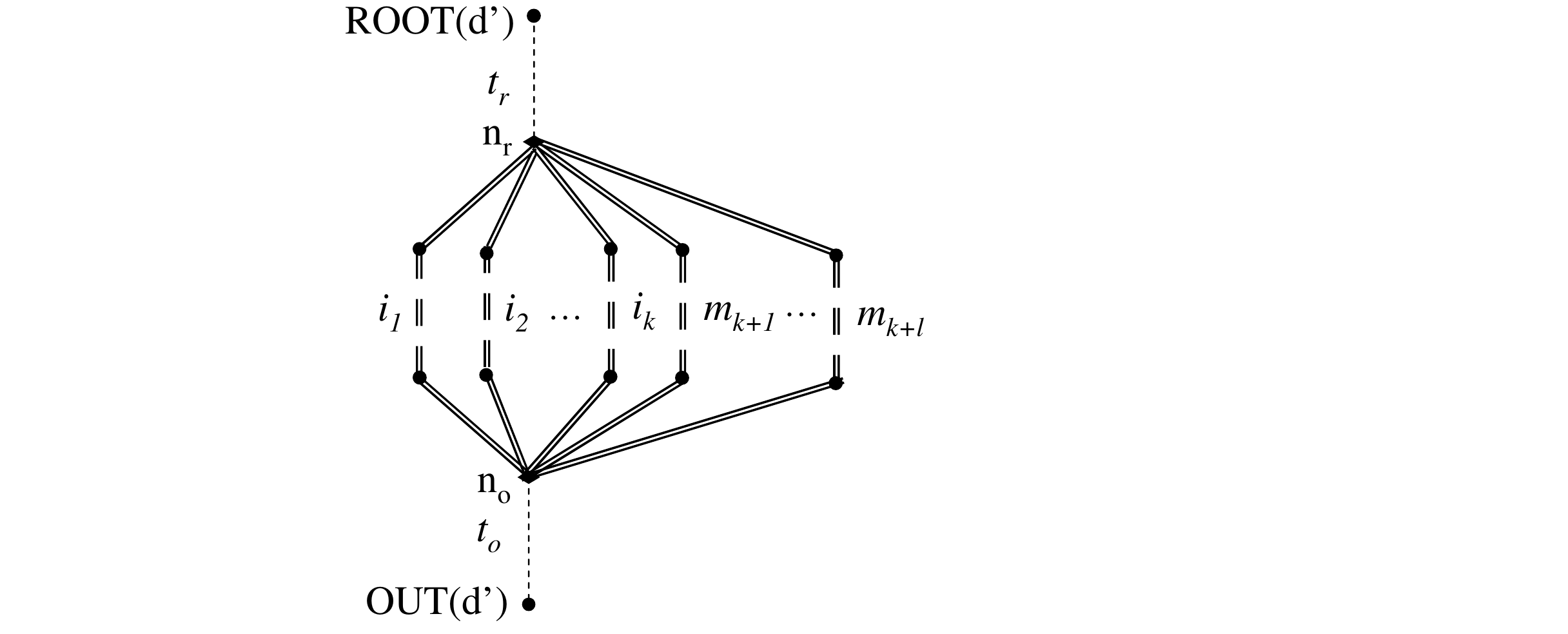}
\caption{DAG pattern $d'$. \label{fig:completeness5}}
%%\vspace{-0.3cm}
\end{figure}
}

\begin{figure}[t]
  \centering
                 \subfloat[DAG pattern $d$ for Lemma~\ref{lem:n-comp-skel}.]{\label{fig:completeness4}\includegraphics[trim=50mm 85mm 80mm 25mm, clip=true,
scale=0.40]{images/completeness4}} 
 \subfloat[DAG pattern $d$ for Theorem~\ref{th:completenessUF-es}.]{\label{fig:completeness5}\includegraphics[trim=60mm 85mm 60mm 25mm, clip=true,
scale=0.40]{images/completeness5}} 
 \caption{DAG patterns for the proof of  Theorem~\ref{th:completenessUF-es}.}
   \label{fig:contour3}
     \vspace{-3mm}
 \end{figure}

We argue now that the set of interleavings of $d'$ is included in the set of interleavings of $d$ (set inclusion). Moreover,    $d$ is union-free only if $d'$ is union-free, such that  if $d'\equiv p$, for an interleaving $p$, then $p$ is the only candidate for $d\equiv p$.    
 First, it is straightforward that all the interleavings of $d'$ are interleavings of $d$ as well. The particularity of $d'$ is that its interleavings do not modify the tokens $t_r$ and $t_o$. More precisely each interleaving will be of the form $t_r//\dots//t_o$. Moreover, by the way $d'$ was defined and given that no R2, R3 or R4 steps applied on $d$, we  argue that all other interleavings of $d$ will either
(a) be redundant, i.e. contained in those of $d'$, (b) add some predicate on $t_r$ or $t_o$ or (c) have a longer root token (resp. result token) than $t_r$ (resp. $t_o$). But this means that an interleaving $p \in \nf{d} - \nf{d'}$ cannot have a containment mapping into an interleaving of the form $t_r//\dots//t_o$. Hence it cannot be equivalent to $d$. So the only interleaving $p$ candidates for $p\equiv d$ are those of $\nf{d'}$. From this it follows that $d$ can be union-free only if $d'$ is union-free.

Note  now that by Lemma~\ref{lem:n-comp-skel} $d'$ is union-free iff there exists some $m_j$ into which all $i_1, \dots, i_k$  and all other $m_i$'s map. This is because of the assumption that among $i_1,\dots,i_k$ there is no branch $i_j$ into which all other $i_i$'s map.

We continue towards showing that $d$ is not union-free with this assumption and  let $m$ denote the branch into which all others map. Note that among  $m_{k+1},\dots, m_{k+l}$ there can be more than one ``copy'' of $m$ (i.e., equivalent to $m$). By $m_c$ we denote all these copies. Among $i_1,\dots, i_k$ there is no copy of $m$ (otherwise R7 would have triggered).

Let $d' \equiv p=t_r//m//t_o$, for $m=t//m'$.  We build next an
interleaving $w$ of $d$ s.t. $w\not \sqsubseteq p$.
%% This would conclude the proof. % since if its unique candidate interleaving does not contain all other interleavings.

 W.l.g. let  us assume that all the $m_c$ copies of $m$ are connected  in $d$ to a node that is strict ancestor of $n_r$\footnote{The remaining cases when
  \begin{itemize}
  \item all the $m_c$ copies of $m$ are connected in $d$ to a node that is strict descendant of $n_o$, or
   \item all the $m_c$ copies of $m$ but one (we cannot have more than one, otherwise R7 would have triggered leaving only one) are connected in $d$ to a node that is strict ancestor of $n_r$  and all the $m_c$ copies of $m$ but one are connected in $d$ to a node that is strict descendant of $n_o$,
        \end{itemize} can be handled similarly. %We do not have to consider more than one branch of %this kind, since otherwise R7 would trigger and in the end only one would remain.
 }. Since R2 or R4  did not apply on these copies of $m$, it means that a strict prefix of the main branch of $m$'s root token $t$ maps in a suffix of the main branch of $t_r$, when the possibly non-empty preceding token-prefix $pr_j$ is collapsed somewhere ``higher''.

Let $\psi$ denote the partial mapping from $t$ into $t_r$ that uses the maximal possible prefix of $t$ across all the copies $m_c$. Let $t$ be $t=t'/t''$, where $t'$ is this maximal prefix (not empty).

We are now ready to build $w$.

We build first the root token $t_r'$ of the $w$ interleaving as follows: let $t_r'$ denote an interleaving of $t_r$ and $t$ defined by the code $i=\mb{t_r}/\mb{t''}$, and $f_i$ defined as ``identity'' on $t_r$ and $t''$, and  $f_i(n)=\psi(n)$ for the main branch nodes of $t'$.

We build the intermediary part $p$ of the $w$ interleaving as follows: starting from $i_x \in \{i_1, \dots, i_k, m'\}$ (or simply from $i_x \in \{i_1 \dots, i_k\}$ in the case $m'$ is empty), let us interpret them as the intermediary parts of the following skeleton patterns
$ s_x = start//i_x//end$.

Let also $s$ denote the pattern $s=start//m//end$.

Let us now consider now the DAG pattern $d'=dag(\cap_x s_x)$. Since none of the $s_x$ patterns is equivalent to $s$, from Lemma~\ref{lem:n-comp-skel} we have that $d' \not \equiv s$. Moreover, since $s \sqsubseteq d'$ (because $s_x \sqsubseteq s$), we must have that $d' \not \sqsubseteq s$. In other words, there must exist an interleaving $w'$ of $d'$, of the form $start//p//end$ such that $w' \not \sqsubseteq s$. Finally, this means $p$ is such that while all the $i_1, \dots i_k, m'$ map into it, we have that $m$ does not map into it.

Finally, we define $w$ as $ w=t_r'//p//t_o$.  It is easy to check that $w$ is an interleaving of
$d$ ($d$ has a containment mapping into $w$) but $w \not \sqsubseteq
p=t_r//t//m'//t_o$. Hence $d$ is not union-free.
\eat{
In order to complete the proof, recall that we already handled the case in
which one of the skeleton queries has only one token (Lemma~\ref{l:onetokenq_nviews}).
}
\end{proof}
 \textbf{Remark.} We can draw the following conclusions from the proof of
Theorem~\ref{th:completenessUF-es}:  When \apprules is applied to DAG patterns built from multi-token views from \xppes, after R1 steps, followed eventually by R2,
R3 and R4 steps, we obtain the branches in parallel $i_1, \dots, i_k$
starting from the last node of the root token ($t_r$) and ending with
the first node of the result token ($t_o$). Other branches in parallel
may exist in $d$, but connected to other nodes of $t_r$ and
$t_o$. Then, by eventually some R7 steps, the DAG pattern must become
a tree, otherwise it is not union-free.  Under the extended skeletons restrictions, R5 and R6 are not necessary for completeness. The resulting tree is
$t_r//i_1//t_o$, where $i_1$ is one of the branches in parallel, into
which all other, $i_2, \dots, i_k$ map.

\section{Proof of Theorem~\ref{th:completenessUF-desc-akin} (Rewrite-plans for PTIME)}
\label{sec:completenessUF-desc-akin} 
%\vspace{1em}
We give in this section  the completeness proof for rewrite plans formed by akin patterns.

%\begin{proof}%({\bf Theorem~\ref{th:completenessUF-desc-akin}})

We  show that,
given $n$ akin tree patterns $v_1, \dots, v_n$, \apprules decides union-freedom
for $d=dag(v_1 \cap \dots \cap v_n)$. Let each $v_j$ be defined as $v_j=t_r^j//i_j//t_o^j$.
\\

\vspace{-3mm}
\noindent \textbf{Special case.} We start by considering the special case when the patterns have the same main branch for their result tokens as well.

By
Lemmas~\ref{lem:skel-necessary} and~\ref{lem:n-comp-skel}, we know
that $d$ is union-free only if the intermediary parts $i_j$ are such
that their skeletons map in the skeleton of one of them. Without loss
of generality, let us assume that all $s(i_j)$ map in $s(i_1)$. We
continue with this assumption.

First, the initial R1 steps coalesce the root and result tokens of the
$n$ branches, yielding a DAG pattern similar to the one illustrated in
Figure~\ref{fig:completeness3}. Then, the only rules that may be
applicable are R6 and R7.  Let us assume that \apprules stops
outputting a pattern that is not a tree.  We show that $d$ is not
union-free.

If the algorithm~\apprules stops without outputting a tree in some
run, then it will also stop without outputting a tree in the following
particular rewriting strategy\:
%% reminder{some details on why this is true}:

  \begin{itemize}
  \item we first apply R6 on the ``biggest''branches in parallel $i_j,i_k$ such that $s(i_j)\equiv s(i_k)\equiv s(i_1)$, if any. It is straightforward that R6 must apply for these branches, coalescing entirely the two branches into one branch. After this phase, there will be no other parallel branch with skeleton $s(i_1)$, besides $i_1$ itself.

   \item   Then, R6 is applied only if applicable on all the branches in parallel at once. This phase  will terminate with a refined $d$ similar to the one illustrated in Figure~\ref{fig:completeness2}, where $2 \leq k \leq n$, $p_r$ denotes the common part following the root (may have several tokens if R6 was applied) and $t_1, \dots t_k$ denote the sibling tokens on which R6 no longer applies (i.e., they are not all \emph{similar} hence they do not all have the same skeleton).

       Also, rule R7 is applied freely, and it can rewrite out some of the branches in parallel.

  After this phase, while there exists a mapping from each $s(p_i)$ into $s(p_1)$, there is no mapping from $p_i$ into $p_1$ . Note also that, by the first phase of the rewriting strategy, we cannot have the opposite mapping from $s(p_1)$ into $s(p_i)$.

  \item  Finally, rule R6 is applied only between $p_1$ on the one hand, and other branches $p_i$ on the other hand, while R7 is still applied freely.
  \end{itemize}

We obtain a DAG pattern similar to the one in Figure~\ref{fig:completeness6}. Let us assume that besides $p_1$ there are $l$ remaining branches in parallel, connected by a //-edge either to $p_r$ or to various tokens of $p_1$.

\begin{figure}[t]
%\hspace{0.5cm}
\centering
\includegraphics[trim=65mm 0mm 50mm 0mm, clip=true,
scale=0.35]{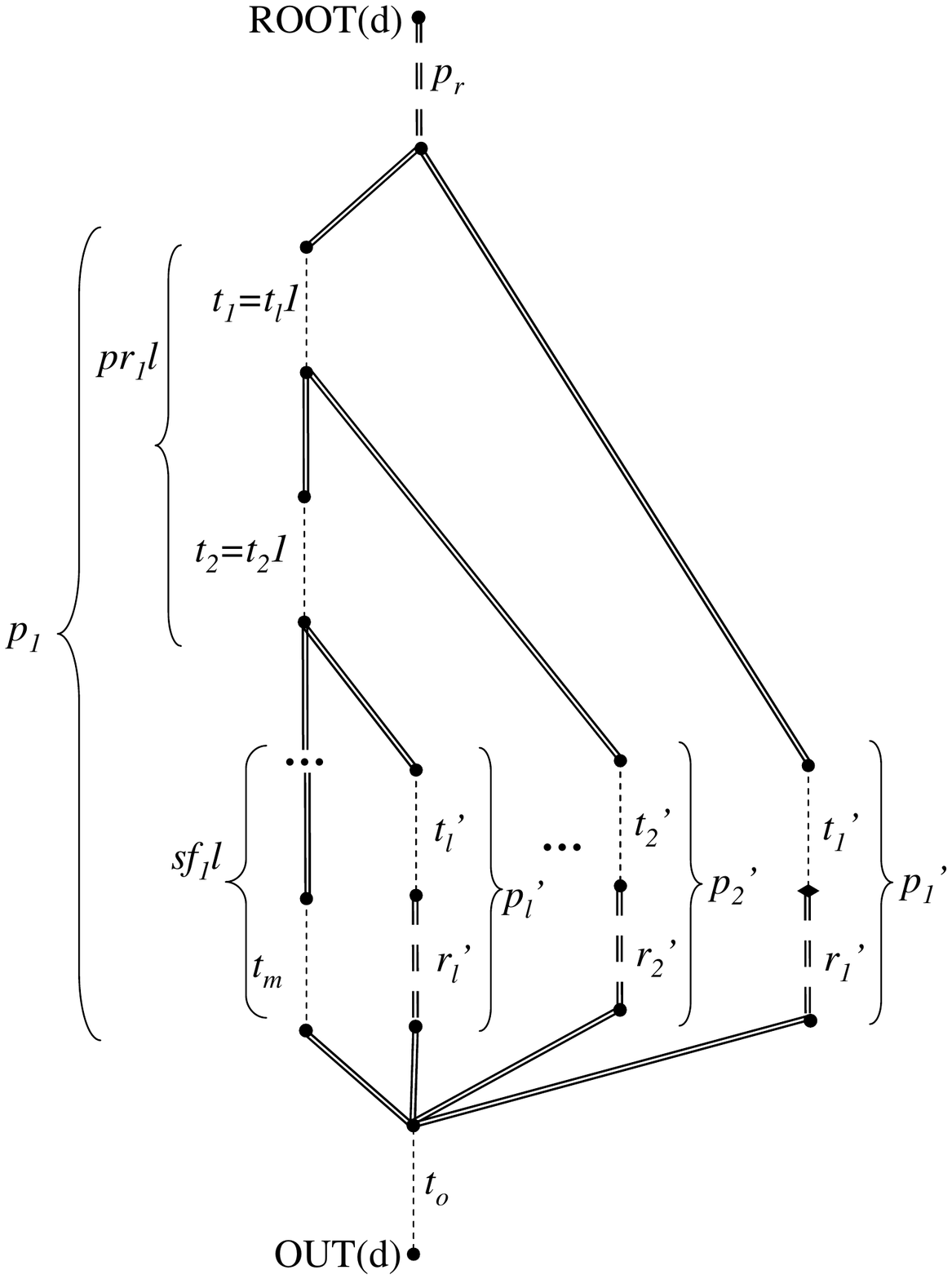}
\caption{DAG pattern $d'$. \label{fig:completeness6}}
\vspace{-0.3cm}
\end{figure}

Let $p_1=t_1//\dots//t_m$. For $i=1,l$, let
$p_i'=t_i'//r_i'$ denote now these branches in parallel with (part of)
$p_1$. For each $i=1,l$, let $t_i^1$ denote the token in $p_1$ that is
sibling of the token $t_i'$. Note that $t_i'$ and $t_i^1$ must be
dissimilar, hence will have different skeletons. Let $sf_i^1$ denote
the token-suffix of $p_1$ that is in parallel with $p_i'$ and let
$pr_i^1$ denote the rest of $p_1$ (a token-prefix).  For each $i=1,l$,
we can thus reformulate $p_1$ as $p_1=pr_i^1//sf_i^1$, where the root 
token of $sf_i^1$ is $t_i^1$. Note that for each $i$ we have that
$s(p_i')$ maps into $s(sf_i^1)$, while the opposite is not true.

It is immediate  that $d$ can be union-free, for some $c$ such that  $d \equiv c$, only if $c$ is of the form $c=p_r//m//t_o$ where $s(m)=s(p_1)$, since  this is the minimal skeleton for an interleaving.

All the candidate interleavings $c$ will be defined by the code
$i=\mb{p_r//m//t_o}$ and some function $f_i:\mbn{d} \rightarrow
i$. What distinguishes the various $c$'s is the definition of 
$f_i$  on the nodes of the branches $p_1', \dots, p_l'$
(since the other main branch nodes in $d$ have only one possible
image).  We  show next  that for any such $f_i$ and associated
interleaving $c$ we can build the $w$ witness with  $w \not
\sqsubseteq c$.

Let $f_i$ be fixed and let $c$ denote the corresponding interleaving
for code $i$ and function $f_i$. Note that we can interpret $f_i$ as a
series of rewrite steps over $d$ that collapse the pairs of nodes $(n,
f_i(n))$, for all the nodes $n$ in the $p_1', \dots, p_l'$ branches,
outputting as end result the tree pattern $c$. These steps do not
modify the skeleton of $p_1$, hence can only bring some new predicates
starting by //-edge.

Next, we describe how the interleaving $w \not \sqsubseteq c$ is built,
from the current pattern $d$ of Figure~\ref{fig:completeness6}.

Let $n_c \in \mb{c}$ denote the lowest main branch node in $c$'s $m$ part which has a subtree predicate $st$ that is not present (in other words, cannot be mapped) at the associated node $n_1$ in the $p_1$ part of $d$. $st$ must start with a //-edge and must come from a node of some (maybe several) branches $p_i'$. (We know that such a node $n_c$ must exist, otherwise the $p_i'$ branches would fully map in the corresponding branch in parallel $sf_i^1$ and rule R7 would have applied).

Without loss of generality,  let $n_{i_1}',\dots, n_{i_s}'$, for $\{i_1,\dots, i_{s}\} \subseteq \{1,\dots, l\}$, denote the nodes from the branches $p_{i_1}', \dots, p_{i_s}'$ that are the ``source'' of $st$\footnote{They have a predicate $st'$ into which $st$ maps.}. So we have  $f_i(n_1)=f_i(n_{i_1}')= \dots =f_i(n_{i_s}')=n_c$ and we can say that $n_c$ is the result of coalescing $n_1$ with $n_{i_1}', \dots n_{i_s}'$.

Now, we can see the left branch $p_1$ as being divided into two parts, the one down to the token of $n_1$ (that token included), denoted $p_{11}$, and the rest, denoted $p_{12}$. So we can write $p_1$ as $p_1=p_{11}//p_{12}$. 

Similarly, for each $x \in \{i_1,\dots, i_{s}\}$ we can see each main branch $pr_{x}^{1}//p_{x}'$ as being divided into two parts, the one down to the token of $n_{x}'$ (that token included), denoted $p_{x1}'$, and the rest, denoted  $p_{x2}'$. So we can write each main branch $pr_{x}^{1}//p_{x}'$ of $d$ as $pr_{x}^{1}//p_{x}'=p_{x1}'//p_{x2}'$.

%Let us ignore for now the branches that do not have indices in $\{i_1,\dots, i_{s}\}$.

Note that by the way $n_c$ was chosen (as the lowest node) we can
conclude that by $f_i$ (on the main branch nodes) we can fully map
$\tp{d}{p_{x2}'}$ into $\sub{d}{n_1}$, for all $x$ (i.e., there are no
other added predicates below $n_1$'s level). It is also easy to see
that while $s(p_{x1}')$ maps in $s(p_{11})$ (by $f_i$), the
opposite is not true, otherwise R6 steps would have applied up to this
point.

We are now ready to construct $w$. First, we obtain a part $p$ of the
$w$ interleaving as follows: starting from the set of skeleton queries
$s(p_{x1}')$, for all $x \in \{i_1,\dots, i_{s}\}$, let us interpret
them as the intermediary parts of the following skeleton patterns $ s_x = start//s(p_{x1}')//end$. 

Let also $s$ denote the skeleton pattern $s=start//s(p_{11})//end$.

Let us now consider the DAG pattern $d'=dag(\cap_x s_x)$. Since none
of the skeleton patterns $s_x$ is equivalent to $s$, from
Lemma~\ref{lem:n-comp-skel} we have that $d' \not \equiv s$. Moreover,
since $s \sqsubseteq d'$ (because $s \sqsubseteq s_x$, by the way $c$
was defined), we must have that $d' \not \sqsubseteq s$. In other
words, there must exist an interleaving $w'$ of $d'$, of the form
$start//p//end$ such that $w' \not \sqsubseteq s$. Finally, this means
$p$ is such that while all the $s(p_{x1}')$ map into it, we have that
$s(p_{11})$ does not map into it. We will use this property. For each
$p_{x1}'$, let $f_{x1}$ denote a mapping from $s(p_{x1}')$ into $p$.

Next, we obtain a second part of $w$ as follows. Let $pr_{11}$ denote
the maximal token-suffix of $p_{11}$ such that $s(p_{11})$ can map in
$p$, and let $sf_{11}$ denote the remaining part. $sf_{11}$ cannot be
empty, so it is formed by at least the output token of $p_{11}$, the one
with node $n_1$.  So we can see $p_{11}$ as
 $p_{11}=pr_{11}//sf_{11}$.  
 
 Let $f_p$ denote a partial mapping from
$s(p_{11})$ into $p$ that exhibits $sf_{11}$.

We will define $w$ by a code $i'$ and function $f_i'$ as follows:

\begin{itemize}
\item $i' = \lambda(p_r//p//sf_{11}//p_{12}//t_o)$,
\item $f_i'$ maps nodes of \mbn{d} into $i'$ positions as follows:
\begin{itemize}
\item $f_i'$ is ``identity'' for  the main branch nodes of $p_r$, $t_o$, for the $sf_{11}$ part of the $p_{11}$ prefix of $p_1$ and for the $p_{12}$ suffix of $p_1$,
 \item for the remaining main branch nodes $n$ in $p_{11}$ (i.e., those of $pr_{11}$), $f_i'(n) = f_p(n)$,
    \item for the main branch nodes $n$ of the $p_{x1}'$ prefix of the $pr_x^1//p_x'$ branch in $d$, for  $x \in \{i_1,\dots, i_{s}\}$, $f_i'(n)=f_{x1}(n)$
 \item for the remaining nodes $n$ in the $pr_x^1//p_x'$ branches (i.e. those in $p_{x2}'$), $f_i'(n)=f_i'(f_i(n))$.
 \item finally, for all the main branch nodes of the remaining branches $p_{y}'$, for $y \not \in \{i_1,\dots, i_{s}\}$, $f_i'(n)=f_i'(f_i(n))$.

     (they go where their images under $f_i$ go.)
\end{itemize}
\end{itemize}
We now argue that $w$ is an interleaving of $d$ and $w \not \sqsubseteq c$. First, it is easy to check that $w$ is an interleaving for $d$.
Recall that $c$ is s.t.  $s(c) = s(p_r//p_1//t_o)=s(p_r//pr_{11}//sf_{11}//p_{12}//t_o)$.  Second, it is also easy to check that $c$ can have a containment
mapping in $w$ iff its $sf_{11}$ part maps in the $sf_{11}$ of
$w$. But this is not possible because the $st$ subtree predicate is
not present on the $f_i'(n_1)$ node of $w$ (which is found somewhere
in the output token of the $sf_{11}$ part).
\\

\vspace{-3mm}
\noindent \textbf{General case.} We now consider the general case, when the result tokens do not necessarily have the same main branch. After the possible rewrite R1(i) steps on the root tokens, and after the possible rewrite steps of R1(ii), R2(ii), R3(ii) and  R4(ii) on the result tokens, we may now obtain a DAG pattern in which the branches in parallel may not be ``connected'' to $t_o$ at its highest node ($n_o$), but at some other node that is strict descendant of $n_o$. If this is not the case, then we are back to the special case discussed previously.

Otherwise, let us now consider the DAG pattern $d'$ obtained from $d$ by connecting the endpoints of the branches in parallel at $n_o$. We can easily see that the interleavings of $d'$  are all among those of $d$ and moreover, $d$ is union-free only if $d'$ is union-free, with $d' \equiv d\equiv p$, for some $p \in interleave(d')$. This is because the interleavings of $d$ that are not interleavings of $d'$ as well are those that add some predicates on $t_o$ that are not present in all the interleavings.

By Lemmas~\ref{lem:skel-necessary} and~\ref{lem:n-comp-skel}, we know
that $d$ is union-free only if the intermediary parts $i_j$ are such
that their skeletons map in the skeleton of one of them, which in addition, in the current $d$ pattern, must start at $n_r$ and end at $n_o$. Without loss of generality, let us assume that this is $i_1$ (note that the $i_1$ branch will not be affected by the transformation from $d$ to $d'$).

From the special case, we know under what conditions $d'$  is union-free, and it is immediate that, when they hold, the interleaving $p$ is obtained from $i_1$ possibly by adding some predicates of the form $[.//\dots]$ to some of its main branch nodes. Importantly, each such predicate is added on the highest possible main branch node of $i_1$.

Finally, it is now easy to check that an interleaving $p$ of $d'$ obtained in this way will always have a containment mapping in any interleaving $p'$ of $d$: everything except the added predicates will map (by identity), while the added predicates (of the form $[.//\dots]$) will map in their respective occurrence in $p'$ (by necessity, found at a lower main branch node then in the one in $p$).

This ends the completeness proof of \apprules over unfoldings of rewrite plans that intersect only  akin views from \xppdesc.

\end{document}